\documentclass{easychair}
\usepackage{amsmath,amssymb,amsthm}
\usepackage[bb=boondox]{mathalfa}
\usepackage{xcolor}
\usepackage{turnstile}
\usepackage[all]{xy}
\usepackage{xspace}
\usepackage{geometry}
\usepackage{hyperref}
\usepackage[noabbrev,capitalize]{cleveref}
\usepackage[normalem]{ulem}
\usepackage{thmtools}

\newdimen\proofrulebreadth \proofrulebreadth=.05em
\newdimen\proofdotseparation \proofdotseparation=1.25ex
\newdimen\proofrulebaseline \proofrulebaseline=2ex
\newcount\proofdotnumber \proofdotnumber=3
\let\then\relax
\def\hfi{\hskip0pt plus.0001fil}
\mathchardef\squigto="3A3B
\newif\ifinsideprooftree\insideprooftreefalse
\newif\ifonleftofproofrule\onleftofproofrulefalse
\newif\ifproofdots\proofdotsfalse
\newif\ifdoubleproof\doubleprooffalse
\let\wereinproofbit\relax
\newdimen\shortenproofleft
\newdimen\shortenproofright
\newdimen\proofbelowshift
\newbox\proofabove
\newbox\proofbelow
\newbox\proofrulename
\def\shiftproofbelow{\let\next\relax\afterassignment\setshiftproofbelow\dimen0 }
\def\shiftproofbelowneg{\def\next{\multiply\dimen0 by-1 }%
\afterassignment\setshiftproofbelow\dimen0 }
\def\setshiftproofbelow{\next\proofbelowshift=\dimen0 }
\def\setproofrulebreadth{\proofrulebreadth}

\def\prooftree{% NESTED ZERO (\ifonleftofproofrule)
\ifnum  \lastpenalty=1
\then   \unpenalty
\else   \onleftofproofrulefalse
\fi
\ifonleftofproofrule
\else   \ifinsideprooftree
        \then   \hskip.5em plus1fil
        \fi
\fi
\bgroup% NESTED ONE (\proofbelow, \proofrulename, \proofabove,
\setbox\proofbelow=\hbox{}\setbox\proofrulename=\hbox{}%
\let\justifies\proofover\let\leadsto\proofoverdots\let\Justifies\proofoverdbl
\let\using\proofusing\let\[\prooftree
\ifinsideprooftree\let\]\endprooftree\fi
\proofdotsfalse\doubleprooffalse
\let\thickness\setproofrulebreadth
\let\shiftright\shiftproofbelow \let\shift\shiftproofbelow
\let\shiftleft\shiftproofbelowneg
\let\ifwasinsideprooftree\ifinsideprooftree
\insideprooftreetrue
\setbox\proofabove=\hbox\bgroup$\displaystyle % NESTED TWO
\let\wereinproofbit\prooftree
\shortenproofleft=0pt \shortenproofright=0pt \proofbelowshift=0pt
\onleftofproofruletrue\penalty1
}

\def\eproofbit{% NESTED TWO
\ifx    \wereinproofbit\prooftree
\then   \ifcase \lastpenalty
        \then   \shortenproofright=0pt  % 0: some other object, no indentation
        \or     \unpenalty\hfil         % 1: empty hypotheses, just glue
        \or     \unpenalty\unskip       % 2: just had a tree, remove glue
        \else   \shortenproofright=0pt  % eh?
        \fi
\fi
\global\dimen0=\shortenproofleft
\global\dimen1=\shortenproofright
\global\dimen2=\proofrulebreadth
\global\dimen3=\proofbelowshift
\global\dimen4=\proofdotseparation
\global\count255=\proofdotnumber
$\egroup  % NESTED ONE
\shortenproofleft=\dimen0
\shortenproofright=\dimen1
\proofrulebreadth=\dimen2
\proofbelowshift=\dimen3
\proofdotseparation=\dimen4
\proofdotnumber=\count255
}

\def\proofover{% NESTED TWO
\eproofbit % NESTED ONE
\setbox\proofbelow=\hbox\bgroup % NESTED TWO
\let\wereinproofbit\proofover
$\displaystyle
}%
\def\proofoverdbl{% NESTED TWO
\eproofbit % NESTED ONE
\doubleprooftrue
\setbox\proofbelow=\hbox\bgroup % NESTED TWO
\let\wereinproofbit\proofoverdbl
$\displaystyle
}%
\def\proofoverdots{% NESTED TWO
\eproofbit % NESTED ONE
\proofdotstrue
\setbox\proofbelow=\hbox\bgroup % NESTED TWO
\let\wereinproofbit\proofoverdots
$\displaystyle
}%
\def\proofusing{% NESTED TWO
\eproofbit % NESTED ONE
\setbox\proofrulename=\hbox\bgroup % NESTED TWO
\let\wereinproofbit\proofusing
\kern0.3em$
}

\def\endprooftree{% NESTED TWO
\eproofbit % NESTED ONE
  \dimen5 =0pt% spread of hypotheses
\dimen0=\wd\proofabove \advance\dimen0-\shortenproofleft
\advance\dimen0-\shortenproofright
\dimen1=.5\dimen0 \advance\dimen1-.5\wd\proofbelow
\dimen4=\dimen1
\advance\dimen1\proofbelowshift \advance\dimen4-\proofbelowshift
\ifdim  \dimen1<0pt
\then   \advance\shortenproofleft\dimen1
        \advance\dimen0-\dimen1
        \dimen1=0pt
        \ifdim  \shortenproofleft<0pt
        \then   \setbox\proofabove=\hbox{%
                        \kern-\shortenproofleft\unhbox\proofabove}%
                \shortenproofleft=0pt
        \fi
\fi
\ifdim  \dimen4<0pt
\then   \advance\shortenproofright\dimen4
        \advance\dimen0-\dimen4
        \dimen4=0pt
\fi
\ifdim  \shortenproofright<\wd\proofrulename
\then   \shortenproofright=\wd\proofrulename
\fi
\dimen2=\shortenproofleft \advance\dimen2 by\dimen1
\dimen3=\shortenproofright\advance\dimen3 by\dimen4
\ifproofdots
\then
        \dimen6=\shortenproofleft \advance\dimen6 .5\dimen0
        \setbox1=\vbox to\proofdotseparation{\vss\hbox{$\cdot$}\vss}%
        \setbox0=\hbox{%
                \advance\dimen6-.5\wd1
                \kern\dimen6
                $\vcenter to\proofdotnumber\proofdotseparation
                        {\leaders\box1\vfill}$%
                \unhbox\proofrulename}%
\else   \dimen6=\fontdimen22\the\textfont2 % height of maths axis
        \dimen7=\dimen6
        \advance\dimen6by.5\proofrulebreadth
        \advance\dimen7by-.5\proofrulebreadth
        \setbox0=\hbox{%
                \kern\shortenproofleft
                \ifdoubleproof
                \then   \hbox to\dimen0{%
                        $\mathsurround0pt\mathord=\mkern-6mu%
                        \cleaders\hbox{$\mkern-2mu=\mkern-2mu$}\hfill
                        \mkern-6mu\mathord=$}%
                \else   \vrule height\dimen6 depth-\dimen7 width\dimen0
                \fi
                \unhbox\proofrulename}%
        \ht0=\dimen6 \dp0=-\dimen7
\fi
\let\doll\relax
\ifwasinsideprooftree
\then   \let\VBOX\vbox
\else   \ifmmode\else$\let\doll=$\fi
        \let\VBOX\vcenter
\fi
\VBOX   {\baselineskip\proofrulebaseline \lineskip.2ex
        \expandafter\lineskiplimit\ifproofdots0ex\else-0.6ex\fi
        \hbox   spread\dimen5   {\hfi\unhbox\proofabove\hfi}%
        \hbox{\box0}%
        \hbox   {\kern\dimen2 \box\proofbelow}}\doll%
\global\dimen2=\dimen2
\global\dimen3=\dimen3
\egroup % NESTED ZERO
\ifonleftofproofrule
\then   \shortenproofleft=\dimen2
\fi
\shortenproofright=\dimen3
\onleftofproofrulefalse
\ifinsideprooftree
\then   \hskip.5em plus 1fil \penalty2
\fi
}

\newcommand{\hinput}[1]{}

\newcommand{\CBN}{CBN\xspace}
\newcommand{\CBV}{CBV\xspace}
\newcommand{\CBPV}{CBPV\xspace}
\newcommand{\PCF}{PCF\xspace}
\newcommand{\PCFh}{PCF{\scriptsize\textsf{H}}\xspace}

\newcommand{\typesystem}{\mathcal{H}}

\newcommand{\defn}[1]{\textbf{#1}}
\newcommand{\ih}{IH\xspace}
\newcommand{\Eg}{{\em E.g.}\xspace}
\newcommand{\eg}{{\em e.g.}\xspace}
\newcommand{\ie}{{\em i.e.}\xspace}
\newcommand{\eqdef}{\overset{\textrm{def}}{=}}
\newcommand{\defeq}{:=}
\newcommand{\eqgram}{\mathrel{::=}}
\newcommand{\set}[1]{\{#1\}}
\newcommand{\fv}[1]{\mathsf{fv}(#1)}
\newcommand{\id}{\mathtt{id}}

\renewcommand{\emptyset}{\varnothing}
\newcommand{\length}[1]{\#(#1)}

\renewcommand{\theenumi}{\arabic{enumi}}
\renewcommand{\theenumii}{\arabic{enumii}}
\renewcommand{\theenumiii}{\arabic{enumiii}}

\makeatletter
\renewcommand\p@enumii{\theenumi.}
\renewcommand\p@enumiii{\theenumi.\theenumii.}
\renewcommand\p@enumiv{\theenumi.\theenumii.\theenumiii.}
\makeatother

\colorlet{darkgreen}{green!60!black}
\colorlet{DARKGREEN}{green!60!black}

\theoremstyle{break}
\newtheorem{dummythm}{dummythm}
\newtheorem{lemma}[dummythm]{Lemma}

\newtheorem{theorem}[dummythm]{Theorem}

\theoremstyle{definition}

\theoremstyle{remark}
\newtheorem{remark}[dummythm]{Remark}

\newcommand{\HS}{\hspace{.5cm}}

\newcommand{\NF}[1]{\mathsf{NF}_{#1}}

\newcommand{\fun}{f}

\newcommand{\var}{x}
\newcommand{\vartwo}{y}
\newcommand{\varthree}{z}

\newcommand{\val}{\mathbf{v}}
\newcommand{\valtwo}{\mathbf{w}}
\newcommand{\valnat}{\mathbf{k}}
\newcommand{\valnattwo}{\mathbf{l}}

\newcommand{\tm}{t}
\newcommand{\tmtwo}{s}
\newcommand{\tmthree}{u}
\newcommand{\tmfour}{r}

\newcommand{\tmsix}{q}

\newcommand{\lam}[2]{\lambda#1.\,#2}
\newcommand{\zero}{{\bf 0}}
\newcommand{\succsym}{\mathbf{S}}
\renewcommand{\succ}[1]{\succsym(#1)}
\newcommand{\bound}[2]{#1.\,#2}
\newcommand{\ifz}[4]{\mathsf{if}(#1, #2, \bound{#3}{#4})}
\newcommand{\fix}[2]{\mathsf{fix}\,\bound{#1}{#2}}

\newcommand{\sub}[2]{\{#1:=#2\}}

\newcommand{\tov}[1]{\to_{#1}}
\newcommand{\tovsil}{\to}

\newcommand{\NFv}[1]{\mathsf{NF}_{#1}}
\newcommand{\rulename}{\rho}

\newcommand{\arVd}[1]{\ar[d]_-{#1}}
\newcommand{\arVr}[1]{\ar[r]_-{#1}}
\newcommand{\arsdVd}[1]{\ar@{.>}[d]_-{#1}}
\newcommand{\arsdVr}[1]{\ar@{.>}[r]_-{#1}}

\newcommand{\deriv}{\pi}
\newcommand{\derivtwo}{\zeta}

\newcommand{\derivof}[2]{#1 \rhd #2}

\newcommand{\typ}{\tau}

\newcommand{\typtwo}{\sigma}

\newcommand{\mtyp}{\mathcal{T}}
\newcommand{\mtyptwo}{\mathcal{S}}
\newcommand{\mtypthree}{\mathcal{U}}
\newcommand{\mtypfour}{\mathcal{R}}
\newcommand{\zeroTyp}{\mathbb{0}}
\newcommand{\succTyp}[1]{\mathbb{S}(#1)}

\newcommand{\ftyp}{\mathcal{F}}

\newcommand{\abssym}{\mathsf{abs}}
\newcommand{\natsym}{\mathsf{nat}}
\newcommand{\stucksym}{\mathsf{stuck}}
\newcommand{\nature}{\nu}

\newcommand{\enature}{\varepsilon}
\newcommand{\enaturetwo}{\varepsilon'}
\newcommand{\typnat}{\mathbb{N}}
\newcommand{\typtwonat}{\mathbb{M}}
\newcommand{\typabs}{\mathbb{A}}

\newcommand{\mtypnat}{\mathcal{N}}
\newcommand{\mtyptwonat}{\mathcal{M}}
\newcommand{\mtypthreenat}{\mathcal{P}}
\newcommand{\mtypabs}{\mathcal{A}}
\newcommand{\mtyptwoabs}{\mathcal{B}}

\newcommand{\optmarker}{?}
\newcommand{\optmtyp}{\mtyp^{\optmarker}}
\newcommand{\optmtyptwo}{\mtyptwo^{\optmarker}}
\newcommand{\optmtypthree}{\mtypthree^{\optmarker}}

\newcommand{\optmtypnat}{\mtypnat^{\optmarker}}

\newcommand{\mleq}{\lhd}
\newcommand{\none}{\bot}

\newcommand{\mcount}{\mathfrak{m}}
\newcommand{\mcounttwo}{\mathfrak{n}}

\newcommand{\mset}[1]{[#1]}
\newcommand{\emset}{\mset{\,}}
\newcommand{\msetnu}[2]{\mset{#2}^{#1}}
\newcommand{\msetabs}[1]{\msetnu{\abssym}{#1}}
\newcommand{\msetnat}[1]{\msetnu{\natsym}{#1}}
\newcommand{\emsetnu}[1]{\emset^{#1}}
\newcommand{\iI}{{i \in I}}
\newcommand{\jJ}{{j \in J}}
\newcommand{\kK}{{k \in K}}

\newcommand{\dom}[1]{\mathsf{dom}(#1)}
\newcommand{\emptyctx}{\emptyset}

\newcommand{\tctx}{\Gamma}
\newcommand{\tctxtwo}{\Delta}

\newcommand{\fset}[1]{\{\!\!\{#1\}\!\!\}}
\newcommand{\efset}{\fset{\,}}

\newcommand{\fctx}{\Phi}
\newcommand{\fctxtwo}{\Psi}

\newcommand{\judgv}[5][]{#2;#3\vdash^{#1}#4:#5}

\newcommand{\emptyPremise}{\vphantom{{}^@}}
\newcommand{\indrulename}[1]{\texttt{#1}}
\newcommand{\indrule}[3]{
\ensuremath{
\begin{array}{c}
  \prooftree #2
    \justifies #3
    \thickness=0.05em
    \using \indrulename{#1}
  \endprooftree
\end{array}}}

\newcommand{\bBeta}{\indrulename{B}}
\newcommand{\bIfZero}{\indrulename{I0}}
\newcommand{\bIfSucc}{\indrulename{IS}}
\newcommand{\bFix}{\indrulename{F}}

\newcommand{\vruleToBeta}{\indrulename{r-beta}}
\newcommand{\vruleToIfZero}{\indrulename{r-if0}}
\newcommand{\vruleToIfSucc}{\indrulename{r-ifS}}
\newcommand{\vruleToFix}{\indrulename{r-fix}}
\newcommand{\vruleToCongAppL}{\indrulename{r-appL}}
\newcommand{\vruleToCongAppR}{\indrulename{r-appR}}
\newcommand{\vruleToCongSucc}{\indrulename{r-succ}}
\newcommand{\vruleToCongIf}{\indrulename{r-if}}

\newcommand{\ruleTypAbsV}{\indrulename{t-abs$_\mathsf{v}$}}

\newcommand{\vruleNFAbs}{\indrulename{nf-abs}}
\newcommand{\vruleNFApp}{\indrulename{nf-app}}
\newcommand{\vruleNFZero}{\indrulename{nf-zero}}
\newcommand{\vruleNFSuccNat}{\indrulename{nf-succ-nat}}
\newcommand{\vruleNFSuccErr}{\indrulename{nf-succ-stuck}} % Cambiar macro por \vruleNFSuccStuck
\newcommand{\vruleNFIf}{\indrulename{nf-if}}

\newcommand{\vruleTypVarOne}{\indrulename{t-var$_1$}}
\newcommand{\vruleTypVarTwo}{\indrulename{t-var$_2$}}
\newcommand{\vruleTypAbs}{\indrulename{t-abs}}
\newcommand{\vruleTypApp}{\indrulename{t-app}}
\newcommand{\vruleTypZero}{\indrulename{t-zero}}
\newcommand{\vruleTypSucc}{\indrulename{t-succ}}
\newcommand{\vruleTypIfZero}{\indrulename{t-ifZero}}
\newcommand{\vruleTypIfSucc}{\indrulename{t-ifSucc}}
\newcommand{\vruleTypFix}{\indrulename{t-fix}}

\title{Hybrid Intersection Types for \PCF (Extended Version)}

\author{
  Pablo Barenbaum\inst{1}\footnotemark{}
    \footnotetext{
      \textit{
        Partially funded by project grants PUNQ 418/22 and PICT-2023-602.
      }
    }
\and
  Delia Kesner\inst{2}
\and
  Mariana Milicich\inst{2}\footnotemark{}
    \footnotetext{\includegraphics[height=7.0mm]{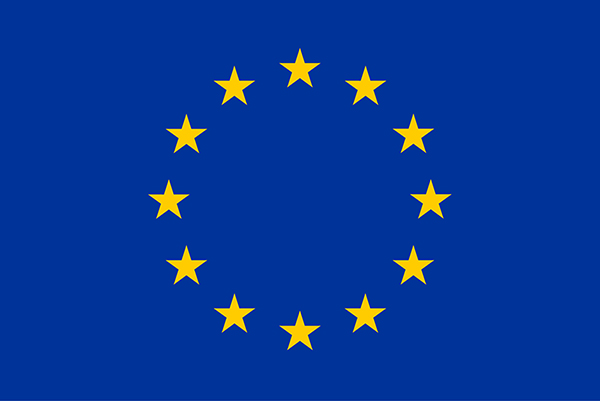}
      \textit{
        This project has received funding
        from the European Union’s Horizon 2020 research and innovation
        programme under the Marie Skłodowska-Curie grant agreement No 945332.
      }
    }
}

\institute{
  Universidad Nacional de Quilmes (CONICET), and Instituto de Ciencias de la Computación, UBA\\
  Argentina
\and
   Université Paris Cité, CNRS, IRIF \\
   France
}

\authorrunning{Barenbaum, Kesner, and Milicich}

\titlerunning{Hybrid Intersection Types for \PCF (Extended Version)}

\begin{document}

\maketitle

% \tableofcontents

\begin{abstract}
Intersection type systems have been \emph{independently} applied to different 
evaluation strategies, such as call-by-name (\CBN) and call-by-value (\CBV). 
These type systems have been then generalized to different 
subsuming paradigms being able, in particular, to \emph{encode} \CBN and \CBV 
in a \emph{unique} unifying framework.
However, there are no intersection type systems that explicitly enable
\CBN and \CBV to cohabit together, without making use of an encoding
into a common target framework.

This work proposes an intersection type system  
for a specific notion of evaluation for PCF, called \PCFh.
Evaluation in \PCFh actually has a \emph{hybrid} nature, in the sense that
\CBN and \CBV operational behaviors cohabit together. 
Indeed, \PCFh combines a \CBV-like behavior for 
function application with a \CBN-like behavior for recursion.
This hybrid nature is reflected in the type system, which turns out to be
\emph{sound} and \emph{complete} with respect to \PCFh:
not only typability implies normalization, but also the converse holds.
Moreover, the type system is \emph{quantitative}, in the sense
that the size of typing derivations provides upper bounds for the
length of the reduction sequences to normal form.
This first type system is then refined to a \emph{tight} one,
offering \emph{exact} information regarding the length of
normalization sequences. 
This is the first time that 
a sound and complete quantitative type system has been designed for
a hybrid computational model.
\end{abstract}

\section{Introduction}
\label{sec:introduction}
Evaluation strategies govern how computation proceeds in programming
languages. Two prominent strategies are \emph{call-by-name}~(\CBN)
and \emph{call-by-value}~(\CBV)~\cite{Plotkin75}. In \CBN, arguments
are not evaluated before being consumed by functions, while in \CBV,
arguments must be fully evaluated to \emph{values} before function
application can proceed.

These two ways to perform evaluation have 
their own advantages and drawbacks.
\CBV evaluation can sometimes be less costly than \CBN,
since \CBV evaluates arguments only once, 
while \CBN may have to evaluate many copies of a single argument.
Conversely, sometimes \CBN evaluation can be less costly than \CBV,
as \CBN does not evaluate an \emph{unused} argument
while \CBV always evaluates arguments, even if their value is not needed.
For instance, if the function $g$
is constantly $0$ and $\Omega$ is a looping program,
then $g(\Omega)$ produces $0$ as a result in \CBN
while \CBV evaluation does not terminate.

The main motivation of this work is to enhance our
understanding of the \emph{quantitative semantics}
of programming language constructs such as
inductive data types and recursion.
The quantitative semantics
of minimalistic languages such as the $\lambda$-calculus
are relatively well-understood in both
the \CBN and \CBV settings (see~\cite{deCarvalho07,Ehrhard12}),
but inductive data types and recursion cannot
be expressed directly, and must instead rely on some sort
of \emph{encoding}.
In this work, we study \PCFh, an extension of
the $\lambda$-calculus that incorporates natural numbers,
conditional expressions, and recursion through a fixed-point
operator without resorting to an encoding.

\paragraph{Non-idempotent intersection types.}
Intersection types (IT), pioneered by Coppo and Dezani~\cite{CoppoD80},
extend simple types with a new \emph{intersection} type constructor~($\cap$),
in such a way that a term is assigned the type $\typ \cap \typtwo$
if and only if it is assigned both types $\typ$ and $\typtwo$.
Originally, intersection is commutative, associative
and, in particular, \emph{idempotent}, \ie $\typ$ and $\typ\cap\typ$ are equivalent.
An intersection $\typ_1\cap\hdots\cap\typ_n$ can thus be specified by a \emph{set}
$\set{\typ_1,\hdots,\typ_n}$.
IT were motivated by semantical concerns. A key property
is that typability is preserved by \emph{conversion}
(reduction \emph{and} expansion), in contrast with simple types,
in which typability is preserved by reduction but not expansion.
This technical property allows to study normalization from a type theoretical
point of view. 
In intersection types, a term is typable if and only if it is normalizing,
while in simple types, typable terms are normalizing but the converse does not hold.

A more recent line of research is that of \emph{non-idempotent}
IT~\cite{Gardner94,deCarvalho07}, in which $\typ\cap\typ$ is not equivalent to $\typ$.
Here, an intersection $\typ_1\cap\hdots\cap\typ_n$
may be specified by a \emph{multiset} $\mset{\typ_1,\hdots,\typ_n}$.
These systems draw inspiration from Linear Logic (LL)~\cite{Girard87},
as in fact, non-idempotent intersection corresponds to
LL's multiplicative conjunction ($\otimes$)~\cite{deCarvalho07}.
Just like in the idempotent case,
non-idempotent IT characterize normalization by means of typability~\cite{deCarvalho07}.
Moreover, non-idempotent types have some advantages over idempotent ones;
one key benefit  is their ability to capture \emph{quantitative}
information about the dynamic behavior of programs.
For example, in non-idempotent IT one can not only prove that
a term $\tm$ is normalizing if and only if it is typable,
but also obtain an \emph{upper bound} for the length of the evaluation sequence
of $\tm$ to normal form.
Since these systems are inspired by linear logic and resource consumption~\cite{Gardner94}, 
they are called \emph{quantitative} type systems.
Several works have explored non-idempotent intersection type systems in the context of \CBN~\cite{deCarvalho09, BucciarelliKV17} and
\CBV~\cite{Ehrhard12, Accattoli18, KesnerV22}.
In these systems, the quantitative information can usually be
recorded using \emph{counters} that decorate typing judgments.
For example, a typical judgment may be of the form $\tctx \vdash^n \tm : \typ$,
meaning that $\tm$ normalizes to a normal form of type $\typ$, and 
that the number of reduction steps required to reduce $\tm$ to normal form is 
\emph{bounded} by the value of the counter $n \in \mathbb{N}$.

Quantitative properties provided by non-idempotent intersection type
systems can be refined using a syntactic property, \emph{tightness},
achieved by using only \emph{minimal} type derivations~\cite{AccattoliGK20}. 
Tightness allows to extract from typing derivations not only \emph{upper bounds} 
but rather \emph{exact} quantitative information.
Indeed, in tight derivations, a concluding judgment is still 
of the form $\tctx \vdash^n \tm : \typ$,
but now the counter $n \in \mathbb{N}$ indicates that
\emph{exactly} $n$ reduction steps are required to reduce $\tm$ to normal form.

\paragraph{\PCF and \PCFh.}
\emph{Programming Computable Functions} (\PCF)~\cite{Plotkin77} 
is a 
functional programming language designed to serve as a foundational computational model.
It extends the simply-typed $\lambda$-calculus with
natural numbers, conditional\ expressions,
and general recursion (through a fixed-point operator).
\PCF is simultaneously minimalistic, expressive, and rigorous,
making it an ideal vehicle to study programming language semantics.
Historically, it has been used to study the correspondence
between operational and denotational semantics, \ie, the \emph{full abstraction}
problem~\cite{Plotkin77, Milner77, AbramskyJM00, EhrhardPT18}.

In this work, we study a variant of \PCF called \PCFh,
in which function application and evaluation of conditional
expressions follow a \CBV discipline.
Besides that, the rule to \emph{unfold} a fixed-point is the standard one, namely:
\[
  \fix{\var}{\tm} \to \tm\sub{\var}{\fix{\var}{\tm}}
\]
A noteworthy remark is that the semantics of
recursion is akin to \CBN rather than to \CBV
since the fixed-point can make copies of itself, which is \emph{not}
a value yet.
Given this ``mixed'' operational behavior, we say that evaluation in \PCFh is \emph{hybrid}.
The fact that evaluation is hybrid is challenging from the point of view of 
quantitative semantics~\cite{KesnerV22,BucciarelliKRV23} because it requires 
synthesizing characteristics of \CBN and \CBV in a single system.

\paragraph{A note on nomenclature.}
Languages closely related to \PCFh are sometimes said to
be ``call-by-value''~\cite{DowekL11}.
As explained above, evaluation in \PCFh is actually hybrid
because unfolding a fixed-point may produce copies of itself,
substituting a variable by an expression that is not a value yet.
We write \PCFh to highlight the fact that evaluation is hybrid,
as indicated by the subscript ``\textsf{H}''.

\paragraph{Related work.}
There is a vast amount of literature on \PCF.
In particular, various notions of evaluation have been proposed for \PCF;
for instance, \cite{HondaY99}
provides insights on differences and similarities
between \CBV semantics and \CBN semantics.
A probabilistic version of \PCF is presented in~\cite{EhrhardPT18},
which follows a \CBN discipline but is still hybrid in that
it exhibits a \CBV behavior for its conditional operator.  

Furthermore, some works explore semantical aspects of \PCF
that are connected to our own.
In~\cite{Ehrhard16}, Ehrhard studies an
interpretation of a Call-by-Push-Value (\CBPV) $\lambda$-calculus
with disjoint unions and fixed-points (encompassing \PCF)
in LL, and it derives an intersection type
system from the Scott model of LL.
The unfolding of a fixed-point operator can duplicate itself
by copying a value that represents its suspended computation.
Technically, this is achieved by relying on the of-course ($!$)
modality of LL.
The resulting type system is in the style of Coppo--Dezani,
with \emph{idempotent} intersections, and thus not quantitative.

Recent work in~\cite{Ehrhard23}
derives an intersection type system for
a differential extension of \PCF from a relational semantics.
It is shown that typability in this system implies normalization
by means of a reducibility argument, but the quantitative aspect
is not studied. In contrast, we prove soundness and completeness
by elementary means, and we characterize quantitative upper bounds
and exact bounds via our non-idempotent type system.

There are other languages in the literature that allow the encoding of \CBN/\CBV behavior,
such as Levy's \emph{call-by-push-value} (\CBPV)~\cite{Levy99}
and Ehrhard and Guerrieri's \emph{Bang Calculus}~\cite{EhrhardG16}.
They attempt to unify \CBN and \CBV in a single framework,
but they are not hybrid because they only allow duplicating values.
Recent works~\cite{KesnerV22,BucciarelliKRV23} propose
quantitative type systems for \emph{adequate} variants of the Bang Calculus.
These works study upper bounds and exact bounds. 
However, they are not able to express recursion and integers directly, 
but only through encodings.

\paragraph{Contributions.}
As a main contribution, we propose a quantitative model for 
the hybrid calculus \PCFh,
based on a non-idempotent intersection type  system.
The type system is proved to be sound and complete with respect to the
operational semantics: it characterizes normalizing terms and
provides upper bounds for the length of reductions to normal form.
A peculiar feature of this system is that it distinguishes between
those variables bound by $\lambda$-abstractions from those bound by
fixed-point operators. 
The former may be substituted by \emph{values} while the latter may be 
substituted by arbitrary \emph{expressions}.
Type assignment rules work differently for these two kinds of variables 
because a value is assigned one type each time it will be used (as in \CBV), 
while an arbitrary expression may be copied to produce many values (as in \CBN).

As a second contribution, we refine the type system to obtain
a \emph{tight} version, which provides \emph{exact} bounds for
reduction lengths to normal form, instead of upper bounds.

To the best of our knowledge, this is the first quantitative type interpretation 
that is adequate for a hybrid system.

\paragraph{Structure of the paper.}
\cref{sec:PCF} covers the syntax and operational semantics of \PCFh.
\cref{sec:quantitative_type_system_PCF} introduces a quantitative type system 
for \PCFh.
Finally, the paper concludes by summarizing our work and outlining potential directions for future work.

\section{The \PCFh calculus}
\label{sec:PCF}

This section introduces the syntax of \PCFh, 
together with its associated operational semantics,
which follows a hybrid evaluation.
We characterize the set of normal forms induced by the operational semantics 
through a grammar (\cref{propCBV:characterization_NF}).
Moreover, as the evaluation strategy is not necessarily deterministic, 
we show that it enjoys the Diamond Property (\cref{propCBV:diamond}), 
which implies in particular that all evaluation sequences to normal form have the same length,
thus justifying the use of the proposed strategy to study quantitative behaviors.

Given a denumerable set of 
\textbf{variables} $(\var, \vartwo, \varthree, \hdots)$,
we define \defn{terms} $(\tm, \tmtwo, \tmthree, \hdots)$,
\defn{values} $(\val, \valtwo, \hdots)$,
and \defn{natural values} $(\valnat, \valnattwo, \hdots)$
using the following grammar:
\[
  \tm \eqgram \var
      \mid    \lam{\var}{\tm}
      \mid    \tm \, \tm
      \mid    \zero
      \mid    \succ{\tm}
      \mid    \ifz{\tm}{\tm}{\var}{\tm}
      \mid    \fix{\var}{\tm}
\]
\[
  \val \eqgram \lam{\var}{\tm} \mid \valnat
  \HS\HS
  \valnat \eqgram \zero \mid \succ{\valnat}
\]
Terms include
\textbf{variables} $\var$,
\textbf{abstractions} $\lam{\var}{\tm}$, 
\textbf{applications} $\tm \, \tmtwo$, a
\textbf{fixed-point operator}~$\fix{\var}{\tm}$,
as well as syntax for natural numbers:
\textbf{zero}~$\zero$,
\textbf{the successor constructor}~$\succ{\tm}$,
and a
\textbf{conditional operator}~$\ifz{\tm}{\tmtwo}{\var}{\tmthree}$.
Free and bound occurrences of variables are defined as expected,
considering that the free occurrences of $\var$ in $\tmthree$
are bound in $\ifz{\tm}{\tmtwo}{\var}{\tmthree}$
and in $\fix{\var}{\tmthree}$. 
Terms are considered up to $\alpha$-renaming of bound variables.
We write $\tm\sub{\var}{\tmtwo}$ for the capture-avoiding substitution
of the free occurrences of $\var$ by $\tmtwo$ in $\tm$.

The operational semantics of \PCFh is given by a reduction
relation $\tov{\rulename}$ indexed by a \defn{rule name}
$\rulename \in \set{\bBeta,  \bIfZero, \bIfSucc, \bFix}$,
and it is defined by the following set of rules:
\[
  \indrule{\vruleToBeta}{
    %\emptyPremise
  }{
    (\lam{\var}{\tm})\,\val \tov{\bBeta} \tm\sub{\var}{\val}
  }
  \indrule{\vruleToIfZero}{
    %\emptyPremise
  }{
    \ifz{\zero}{\tm}{\var}{\tmtwo}
    \tov{\bIfZero}
    \tm
  }
\]
\[
  \indrule{\vruleToIfSucc}{
    %\emptyPremise
  }{
    \ifz{\succ{\valnat}}{\tm}{\var}{\tmtwo}
    \tov{\bIfSucc}
    \tmtwo\sub{\var}{\valnat}
  }
  \indrule{\vruleToFix}{
    %\emptyPremise
  }{
    \fix{\var}{\tm}
    \tov{\bFix}
    \tm\sub{\var}{\fix{\var}{\tm}}
  }
\]
\[
  \indrule{\vruleToCongAppL}{
    \tm \tov{\rulename} \tm'
  }{
    \tm\,\tmtwo \tov{\rulename} \tm'\,\tmtwo
  }
  \indrule{\vruleToCongAppR}{
    \tmtwo \tov{\rulename} \tmtwo'
  }{
    \tm\,\tmtwo \tov{\rulename} \tm\,\tmtwo'
  }
\]
\[
  \indrule{\vruleToCongSucc}{
    \tm \tov{\rulename} \tm'
  }{
    \succ{\tm} \tov{\rulename} \succ{\tm'}
  }
  \indrule{\vruleToCongIf}{
    \tm \tov{\rulename} \tm'
  }{
    \ifz{\tm}{\tmtwo}{\var}{\tmthree}
    \tov{\rulename}
    \ifz{\tm'}{\tmtwo}{\var}{\tmthree}
  }
\]
The resulting evaluation strategy is \emph{closed}
(no reduction of terms with free variables) and 
\emph{weak} (no evaluation under abstractions).

Several variants of \PCF can be found in the literature.
In this presentation, function application follows a \CBV discipline because
the argument of an application must be evaluated until
it becomes a value, so that the rule $\vruleToBeta$
can be applied. The conditional operator follows a
\CBV discipline as well, as the guard must be fully evaluated to a
value so that one of the rules $\vruleToIfZero$ or
$\vruleToIfSucc$ can be applied.
The conditional $\ifz{\tm}{\tmtwo}{\var}{\tmthree}$ returns 
$\tmtwo$ if $\tm$ is zero, and $\tmthree$ if $\tm$ is non-zero,
binding $\var$ to its predecessor. 
A consequence of this is that the predecessor function can be defined
using the conditional operator; \eg
$\ifz{\succ{\valnat}}{\zero}{\var}{\var} \tov{\bIfSucc} \valnat$, 
a variant which can  be found in~\cite{EhrhardPT18}.

To illustrate this strategy, let us evaluate the term
$\ifz{\succ{\zero}}{\id}{\var}{\lam{\vartwo}{\var}} \, \succ{\succ{\zero}}$
in one step, where $\id \defeq \lam{\varthree}{\varthree}$:
\[
  \indrule{\vruleToCongAppL}{
    \indrule{\vruleToIfSucc}{
      \emptyPremise
    }{
      \ifz{\succ{\zero}}{\id}{\var}{\lam{\vartwo}{\var}} \tov{\bIfSucc}
      (\lam{\vartwo}{\var})\sub{\var}{\zero}
    }
  }{
    \ifz{\succ{\zero}}{\id}{\var}{\lam{\vartwo}{\var}} \, \succ{\succ{\zero}}
    \tov{\bIfSucc}
    (\lam{\vartwo}{\zero}) \, \succ{\succ{\zero}}
  }
\]
After one more reduction step, by rule $\vruleToBeta$, we obtain
$(\lam{\vartwo}{\zero}) \, \succ{\succ{\zero}} \tov{\bBeta} \zero$.

Let us write $\tm \tovsil \tmtwo$ if $\tm \tov{\rulename} \tmtwo$ for some
rule name $\rulename$.
A term $\tm$
is \defn{$\tovsil$-reducible} if there exists $\tmtwo$ such that
$\tm \tovsil \tmtwo$; and $\tm$
is \defn{$\tovsil$-irreducible} if $\tm$ is not $\tovsil$-reducible.
For the example above, $\zero$ is $\tovsil$-irreducible, while
$\ifz{\succ{\zero}}{\id}{\var}{\lam{\vartwo}{\var}} \, \succ{\succ{\zero}}$
is $\tovsil$-reducible.

Terms that cannot be further $\tovsil$-reduced are
also called \defn{normal forms}. In the following, we characterize
the normal forms of the proposed strategy by means of an inductive definition.
Normal forms are of different kinds or \emph{natures}. 
For instance, $\lam{\var}{\var\,\var}$ is a normal form of
\emph{abstraction} nature,
while $\succ{\zero}$ is a normal form of \emph{natural number} nature.
The preceding normal forms are intuitively ``meaningful'',
and they are said to be \emph{proper}; these normal forms
are those that can be successfully typed in a quantitative type system, 
as we discuss in the next section.
In contrast, there are intuitively ``meaningless'' normal forms such as
$\succ{\lam{\var}{\tm}}$, or $\zero\,\tm$, or
$\ifz{\lam{\var}{\tm}}{\tmtwo}{\vartwo}{\tmthree}$.
These are said to be \emph{stuck}.
Stuck normal forms do not have any meaning,
so they cannot be typed in a quantitative type system. 
We return to this point in the next section.

Formally,
we first distinguish between different \emph{natures} which will be used to
index the sets of normal forms.
Thus, we establish the sets of \defn{proper natures},
$\nature \in \set{\abssym, \natsym}$
and \defn{fallible natures}, 
$\enature \in \set{\abssym, \natsym, \stucksym}$.

We write $\NFv{\enature}$ to denote the set of \defn{normal forms} indexed 
by a fallible nature $\enature$, which is defined by the following rules:
\[
  \indrule{\vruleNFAbs}{
    \emptyPremise
  }{
    \lam{\var}{\tm} \in \NFv{\abssym}
  }
  \indrule{\vruleNFApp}{
    \tm \in \NFv{\enature}
    \HS
    \tmtwo \in \NFv{\enature'}
    \HS
    \enature \neq \abssym
  }{
    \tm\,\tmtwo \in \NFv{\stucksym}
  }
  \indrule{\vruleNFZero}{
    \emptyPremise
  }{
    \zero \in \NFv{\natsym}
  }
\]
\[
  \indrule{\vruleNFSuccNat}{
    \tm \in \NFv{\natsym}
  }{
    \succ{\tm} \in \NFv{\natsym}
  }
  \indrule{\vruleNFSuccErr}{
    \tm \in \NFv{\enature}
    \HS
    \enature \neq \natsym
  }{
    \succ{\tm} \in \NFv{\stucksym}
  }
  \indrule{\vruleNFIf}{
    \tm \in \NFv{\enature}
    \HS
    \enature \neq \natsym
  }{
    \ifz{\tm}{\tmtwo}{\var}{\tmthree} \in \NFv{\stucksym}
  }
\]
For instance, $\lam{\vartwo}{\id \, \id} \in \NF{\abssym}$
because of rule $\vruleNFAbs$, while $(\lam{\vartwo}{\vartwo \, \zero}) \, \id$
is not in normal form since its left subterm is
in normal form but with the proper nature
$\abssym$.

The results below show that this inductively defined notion of normal form
characterizes exactly the set of irreducible terms.
This syntactic characterization is a key ingredient to show that
the quantitative type system studied in \cref{sec:quantitative_type_system_PCF}
is sound and complete.
We start with a simple observation regarding the expected form of normal forms 
based on their assigned proper natures:
\begin{restatable}{lemma}{formsnfnatures}
\label{lemCBV:forms_nf_natures}
Let $\tm \in \NFv{\enature}$.
Then:
\begin{enumerate}
\item
  $\enature = \abssym$
  if and only if
  $\tm$ is of the form $\lam{\var}{\tmtwo}$,
\item
  $\enature = \natsym$
  if and only if
  $\tm$ is of the form $\valnat$.
\end{enumerate}
\end{restatable}
\begin{proof}
See \cref{app:PCF} for details.
\end{proof}

The characterization of normal forms then follows using the previous lemma:
\begin{restatable}[Characterization of normal forms]{proposition}{characterizationNF}
\label{propCBV:characterization_NF}
For any closed term $\tm$ the following are equivalent:
\begin{enumerate}
\item
  There exists a fallible nature $\enature$ such that $\tm \in \NFv{\enature}$.
\item
  $\tm$ is $\tovsil$-irreducible.
\end{enumerate}
\end{restatable}
\begin{proof}
See \cref{app:PCF} for details.
\end{proof}

Evaluation of applications is non-deterministic, \eg\ consider
$((\lam{\var}{\succ{\var}}) \, \zero) \,
(\ifz{\zero}{\id}{\vartwo}{\vartwo \, \id})$.
This term reduces in one step to
$\succ{\zero} \, (\ifz{\zero}{\id}{\vartwo}{\vartwo \, \id})$
by rules $\vruleToCongAppL$ and $\vruleToBeta$, and it also 
reduces in one step to $((\lam{\var}{\succ{\var}}) \, \zero) \, \id$
by rule $\vruleToIfZero$.
Thus, we need to ensure that 
evaluating terms leads to unique normal forms,
despite some form of non-determinism during the computation.
To achieve this, we prove the Diamond Property for the relation $\tovsil$, 
saying that if a term reduces in one step
to two different terms $\tm_1$ and $\tm_2$, then both terms converge in one step
to the same common reduct.
This is a strong form of confluence, which in particular ensures
that all reductions to normal form have the same length:

\begin{restatable}[Diamond property]{proposition}{diamond}
\label{propCBV:diamond}
If $\tm \tov{\rulename_1} \tm_1$ and $\tm \tov{\rulename_2} \tm_2$
where $\tm_1 \neq \tm_2$
then there exists $\tm'$
such that $\tm_1 \tov{\rulename_2} \tm'$ and $\tm_2 \tov{\rulename_1} \tm'$.
\end{restatable}

\begin{proof}
By induction on $\tm$. Details in \cref{app:PCF}.
\end{proof}

Let us take the term of the example above to illustrate this property:
\[
  \xymatrix{
    ((\lam{\var}{\succ{\var}}) \, \zero) \, (\ifz{\zero}{\id}{\vartwo}{\vartwo \, \id}) 
      \arVr{\bBeta}
      \arVd{\bIfZero}
  & \succ{\zero} \, (\ifz{\zero}{\id}{\vartwo}{\vartwo \, \id})
      \arsdVd{\bIfZero}
  \\
    ((\lam{\var}{\succ{\var}}) \, \zero) \, \id
      \arsdVr{\bBeta}
  & \succ{\zero} \, \id
  }
\]
Moreover $\succ{\zero} \, \id \in \NFv{\stucksym}$ by rule $\vruleNFApp$.

% \begin{remark}
% \label{remCBV:forms_nf}
% From the definition above we have the following remarks:
% \begin{itemize}
% \item 
%   Terms of the form $\lam{\var}{\tm}$ only belong to the set $\NFv{\abssym}$.
% \item
%   Terms of the form $\zero$ or $\succ{\valnat}$ only belong to the set $\NFv{\natsym}$.
% \end{itemize}
% \end{remark}

\section{A Quantitative Type System for \PCFh}
\label{sec:quantitative_type_system_PCF}

In this section, we introduce system $\typesystem$ (for \emph{hybrid}), a 
non-idempotent intersection type system for \PCFh.
A type derivation for a term in system $\typesystem$ provides upper bounds for 
its normalization sequences.
Moreover, by considering only \emph{tight} typing derivations, we obtain 
\emph{exact} bounds. 
System $\typesystem$ aligns with other formulations of non-idempotent
intersection type systems and satisfies fundamental properties such as Subject 
Reduction (\cref{lemCBV:subject_reduction}) and Subject Expansion 
(\cref{lemCBV:subject_expansion}). 
Remarkably, this system captures the combined essence of the two underlying 
evaluation strategies found in \PCFh, which are \CBV and \CBN.

The remainder of the section unfolds as follows: \cref{sec:typesystem}
defines the set of types and typing rules governing $\typesystem$.
\cref{sec:properties_typesystem} delves into the proof of soundness
and completeness of $\typesystem$ with respect to the evaluation
strategy \PCFh. For this, we show two quantitative results:
\cref{thmCBV:soundness_completeness}
provides upper bounds for normalization sequences
for any kind of type derivations, while
\cref{thmCBV:tight_soundness_completeness} provides exact bounds for normalization
sequences for those type derivations that are tight.

\subsection{The Type System $\typesystem$}
\label{sec:typesystem}

Recall that values in \PCFh are either abstractions or natural values of the 
form $\succ{\succ{\hdots\succ{\zero}}}$.
A value in \PCFh may be \emph{used} zero, one, or many times.
For example, the underlined identity function in $(\lam{f}{f\,(f\,\zero)})\,\underline{\id}$
is used 
twice in a reduction to normal form since it is bound to the variable f that appears twice, 
and the underlined zero in $(\lam{\var}{\ifz{\var}{\zero}{\vartwo}{\vartwo}})\,\underline{\zero}$
is used  
exactly once in a reduction to normal form because it is bound to a single occurrence of $\var$ 
tested for zero equality.

A noteworthy characteristic of non-idempotent intersection type systems 
is their ability to capture the several roles a single expression may play
in different contexts, so in these type systems
\emph{terms do not necessarily have a unique type}.
Indeed, in system $\typesystem$, 
values are given a \emph{multitype}
of the form
$\mtyp = \msetnu{\nature}{\typ_1,\hdots,\typ_n}$,
which consists of a multiset of types $\mset{\typ_1,\hdots,\typ_n}$
decorated with a nature $\nature$,
and
corresponding to the (non-idempotent) intersection $\typ_1\cap\hdots\cap\typ_n$.
Each $\typ_i$ corresponds to one \emph{use} of the value.
System $\typesystem$ captures the hybrid behavior of \PCFh
by splitting the typing information into two parts,
called the \emph{typing context} and the \emph{family context}.
The typing context contains typing information for
variables bound by \emph{abstractions} and \emph{conditional operators}.
These variables behave like in \CBV, in the sense that they can only be substituted
by \emph{values}, and thus they are assigned a multitype.
The family context contains typing information for
variables bound by \emph{fixed-point operators}.
These variables behave like in \CBN, as they can be substituted
by \emph{unevaluated expressions}.
Since an unevaluated expression may be copied many times
and each copy produces a value, these variables are assigned a
multiset of multitypes $\ftyp = \fset{\mtyp_1,\hdots,\mtyp_n}$,
called a \emph{multitype family}.

We note $\mset{\hdots}$ for the multisets of types and $\fset{\hdots}$ for the multisets of multitypes. 
This is just for visual clarity
and to emphasize the different roles they play in the type system.
However, both notations denote multisets, and both
behave like non-idempotent intersections.

Formally, types of $\typesystem$
are given by the following grammar:
\[
\begin{array}{rrclrcl}
  (\textbf{Types}) &
    \typ & \eqgram & \typabs \mid \typnat
  \\
  (\textbf{$\abssym$-Types}) &
    \typabs & \eqgram & \optmtyp \to \mtyp
  \\
  (\textbf{$\natsym$-Types}) &
    \typnat, \typtwonat, \hdots
            & \eqgram & \zeroTyp
               \mid \succTyp{\mtypnat}
  \\
  (\textbf{Optional Multitypes}) &
    \optmtyp & \eqgram & \none \mid \mtyp
  \\
  (\textbf{Multitypes}) &
    \mtyp, \mtyptwo, \mtypthree, \mtypfour, \hdots
             & \eqgram & \mtypabs
                \mid \mtypnat
  \\
  (\textbf{$\abssym$-Multitypes}) &
    \mtypabs,\mtyptwoabs & \eqgram & \msetnu{\abssym}{\typabs_i}_\iI
  \\
  (\textbf{$\natsym$-Multitypes}) &
    \mtypnat,\mtyptwonat,\mtypthreenat,\hdots
          & \eqgram & \msetnu{\natsym}{\typnat_i}_\iI
  \\
  (\textbf{Multitype Families}) &
    \ftyp & \eqgram & \fset{\mtyp_i}_\iI
\end{array}
\]
A $\nature$-multitype is a multiset of $\nature$-types decorated
with a nature~$\nature$.
\Eg, $\mset{\zeroTyp,\zeroTyp,\succTyp{\emsetnu{\natsym}}}$
is a multiset of $\natsym$-types,
and $\msetnu{\natsym}{\zeroTyp,\zeroTyp,\succTyp{\emsetnu{\natsym}}}$
is a $\natsym$-multitype.
The decoration is important to distinguish $\emsetnu{\natsym}$
from $\emsetnu{\abssym}$,
so that a clear distinction is made between variables
that are bound to natural numbers from those that are bound to abstractions.

Note that a multitype is either a $\natsym$-multitype or an $\abssym$-multitype.
A $\nature$-multitype is said to be of nature $\nature$.
Two multitypes are \defn{compatible} if they are of the same nature.
For example, $\msetnu{\natsym}{\zeroTyp}$
is compatible with $\msetnu{\natsym}{\succTyp{\emsetnu{\natsym}}}$
but not with $\msetnu{\abssym}{\optmtyp\to\mtyp}$.
This notion is extended to optional multitypes
by declaring that $\optmtyp$ and $\optmtyptwo$
are compatible if either of them is $\none$
or if they are compatible multitypes.
The \defn{sum} of compatible multitypes
is defined by
$\msetnu{\abssym}{
   \typabs_1,\hdots,\typabs_n
 }
 +
 \msetnu{\abssym}{
   \typabs_{n+1},\hdots,\typabs_{n+m}
 }
= \msetnu{\abssym}{
    \typabs_1,\hdots,\typabs_{n+m}
  }$
and, similarly,
$\msetnu{\natsym}{
   \typnat_1,\hdots,\typnat_n
 }
 +
 \msetnu{\natsym}{
   \typnat_{n+1},\hdots,\typnat_{n+m}
 }
= \msetnu{\natsym}{
    \typnat_1,\hdots,\typnat_{n+m}
  }$.
This operation is extended to compatible optional multitypes
in the following way: 
$\none+\none = \none$
and $\none+\mtyp = \mtyp$
and $\mtyp+\none = \mtyp$.
Note that $\none$ is the neutral element of the sum.

A \defn{family context}, written $\fctx,\fctxtwo,\hdots$,
is a function mapping variables to multitype families
such that $\fctx(\var) \neq \efset$ for finitely many variables $\var$.
The \defn{sum} of family contexts is defined by
$(\fctx + \fctxtwo)(\var) = \fctx(\var) \oplus \fctxtwo(\var)$,
where $\oplus$ is the sum of multisets of multitypes, defined as
the disjoint union of multisets.
The \defn{domain} of a family context $\fctx$ is defined as
$\dom{\fctx} \eqdef \set{\var \mid \fctx(\var) \neq \efset}$,
and $\emptyctx$ denotes the empty family context,
mapping every variable to $\efset$.

A \defn{typing context}, written $\tctx,\tctxtwo,\hdots$,
is a function mapping variables to optional multitypes
such that $\tctx(\var) \neq \none$ for finitely many variables $\var$.
Two typing contexts $\tctx,\tctxtwo$
are compatible if $\tctx(\var)$ and $\tctxtwo(\var)$
are compatible for every variable $\var$.
The \defn{sum} of compatible typing contexts is defined by
$(\tctx+\tctxtwo)(\var) = \tctx(\var)+\tctxtwo(\var)$.
The \defn{domain} of a typing context $\tctx$ is defined as
$\dom{\tctx} \eqdef \set{\var \mid \tctx(\var) \neq \none}$,
and $\emptyctx$ denotes the empty typing context,
mapping every variable to $\emset$.

A binary relation of \defn{subsumption} $\optmtyp_1 \mleq \mtyp_2$ is defined by
two cases, stating that $\none \mleq \emsetnu{\nature}$ and $\mtyp \mleq \mtyp$
hold.
Subsumption is used to introduce a controlled form of weakening
in the system (see Example \ref{ex:1} below).

A \defn{multi-counter} $\mcount$ is a multiset of rule names, whose cardinality 
is denoted by $\length{\mcount}$.
\defn{Typing judgments} are of the form 
$\judgv[\mcount]{\fctx}{\tctx}{\tm}{\mtyp}$, where $\fctx$ is a family context, 
$\tctx$ is a typing context, $\tm$ is a term and $\mtyp$ is a multitype. 
Moreover, each judgment typing a term $\tm$ is  decorated with a multi-counter
$\mcount$, which traces all the rewriting rules that are used to evaluate the 
term $\tm$ to normal form.
This means in particular that $\length{\mcount}$ reflects the length of 
evaluation sequences to normal form.
The decision to use counters that are multisets of rule names instead of natural 
numbers is to store more precise information. 
Moreover, we use a multiset rather than a 4-uple of distinct counters, as used 
for example in~\cite{KesnerV22}, to avoid unnecessarily inflating the notation. 

As mentioned before, variables can receive two distinct assignments depending 
on whether they occur in the family context or the typing context.
When a variable is assigned a multitype family in the family context,
it will be involved in the evaluation of a fixed-point operator,
meaning that it is intended to be substituted by an unevaluated expression, like in \CBN.
Instead, when a variable is assigned a multitype
in the typing context,
it is intended to be substituted by a value, like in \CBV.
There is an invariant in $\typesystem$ stating that 
family contexts and typing contexts do not share variables,
\ie $\dom{\fctx} \cap \dom{\tctx} = \emptyset$.
When studying properties for $\typesystem$ we
assume implicitly that this invariant holds.

Typing rules for the typing system are the following:
\[
  \indrule{\vruleTypVarOne}{
  }{
    \judgv[\emset]{\emptyctx}{\var : \mtyp}{\var}{\mtyp}
  }
  \indrule{\vruleTypVarTwo}{
  }{
    \judgv[\emset]{\var:\fset{\mtyp}}{\emptyctx}{\var}{\mtyp}
  }
\]
\[
  \indrule{\vruleTypAbs}{
    (\judgv[\mcount_i]{\fctx_i}{\tctx_i,\var:\optmtyp_i}{\tm}{\mtyptwo_i})_\iI
  }{
    \judgv[+_\iI\mcount_i]{+_\iI\fctx_i}{+_\iI\tctx_i}{\lam{\var}{\tm}}{\msetnu{\abssym}{\optmtyp_i \to \mtyptwo_i}_\iI}
  }
\]
\[
  \indrule{\vruleTypApp}{
    \judgv[\mcount_1]{\fctx_1}{\tctx_1}{\tm}{\msetnu{\abssym}{\optmtyp\to\mtyptwo}}
    \HS
    \optmtyp \mleq \mtyp
    \HS
    \judgv[\mcount_2]{\fctx_2}{\tctx_2}{\tmtwo}{\mtyp}
  }{
    \judgv[\mset{\bBeta}+\mcount_1+\mcount_2]{\fctx_1+\fctx_2}{\tctx_1+\tctx_2}{\tm\,\tmtwo}{\mtyptwo}
  }
\]
\[
  \indrule{\vruleTypZero}{
    \emptyPremise
  }{
    \judgv[\emset]{\emptyctx}{\emptyctx}{\zero}{\msetnu{\natsym}{\zeroTyp}_\iI}
  }
  \indrule{\vruleTypSucc}{
    \judgv[\mcount]{\fctx}{\tctx}{\tm}{\mtypnat}
    \HS
    \mtypnat= +_\iI \mtypnat_i
  }{
    \judgv[\mcount]{\fctx}{\tctx}{\succ{\tm}}{\msetnu{\natsym}{\succTyp{\mtypnat_i}}_\iI}
  }
\]
\[
  \indrule{\vruleTypIfZero}{
    \judgv[\mcount_1]{\fctx_1}{\tctx_1}{\tm}{\msetnat{\zeroTyp}}
    \HS
    \judgv[\mcount_2]{\fctx_2}{\tctx_2}{\tmtwo}{\mtyp}
  }{
    \judgv[\mset{\bIfZero} + \mcount_1+\mcount_2]{\fctx_1+\fctx_2}{\tctx_1 + \tctx_2}{\ifz{\tm}{\tmtwo}{\var}{\tmthree}}{\mtyp}
  }
\]
\[
  \indrule{\vruleTypIfSucc}{
    \judgv[\mcount_1]{\fctx_1}{\tctx_1}{\tm}{\msetnat{\succTyp{\mtypnat}}}
    \HS
    \optmtypnat \mleq \mtypnat
    \HS
    \judgv[\mcount_2]{\fctx_2}{\tctx_2,\var:\optmtypnat}{\tmthree}{\mtyp}
  }{
    \judgv[\mset{\bIfSucc} + \mcount_1+\mcount_2]{\fctx_1+\fctx_2}{\tctx_1 + \tctx_2}{\ifz{\tm}{\tmtwo}{\var}{\tmthree}}{\mtyp}
  }
\]
\[
  \indrule{\vruleTypFix}{
    \judgv[\mcount]{\fctx,\var:\fset{\mtyp_i}_\iI}{\tctx}{\tm}{\mtyptwo}
    \HS
    (\judgv[\mcount_i]{\fctx_i}{\tctx_i}{\fix{\var}{\tm}}{\mtyp_i})_\iI
  }{
    \judgv[\mset{\bFix} +\mcount+_\iI\mcount_i]{\fctx+_\iI\fctx_i}{\tctx+_\iI\tctx_i}{\fix{\var}{\tm}}{\mtyptwo}
  }
\]

The set $I$ in the typing rules above can be empty.
For example, when $I = \emptyset$ in rule $\vruleTypAbs$, 
we obtain a judgment of the form
$\judgv[\emset]{\emptyctx}{\emptyctx}{\lam{\var}{\tm}}{\emsetnu{\abssym}}$
with no premises.
As we discussed at the beginning of this subsection,
this result is because the abstraction is not used in the program.
When $I = \emptyset$ in rule $\vruleTypFix$, it means that there are zero uses
of the recursive calls in the program.

Moreover, $I = \emptyset$ in $\vruleTypSucc$
we may encounter a special case that is the following:
\[
  \indrule{\vruleTypSucc}{
    \judgv[\mcount]{\fctx}{\tctx}{\tm}{\emsetnu{\natsym}}
  }{
    \judgv[\mcount]{\fctx}{\tctx}{\succ{\tm}}{\emsetnu{\natsym}}
  }
  \indrule{\vruleTypSucc}{
    \judgv[\mcount]{\fctx}{\tctx}{\tm}{\emsetnu{\natsym}}
  }{
    \judgv[\mcount]{\fctx}{\tctx}{\succ{\tm}}{\msetnat{\succTyp{\emsetnu{\natsym}}}}
  }
\]
Note
that there may be several ways to split a $\natsym$-multitype $\mtypnat$
into a sum $+_\iI \mtypnat_i$ of $\natsym$-multitypes $\mtypnat_i$
in rule $\vruleTypSucc$. Therefore, this rule is non-deterministic
when read from top to bottom.

A \defn{(typing) derivation} is a tree obtained by applying the rules above.
We write $\derivof{\deriv}{\judgv[\mcount]{\fctx}{\tctx}{\tm}{\mtyp}}$
when $\deriv$ is a derivation of the judgment $\judgv[\mcount]{\fctx}{\tctx}{\tm}{\mtyp}$.

As anticipated in the introduction, the hybrid \CBN/\CBV operational nature
of \PCFh
is reflected in the quantitative type system.
This is apparent by comparing the typing rules of system $\typesystem$ with
the rules of quantitative type systems for \CBN/\CBV in the literature. 
To see this formal resemblance more clearly, consider for example
the \CBN quantitative type system $\mathcal{N}$ and
the \CBV quantitative type system $\mathcal{V}$,
both described in~\cite{KesnerV22}.
Rule $\vruleTypVarOne$ of $\typesystem$ types variables involved in \CBV computations,
aligning with rule $\indrulename{var$^{\mathcal{V}}_\texttt{c}$}$ of $\mathcal{V}$.
Conversely, rule $\vruleTypVarTwo$ types variables involved in \CBN computations,
corresponding to rule $\indrulename{var$_\texttt{c}$}$ of $\mathcal{N}$.
The typing rule for abstractions, $\vruleTypAbs$,
coincides with rule $\indrulename{abs$^{\mathcal{V}}_\texttt{c}$}$ of $\mathcal{V}$.
The multitype in the conclusion of rule $\vruleTypAbs$
corresponds to the number of times an abstraction is \emph{used},
as in quantitative type systems for \CBV~\cite{Ehrhard12},
rather than to the number of times the abstraction duplicates
its argument, as in quantitative type systems for \CBN~\cite{Gardner94,deCarvalho07},
reflecting the fact that function applications in \PCFh are evaluated like in \CBV.
Conversely, in rule $\vruleTypFix$
the cardinality of the multitype family $\fset{\mtyp_i}_\iI$
corresponds intuitively to the number of recursive calls,
\ie to the number of times that the fixed-point operator duplicates itself,
which is similar to the typing rule of an application in \CBN.

The goal of providing \PCFh with such system is to characterize normalization
of its evaluation strategy, so, as we briefly mentioned above, the purpose of 
multi-counters $\mcount$ in $\typesystem$ is to provide \emph{upper bounds} for 
the number of evaluation steps to normal form by means of its cardinality.
However, we would also like system $\typesystem$ to go beyond this, by giving 
\emph{exact bounds}, and for that, we make use of special derivations called 
\emph{tight}.
A multitype is \defn{tight} if it is of the form $\emsetnu{\nature}$.
A derivation $\derivof{\deriv}{\judgv[\mcount]{\fctx}{\tctx}{\tm}{\mtyp}}$ is
\defn{tight} if $\fctx$ and $\tctx$ are both empty and if the multitype $\mtyp$ 
is tight.

For example, the following derivation types the term $\lam{\var}{\var \, \zero}$:
\[
  \centering
  \scalebox{0.9}{
    \indrule{\vruleTypAbs}{
      \indrule{\vruleTypApp}{
        \indrule{\vruleTypVarOne}{
          \emptyPremise
        }{
          \judgv[\emset]
            {\emptyctx}{\var : \msetnu{\abssym}{\emsetnu{\natsym} \to \emsetnu{\abssym}}}
            {\var}{\msetnu{\abssym}{\emsetnu{\natsym} \to \emsetnu{\abssym}}}
        }
        \HS
        \indrule{\vruleTypZero}{
          \emptyPremise
        }{
          \judgv[\emset]{\emptyctx}{\emptyctx}{\zero}{\emsetnu{\natsym}}
        }
      }{
        \judgv[\mset{\bBeta}]
          {\emptyctx}{\var : \msetnu{\abssym}{\emsetnu{\abssym} \to \emsetnu{\natsym}}}
          {\var \, \zero}
          {\msetnu{\abssym}{\emsetnu{\abssym}}}
      }
    }{
      \judgv[\mset{\bBeta}]
        {\emptyctx}{\emptyctx}
        {\lam{\var}{\var \, \zero}}
        {\msetnu{\abssym}{\msetnu{\abssym}{\emsetnu{\natsym} \to \emsetnu{\abssym}} \to \emsetnu{\abssym}}}
    }
  }
\]
and the derivation is not tight since the resulting
multitype for the term is not empty.
Moreover, the multi-counter is not an exact bound but rather an upper bound,
given that this term is a normal form. 

On the other hand, we can also build a \emph{tight} derivation for the same term, 
ending in $\judgv[\emset]{\emptyctx}{\emptyctx}{\lam{\var}{\var \, \zero}}{\emsetnu{\abssym}}$
by rule $\vruleTypAbs$ with no premises. Furthermore, the multi-counter here provides
an exact bound.

\paragraph{Controlled weakening}
As mentioned in the introduction,
due to their quantitative nature,
non-idempotent intersection type systems are \emph{resource-aware}.
In consonance with this, system $\typesystem$ is \emph{linear},
hence every single type assumption is used only once.
However, in two specific points a controlled form of \emph{weakening} is needed,
namely in the $\vruleTypApp$ and $\vruleTypIfSucc$ rules,
reflected in the premises that contain the subsumption relation
$\optmtyp \mleq \mtyp$. This allows a variable to be bound to a value which
is then \emph{discarded} without being used.
Observe that the subsumption relation $\mleq$ is used for
controlling the \CBV-like behavior.
For example, in a term like $(\lam{\var}{\zero})\,\tm$, 
the argument $\tm$ must be fully evaluated to a value before
the application can proceed.
Since system $\typesystem$ intends to characterize normalization,
this means that $\tm$ must be typed to ensure that it is normalizing.
If the type of $\tm$ is $\mtyp$ then, in principle, $\var$ should also
be of type $\mtyp$.
But by linearity, $\zero$ is only typable under the empty typing contexts.
Assuming that $\tm$ is a closed term of type $\emsetnu{\natsym}$,
a type derivation can be given as follows:
\label{ex:1}
\[
  \indrule{\vruleTypApp}{
    \indrule{\vruleTypAbs}{
      \indrule{\vruleTypZero}{
        \emptyPremise
      }{
        \judgv[\emset]
          {\emptyctx}{\emptyctx, \var:\none}
          {\zero}{\emsetnu{\natsym}}
      }
    }{
      \judgv[\emset]
        {\emptyctx}{\emptyctx}
        {\lam{\var}{\zero}}{\msetnu{\abssym}{\none \to \emsetnu{\natsym}}}
    }
    \HS
    \none \mleq \emsetnu{\natsym}
    \HS
    \begin{array}{c}
      \vdots
    \\
      \judgv[\mcount]
        {\emptyctx}{\emptyctx}
        {\tm}{\emsetnu{\natsym}}
    \end{array}
  }{
    \judgv[\mset{\bBeta}+\mcount]
      {\emptyctx}{\emptyctx}
      {(\lam{\var}{\zero}) \, \tm}{\emsetnu{\natsym}}
  }
\]
Note that $\none \mleq \emsetnu{\natsym}$ implies the hypothesis $\var : \none$
appears on the typing judgment of $\zero$,
while it forces the argument $\tm$ on the right premise to be typable.
Also, recall that $\emptyctx, \var:\none$ is equal to $\emptyctx$.

Let us now give an example of a type derivation for fixed-point operators, so
let $\tm \defeq (\fix{\fun}{\lam{n}{\ifz{n}{\zero}{m}{\succ{\succ{\fun \, m}}}}}) \, \succ{\zero}$,
where the left subterm is a function returning the double of any natural number:
\eg $\tm$ computes the double of $\succ{\succ{\zero}}$.
Let 
$\mtypnat \defeq \msetnat{\succTyp{\msetnat{\zeroTyp}}}$,
$\mtyptwonat \defeq \msetnat{\zeroTyp}$,
$\mtypabs \defeq \msetabs{\mtypnat \to \emsetnu{\natsym}}$ and
$\mtyptwoabs \defeq \msetabs{\mtyptwonat \to \emsetnu{\natsym}}$.

A tight type derivation for $\tm$ follows: 
\[
  \indrule{\vruleTypApp}{
    \deriv_\text{rec1}
    \HS
    \indrule{\vruleTypSucc}{
      \indrule{\vruleTypZero}{
        \emptyPremise
      }{
        \judgv[\emset]{\emptyctx}{\emptyctx}{\zero}{\mtyptwonat}
      }
    }{
      \judgv[\emset]{\emptyctx}{\emptyctx}{\succ{\zero}}{\mtypnat}
    }
  }{
    \judgv[\mset{\bBeta, \bFix, \bIfSucc, \bBeta, \bFix, \bIfZero}]
      {\emptyset}{\emptyset}
      {(\fix{\fun}{\lam{n}{\ifz{n}{\zero}{m}{\succ{\succ{\fun \, m}}}}}) \, \succ{\zero}}
      {\emsetnu{\natsym}}  
  }
\]
Note that the argument $\succ{\zero}$ is typed with 
the singleton $\natsym$-multitype $\mtypnat$. 
This reflects the single use of $\succ{\zero}$, 
as it corresponds to a single occurrence of $n$ tested for zero equality.
Moreover, its subterm $\zero$ is assigned the singleton $\natsym$-multitype $\mtyptwonat$,
since it is bound to a single occurrence of $n$ which is also tested for zero equality
in the base case of the recursive function.

The typing derivation for the left subterm (the fixed-point operator) follows:
\[
  \centering
  \deriv_\text{rec1} \defeq \left(
    \scalebox{0.87}{
    \indrule{\vruleTypFix}{
      \indrule{\vruleTypAbs}{
        \indrule{\vruleTypIfSucc}{
          \indrule{\vruleTypVarOne}{
            %\emptyPremise
          }{
            \judgv[\emset]
              {\emptyctx}{n : \mtypnat}
              {n}{\mtypnat}
          }
          \HS
          \deriv_\text{cond}
        }{
          \judgv[\mset{\bIfSucc, \bBeta}]
            {\fun : \fset{\mtyptwoabs}}{n : \mtypnat}
            {\ifz{n}{\zero}{m}{\succ{\succ{\fun \, m}}}}
            {\emsetnu{\natsym}}
        }
      }{
        \judgv[\mset{\bIfSucc, \bBeta}]
          {\fun : \fset{\mtyptwoabs}}{\emptyctx}
          {\lam{n}{\ifz{n}{\zero}{m}{\succ{\succ{\fun \, m}}}}}
          {\mtypabs}  
      }
      \deriv_\text{rec2}
    }{
      \judgv[\mset{\bFix, \bIfSucc, \bBeta, \bFix, \bIfZero}]
        {\emptyctx}{\emptyctx}
        {\fix{\fun}{\lam{n}{\ifz{n}{\zero}{m}{\succ{\succ{\fun \, m}}}}}}
        {\mtypabs}  
    }
  }
  \right)
\]
on which the function $\lam{n}{\ifz{n}{\zero}{m}{\succ{\succ{\fun \, m}}}}$ 
must be typed
considering that $n$ shall be bound to $\succ{\zero}$,
so that $n$ is of type $\mtypnat$. 
Moreover, $\fun$ is bound to the
result of the recursive call,
typed in $\deriv_\text{rec2}$:
\[
  \deriv_\text{rec2} \defeq \left(
    \scalebox{0.82}{
      \indrule{\vruleTypFix}{
        \indrule{\vruleTypAbs}{
          \indrule{\vruleTypIfZero}{
            \indrule{\vruleTypVarOne}{
              \emptyPremise
            }{
              \judgv[\emset]
                {\emptyctx}{n : \mtyptwonat}
                {n}{\mtyptwonat}
            }
            \indrule{\vruleTypZero}{
              \emptyPremise
            }{
              \judgv[\emset]
                {\emptyctx}{\emptyctx}
                {\zero}
                {\emsetnu{\natsym}}
            }
          }{
            \judgv[\mset{\bIfZero}]
              {\emptyctx}{n : \mtyptwonat}
              {\ifz{n}{\zero}{m}{\succ{\succ{\fun \, m}}}}
              {\emsetnu{\natsym}}
          }
        }{
          \judgv[\mset{\bIfZero}]
            {\emptyctx}{\emptyctx}
            {\lam{n}{\ifz{n}{\zero}{m}{\succ{\succ{\fun \, m}}}}}
            {\mtyptwoabs}
        }
      }{
        \judgv[\mset{\bFix, \bIfZero}]
          {\emptyctx}{\emptyctx}
          {\fix{\fun}{\lam{n}{\ifz{n}{\zero}{m}{\succ{\succ{\fun \, m}}}}}}
          {\mtyptwoabs}
      }
    }
  \right)
\]
This derivation is the one in charge of typing the base case,
hence $\lam{n}{\ifz{n}{\zero}{m}{\succ{\succ{\fun \, m}}}}$ is typed
considering that $n$ is bound to $\zero$,
so that $n$ is of type $\mtyptwonat$. 
Since $\deriv_\text{rec2}$ corresponds to a base case,
there are no right premises in the rule $\vruleTypFix$.

We finish with the derivation of the \emph{else} branch of the conditional:
\[
  \centering
  \deriv_\text{cond} \defeq \left(
    \scalebox{0.79}{
      \indrule{\vruleTypSucc}{
        \indrule{\vruleTypSucc}{
          \indrule{\vruleTypApp}{
            \indrule{\vruleTypVarTwo}{
              \emptyPremise
            }{
              \judgv[\emset]
                {\fun : \fset{\mtyptwoabs}}{\emptyctx}
                {\fun}
                {\mtyptwoabs}
            }
            \indrule{\vruleTypVarOne}{
              \emptyPremise
            }{
              \judgv[\emset]
                {\emptyctx}{m : \mtyptwonat}
                {m}
                {\mtyptwonat}
            }
          }{
            \judgv[\mset{\bBeta}]
              {\fun : \fset{\mtyptwoabs}}{m : \mtyptwonat}
              {\fun \, m}
              {\emsetnu{\natsym}}
          }
        }{
          \judgv[\mset{\bBeta}]
            {\fun : \fset{\mtyptwoabs}}{m : \mtyptwonat}
            {\succ{\fun \, m}}
            {\emsetnu{\natsym}}
        }
      }{
        \judgv[\mset{\bBeta}]
          {\fun : \fset{\mtyptwoabs}}{m : \mtyptwonat}
          {\succ{\succ{\fun \, m}}}
          {\emsetnu{\natsym}}          
      }
    }
  \right)
\]
here, the term $\succ{\succ{\fun \, m}}$ represents
the result of doubling $\succ{\zero}$, which is the normal form of $\tm$.
As the normal form is not ``used'' in any way
(\ie there are no further interactions with the environment),
it is typed with $\emsetnu{\natsym}$.

\subsection{Properties of $\typesystem$}
\label{sec:properties_typesystem}
This subsection
develops the meta-theory of system $\typesystem$
and proves the main results of this work.
We begin by studying basic properties of the typing system, such as
relevance and splitting lemmas.
Then,
we state auxiliary lemmas to prove soundness for system $\typesystem$,
as well as auxiliary lemmas to prove completeness.

A first remark states
that a natural number must be typed with a $\natsym$-multitype:
\begin{remark}
\label{remCBV:valnat_has_multitype_nat}
If $\judgv[\mcount]{\fctx}{\tctx}{\valnat}{\mtyp}$,
then $\mtyp$ is of the form $\mtypnat$.
\end{remark}

Recall that System $\typesystem$ is linear so, in particular, all assumptions are used.
This property is known as \emph{relevance}:
\begin{lemma}[Relevance] 
\label{lemCBV:relevance}
If $\derivof{\deriv}{\judgv[\mcount]{\fctx}{\tctx}{\tm}{\mtyp}}$ 
then $\dom{\fctx} \cup \dom{\tctx} \subseteq \fv{\tm}$.
\end{lemma}
% Label lemCBV:relevance

\begin{proof}
The proof is straightforward by induction on $\deriv$.
\end{proof}

The next lemma establishes the equivalence between two different conditions 
regarding the typability of a value. 
Specifically, it states that a value is typable with a family context, a typing 
context, and a multi-counter that are empty if and only if it is typed with an 
empty $\nature$-multiset, for any proper nature $\nature$.

\begin{lemma}
\label{lemCBV:typable_values_emptyctx}
Let $\val$ be a value.
Then the following are equivalent:
\begin{enumerate}
\item 
  $\judgv[\mcount]{\fctx}{\tctx}{\val}{\emsetnu{\nature}}$,
  where $\nature = \abssym$ if $\val$ is an abstraction, and
  $\nature = \natsym$ if $\val = \valnat$
\item
  $\fctx = \emptyctx$, $\tctx = \emptyctx$ and $\mcount = \emset$
\end{enumerate}
\end{lemma}
% Label lem:typable_values_emptyctx

\begin{proof}
We prove both items by induction on $\val$.
Cases $\val = \lam{\var}{\tm}$ and $\val = \zero$ are trivial.
Case $\val = \succ{\valnattwo}$ is straightforward using the \ih.
\qedhere

\end{proof}

The following property establishes a form of stability of 
the subsumption relation concerning the sum operator.
\begin{restatable}[Multitype Splitting Lemma]{lemma}{multitypesplitting}
\label{lemCBV:multitype_splitting}
The following hold:
\begin{enumerate}
\item
  If $\none \mleq \mtyp$
  then $\mtyp$ is of the form $\emsetnu{\nature}$ for some proper nature $\nature$.
\item
  $\optmtyp_1 + \optmtyp_2 \mleq \mtyp$
  if and only if there exist multitypes $\mtyp_1,\mtyp_2$
  such that $\mtyp = \mtyp_1 + \mtyp_2$ and
  $\optmtyp_i \mleq \mtyp_i$ for all $i \in \set{1,2}$.
\item
  $+_\iI \optmtyp_i \mleq \mtyp$
  if and only if 
  either $I$ is empty and $\mtyp = \emsetnu{\nature}$ for some proper nature $\nature$,
  or
  $I$ is non-empty and there exist multitypes $(\mtyp_i)_\iI$
  such that $\mtyp = +_\iI \mtyp_i$ and
  $\optmtyp_i \mleq \mtyp_i$ for all $\iI$.
\end{enumerate}
\end{restatable}

The \nameref{lemCBV:multitype_splitting} is used
to split and merge the type derivation of values:

\begin{lemma}[Generalized Value Splitting / Merging]
\label{lemCBV:generalized_value_splitting_merging}
Let $I$ be a finite set, $(\optmtyp_i)_\iI$ a family of optional multitypes 
and $(\ftyp_i)_\iI$ a family of multitype families.
Then the following are equivalent:
\begin{enumerate}
\item
  $\derivof{\deriv}{\judgv[\mcount]{\fctx}{\tctx}{\val}{\mtyp}}$
  with $\mtyp$ a multitype such that $+_\iI \optmtyp_i \mleq \mtyp$
\item
  There exist family contexts $(\fctx_i)_\iI$, typing contexts $(\tctx_i)_\iI$,
  multi-counters $(\mcount_i)_\iI$ and multitypes $(\mtyp_i)_\iI$
  such that $\fctx = +_\iI \fctx_i$ and $\tctx = +_\iI \tctx_i$ and
  $\mcount = +_\iI \mcount_i$ and
  $\derivof{\deriv_i}{\judgv{\fctx_i}{\tctx_i}{\val}{\mtyp_i}}$ and
  $\optmtyp_i \mleq \mtyp_i$ for all $\iI$.
\end{enumerate}
\end{lemma}
\begin{proof}
The proof is straightforward by induction on the cardinality of $I$, 
using \cref{lemCBV:value_splitting_merging}.
\end{proof}

\subsubsection{Lemmas for Soundness}

To show that the type system $\typesystem$ is sound, 
we follow well-known techniques for non-idempotent types (see~\cite{BucciarelliKV17,AccattoliGK20}).
First, a Subject Reduction lemma is established, which is in turn based on two substitution lemmas:
one for each kind of substitution that \PCFh has.

\begin{restatable}[Value Substitution Lemma]{lemma}{valuesubstitution}
\label{lemCBV:value_substitution}
Let $\derivof{\derivtwo}{\judgv[\mcounttwo]{\fctxtwo}{\tctxtwo}{\val}{\mtyptwo}}$ and
let $\optmtyptwo$ be such that $\optmtyptwo \mleq \mtyptwo$.
If $\derivof{\deriv}{\judgv[\mcount]{\fctx}{\tctx,\var:\optmtyptwo}{\tm}{\mtyp}}$
then there exists a derivation $\deriv'$ such that
$\derivof{\deriv'}{\judgv[\mcount+\mcounttwo]{\fctx+\fctxtwo}{\tctx+\tctxtwo}{\tm\sub{\var}{\val}}{\mtyp}}$.
\end{restatable}

\begin{restatable}[Substitution Lemma]{lemma}{substitution}
\label{lemCBV:substitution}
Let $I$ be a finite set, and
$(\derivof{\derivtwo_i}{\judgv[\mcounttwo_i]{\fctxtwo_i}{\tctxtwo_i}{\tmsix}{\mtyptwo_i}})_\iI$ 
a family of type derivations.
If $\derivof{\deriv}{\judgv[\mcount]{\fctx,\var: \fset{\mtyptwo_i}_\iI}{\tctx}{\tm}{\mtyp}}$
then there exists a derivation $\deriv'$ such that
$\derivof{\deriv'}{
  \judgv[\mcount+_\iI \mcounttwo_i]
    {\fctx+_\iI \fctxtwo_i}{\tctx+_\iI \tctxtwo_i}
    {\tm\sub{\var}{\tmsix}}
    {\mtyp}}$.
\end{restatable}

Both substitution lemmas are proved by induction on $\deriv$.
Details are in \cref{app:quantitative_type_system_PCF}.

Now we move to Subject Reduction, which
gives preservation of types.
Moreover, it also provides quantitative information in the sense
that the multi-counter for typing $\tm'$ is smaller than the
multi-counter for typing $\tm$, by exactly one element.

\begin{restatable}[Subject Reduction]{lemma}{subjectreduction}
\label{lemCBV:subject_reduction}
Let $\tm$ be such that
$\tm \tov{\rulename} \tm'$
and $\derivof{\deriv}{\judgv[\mcount]{\fctx}{\tctx}{\tm}{\mtyp}}$.
Then there exist a derivation $\deriv'$ and a multi-counter $\mcount'$
such that $\mset{\rulename} + \mcount' = \mcount$
and $\derivof{\deriv'}{\judgv[\mcount']{\fctx}{\tctx}{\tm'}{\mtyp}}$.
\end{restatable}
\begin{proof}
By induction on the derivation of $\tm \tov{\rulename} \tm'$.
We show only cases $\vruleToBeta$ and $\vruleToFix$ to see how the substitution
lemmas are applied;
the remaining cases are in \cref{app:quantitative_type_system_PCF}.
\begin{itemize}
\item $\vruleToBeta$.
  Then $\tm = (\lam{\var}{\tmtwo}) \, \val \tov{\bBeta} \tmtwo\sub{\var}{\val} = \tm'$.
  The conclusion of $\deriv$ can then only be derived using rule $\vruleTypApp$,
  so $\deriv$ has the form:
  \[
    \indrule{\vruleTypApp}{
      \indrule{\vruleTypAbs}{
        \derivof{\deriv_1}{\judgv[\mcount_1]{\fctx_1}{\tctx_1, \var : \optmtyptwo}{\tmtwo}{\mtyp}}
      }{
        \judgv[\mcount_1]
          {\fctx_1}{\tctx_1}
          {\lam{\var}{\tmtwo}}
          {\msetabs{\optmtyptwo \to \mtyp}}
      }
      (1)\ \optmtyptwo \mleq \mtyptwo
      \HS
      \derivof{\deriv_2}{\judgv[\mcount_2]{\fctx_2}{\tctx_2}{\val}{\mtyptwo}}
    }{
      \judgv[\mset{\bBeta} + \mcount_1 + \mcount_2]
        {\fctx_1 + \fctx_2}{\tctx_1 + \tctx_2}
        {(\lam{\var}{\tmtwo}) \, \val}
        {\mtyp}
    }
  \]
  where $\fctx = \fctx_1 + \fctx_2$ and $\tctx = \tctx_1 + \tctx_2$
  and $\mcount = \mset{\bBeta} + \mcount_1 + \mcount_2$.
  Given (1), we can apply \cref{lemCBV:value_substitution} to $\deriv_1$ with $\deriv_2$,
  yielding
  $\derivof{\deriv'}{
    \judgv[\mcount_1 + \mcount_2]
      {\fctx_1 + \fctx_2}{\tctx_1 + \tctx_2}{\tmtwo\sub{\var}{\val}}{\mtyp}}$,
   and we conclude with $\mcount' = \mcount_1 + \mcount_2$, since 
   $\mset{\bBeta} + \mcount' =\mset{\bBeta} + \mcount_1 + \mcount_2 = \mcount$. 
\item $\vruleTypFix$.
  Then $\tm = \fix{\var}{\tmtwo}\tov{\rulename} \tmtwo\sub{\var}{\fix{\var}{\tmtwo}}= \tm'$.
  The conclusion of $\deriv$ can then only be derived by rule $\vruleTypFix$, so
  $\deriv$ has the form:
  \[
    \indrule{\vruleTypFix}{
      \derivof{\deriv_0}{\judgv[\mcount_0]{\fctx_0,\var:\fset{\mtyptwo_i}_\iI}{\tctx_0}{\tmtwo}{\mtyp}}
      \HS
      (\derivof{\deriv_i}{\judgv[\mcount_i]{\fctx_i}{\tctx_i}{\fix{\var}{\tmtwo}}{\mtyptwo_i}})_\iI
    }{
      \judgv[\mset{\bFix} +\mcount_0 +_\iI \mcount_i]
        {\fctx_0 +_\iI \fctx_i}{\tctx_0 +_\iI\tctx_i}
        {\fix{\var}{\tmtwo}}{\mtyp}
    }
  \]
  where $\fctx = \fctx_0 +_\iI\fctx_i$ and $\tctx = \tctx_0 +_\iI\tctx_i$ and
  $\mcount = \mcount_0 +_\iI\mcount_i$.
  By \Cref{lemCBV:substitution} 
  on $\deriv_0$ with $(\deriv_i)_\iI$ we obtain
  $\derivof{\deriv'}{
    \judgv[\mcount_0 +_\iI\mcount_i]
      {\fctx_0 +_\iI\fctx_i}{\tctx_0 +_\iI\tctx_i}
      {\tmtwo\sub{\var}{\fix{\var}{\tmtwo}}}
      {\mtyp}}$. 
  We let $\mcount' = \mcount_0 +_\iI\mcount_i$
  and we conclude since $\mset{\bFix} + \mcount'=
  \mset{\bFix} + \mcount_0 +_\iI\mcount_i = \mcount$.
  \qedhere
\end{itemize}
\end{proof}

\subsubsection{Lemmas for Completeness}

To show that the type system $\typesystem$ is complete,
the procedure is analogous to the one to prove soundness.
First, a Subject Expansion lemma is required, 
which is in turn based on two anti-substitution lemmas:
one for each kind of substitution that \PCFh has.

\begin{restatable}[Value Anti-Substitution]{lemma}{valueantisubstitution}
\label{lemCBV:value_anti_substitution}
Let $\derivof{\deriv'}{\judgv[\mcount]{\fctx}{\tctx}{\tm\sub{\var}{\val}}{\mtyp}}$.
Then there exist
family contexts $\fctx_1$, $\fctx_2$,
typing contexts $\tctx_1$, $\tctx_2$,
multi-counters $\mcount_1$, $\mcount_2$, 
an optional multitype $\optmtyptwo$ and a multitype $\mtyptwo$ satisfying:
\begin{enumerate}
\item 
  $\derivof{\deriv}{\judgv[\mcount_1]{\fctx_1}{\tctx_1, \var : \optmtyptwo}{\tm}{\mtyp}}$
\item 
  $\derivof{\derivtwo}{\judgv[\mcount_2]{\fctx_2}{\tctx_2}{\val}{\mtyptwo}}$
\item
  $\fctx = \fctx_1 + \fctx_2$ and $\tctx = \tctx_1 + \tctx_2$ and 
  $\mcount = \mcount_1 + \mcount_2$ and $\optmtyptwo \mleq \mtyptwo$
\end{enumerate}
\end{restatable}

\begin{restatable}[Anti-Substitution]{lemma}{antisubstitution}
\label{lemCBV:anti_substitution}
Let $\derivof{\deriv'}{\judgv[\mcount]{\fctx}{\tctx}{\tm\sub{\var}{\tmsix}}{\mtyp}}$.
Then there exist a finite set $I$, 
family contexts $\fctx_0$, $(\fctxtwo_i)_\iI$, 
typing contexts $\tctx_0$, $(\tctxtwo_i)_\iI$,
multi-counters $\mcount_0$, $(\mcounttwo_i)_\iI$ and 
multitypes $(\mtyptwo_i)_\iI$
satisfying:
\begin{enumerate}
\item
  \label{it:anti_substitution_1}
  $\derivof{\deriv}{\judgv[\mcount_0]{\fctx_0, \var : \fset{\mtyptwo_i}_\iI}{\tctx_0}{\tm}{\mtyp}}$
\item
  \label{it:anti_substitution_2}
  $(\derivof{\derivtwo_i}{\judgv[\mcounttwo_i]{\fctxtwo_i}{\tctxtwo_i}{\tmsix}{\mtyptwo_i}})_\iI$
\item
  \label{it:anti_substitution_3}
  $\fctx = \fctx_0 +_\iI \fctxtwo_i$ and $\tctx = \tctx_0 +_\iI \tctxtwo_i$ and
  $\mcount = \mcount_0 +_\iI \mcounttwo_i$
\end{enumerate}
\end{restatable}

Both anti-substitution lemmas are proved by induction on $\deriv$.
Details are in \cref{app:quantitative_type_system_PCF}.

We can proceed now to show Subject Expansion, which is analogous
to Subject Reduction in \cref{lemCBV:subject_reduction}. However, it
goes in the opposite direction of the reduction relation $\tm \tov{\rulename} \tm'$:
given a type derivation of $\tm'$, it proves that $\tm$ has the same type as $\tm'$.

\begin{restatable}[Subject Expansion]{lemma}{subjectexpansion}
\label{lemCBV:subject_expansion}
Let $\tm \tov{\rulename} \tm'$ and 
$\derivof{\deriv'}{\judgv[\mcount']{\fctx}{\tctx}{\tm'}{\mtyp}}$.
Then there exist a derivation $\deriv$ and a multi-counter $\mcount$ such that
$\mcount = \mset{\rulename} + \mcount'$ and
$\judgv[\mcount]{\fctx}{\tctx}{\tm}{\mtyp}$.
\end{restatable}
\begin{proof}
By induction on the derivation of $\tm \tov{\rulename} \tm'$.
Details are in \cref{app:quantitative_type_system_PCF}.
\end{proof}

Furthermore, to prove completeness (\ie normalization implies
typability) it is necessary to show that 
non-stuck normal forms are typable.
As stated in \cref{sec:PCF}, the focus is only on proper normal forms since the
stuck ones do not have any meaning, \ie they represent computation errors that type systems
such as $\typesystem$ do not capture.
\begin{restatable}[Normal Forms are Typable]{lemma}{NFtypable}
\label{lemCBV:NF_typable}
Let $\tm \in \NFv{\nature}$ with $\nature$ a proper nature.
Then $\judgv[\mcount]{\fctx}{\tctx}{\tm}{\mtyp}$
for some $\mcount,\fctx,\tctx,\mtyp$.
\end{restatable}
\begin{proof}
See \cref{app:quantitative_type_system_PCF} for details.
\end{proof}

We can even state a stronger result than the previous one,
namely that proper normal forms can be obtained from a \emph{tight} derivation. 
Moreover, the multi-counter of such tight derivations is empty, which
is crucial to obtain exact bounds by tight type derivations:

\begin{lemma}[Normal Forms are Tight Typable]
\label{lemCBV:NF_tight_typable}
Let $\tm \in \NFv{\nature}$ with $\nature$ a proper nature.
Then $\judgv[\emset]{\emptyctx}{\emptyctx}{\tm}{\emsetnu{\nature}}$.
\end{lemma}
% Label lemCBV:NF_tight_typable

\begin{proof}
By induction on the derivation of $\tm \in \NFv{\nature}$.
Recall cases $\vruleNFApp$, $\vruleNFSuccErr$ and $\vruleNFIf$ do not apply since
$\stucksym$ is not a proper nature.
\begin{enumerate}
\item $\vruleNFAbs$.
  Then $\tm = \lam{\var}{\tmtwo} \in \NFv{\abssym}$
  and indeed $\judgv[\emset]{\emptyctx}{\emptyctx}{\lam{\var}{\tmtwo}}{\emsetnu{\abssym}}$
  by $\vruleTypAbs$.
\item $\vruleNFZero$.
  Then $\tm = \zero \in \NFv{\natsym}$
  and indeed $\judgv[\emset]{\emptyctx}{\emptyctx}{\zero}{\emsetnu{\natsym}}$
  by $\vruleTypZero$.
\item $\vruleNFSuccNat$.
  Then $\tm = \succ{\tmtwo} \in \NFv{\natsym}$ where $\tmtwo \in \NFv{\natsym}$.
  By \ih on $\tmtwo \in \NFv{\natsym}$,
  $\judgv[\emset]{\emptyctx}{\emptyctx}{\tmtwo}{\emsetnu{\natsym}}$,
  hence $\judgv[\emset]{\emptyctx}{\emptyctx}{\succ{\tmtwo}}{\emsetnu{\natsym}}$
  by applying $\vruleTypSucc$.
  \qedhere
\end{enumerate}
\end{proof}

\subsubsection{Soundness and Completeness of System $\typesystem$}

We now turn to the main results of this paper,
which are the properties of soundness and completeness of
$\typesystem$.
Both properties give an equivalence between typability and normalization.
The following theorem provides upper bounds for normalization sequences
of a given term,
for which no tightness condition is required:

\begin{theorem}[Soundness and Completeness of System $\typesystem$ with Upper Bounds]
\label{thmCBV:soundness_completeness}
Let $\tm$ be a closed term, and $\nature$ be a proper nature.
The following are equivalent:
\begin{enumerate}
\item
  $\judgv[\mcount]{\fctx}{\tctx}{\tm}{\mtyp}$
\item
  There exists a sequence of steps
  $\tm = \tm_0 \tov{\rulename_1} \tm_1 \hdots \tov{\rulename_n} \tm_n$
  where $\tm_n \in \NFv{\nature}$ and $\length{\mcount} \geq \length{\mset{\rulename_1, \hdots, \rulename_n}}$.
\end{enumerate}
\end{theorem}
\begin{proof}
\,\\
($1 \Rightarrow 2$)
  By induction on the size of $\mcount$,
  analyzing whether $\tm \in \NFv{\nature}$ or not.
  \begin{itemize}
  \item If $\tm \in \NFv{\nature}$, 
    then $\tm$ is of the form $\lam{\var}{\tmtwo}$, $\zero$, or $\succ{\valnat}$,
    since $\nature$ is a proper nature.
    Hence $\tm$ can only be derived by rule $\vruleTypAbs$, $\vruleTypZero$, or $\vruleTypSucc$ respectively.
    In the three cases, we can take the empty evaluation sequence and we are done.
  \item If $\tm \notin \NFv{\nature}$.
    Then by \cref{propCBV:characterization_NF} it must exist a term $\tm'$
    and a rule name $\rulename$ such that $\tm \tov{\rulename} \tm'$.
    By \nameref{lemCBV:subject_reduction} there exists $\mcount'$ such that
    $\mcount = \mset{\rulename} + \mcount'$, and 
    $\judgv[\mcount']{\fctx}{\tctx}{\tm'}{\mtyp}$.
    By \ih on $\mcount'$, there exists a sequence of steps
    $\tm' \tov{\rulename_1} \tm_1 \hdots \tov{\rulename_n} \tm_n$ where
    $\tm_n \in \NFv{\nature}$ and $\length{\mcount'} \geq \length{\mset{\rulename_1, \hdots, \rulename_n}}$.
    By joining this sequence with the step
    $\tm \tov{\rulename} \tm'$ we obtain the sequence
    $\tm \tov{\rulename} \tm' \tov{\rulename_1} \tm_1 \hdots \tov{\rulename_n} \tm_n$,
    where $\mcount = \mset{\rulename} + \mcount'$,
    and therefore $\length{\mcount} \geq \length{\mset{\rulename, \rulename_1, \hdots, \rulename_n}}$.
  \end{itemize}

\noindent
($2 \Rightarrow 1$)
  By induction on $n$.
  \begin{enumerate}
  \item $n = 0$.
    Then $\tm \in \NFv{\nature}$, and we conclude with
    $\judgv[\mcount]{\fctx}{\tctx}{\tm}{\mtyp}$ by \cref{lemCBV:NF_typable}.
  \item $n > 0$, assuming the property holds for $n-1$.
    Taking the reduction sequence of $(n-1)$ steps from $\tm_1$ to $\tm_n$,
    and $\length{\mcount'} \geq \length{\mset{\rulename_2,\hdots,\rulename_n}}$
    we can apply the \ih, yielding
    $\judgv[\mcount']{\fctx}{\tctx}{\tm_1}{\mtyp}$.
    Since $\tm \tov{\rulename_1} \tm_1$, then 
    $\judgv[\mcount]{\fctx}{\tctx}{\tm}{\mtyp}$
    by \nameref{lemCBV:subject_expansion},
    with $\mcount = \mset{\rulename_1} + \mcount'$,
    so we conclude with $\length{\mcount} \geq \length{\mset{\rulename_1, \rulename_2, \hdots, \rulename_n}}$.
    \qedhere
  \end{enumerate}
\end{proof}

Soundness and completeness of system $\typesystem$ restricting type derivations
to tight ones provide exact bounds for normalization sequences of a given term:

\begin{restatable}[Tight Soundness and Completeness of System $\typesystem$]{theorem}{tightsoundnesscompleteness}
\label{thmCBV:tight_soundness_completeness}
Let $\tm$ be a closed term, and $\nature$ be a proper nature.
The following are equivalent:
\begin{enumerate}
\item
  $\judgv[\mcount]{\emptyctx}{\emptyctx}{\tm}{\emsetnu{\nature}}$
  with $\deriv$ a tight derivation
\item
  There exists a sequence of steps
  $\tm = \tm_0 \tov{\rulename_1} \tm_1 \hdots \tov{\rulename_n} \tm_n$
  where $\tm_n \in \NFv{\nature}$ and $\mcount = \mset{\rulename_1, \hdots, \rulename_n}$.
\end{enumerate}
\end{restatable}
\begin{proof}
Follows the same structure as in \cref{thmCBV:soundness_completeness}, 
\textit{mutatis mutandis}.
The full proof is in \cref{app:quantitative_type_system_PCF}.

To prove soundness $(1 \Rightarrow 2)$,
in the case where $\tm \in \NFv{\nature}$, 
we know that $\tm$ must be a value,
hence the typing derivation
$\judgv[\mcount]{\emptyctx}{\emptyctx}{\tm}{\emsetnu{\nature}}$
of the hypothesis must be such that $\mcount = \emset$
by \cref{lemCBV:typable_values_emptyctx}.
To prove completeness $(2 \Rightarrow 1)$,
we resort to the stronger \cref{lemCBV:NF_tight_typable},
rather than just to \cref{lemCBV:NF_typable},
to ensure that we can construct a typing derivation
of the form
$\judgv[\emset]{\emptyctx}{\emptyctx}{\tm}{\emsetnu{\nature}}$.
\end{proof}

\section{Conclusions}

This paper proposes a quantitative study of a hybrid evaluation strategy for \PCF,
called \PCFh, without relying on any encoding of natural numbers or
fixed-point operators.

Our key contribution is a quantitative semantics for \PCFh,
by means of a non-idempotent intersection type system called $\typesystem$, 
which is sound and complete with respect to the strategy. 
System $\typesystem$ highlights the hybrid nature
of the \PCFh semantics, in the sense that \CBN and \CBV quantitative
behaviors coexist within the same framework.
Moreover, not only $\typesystem$
provides upper bounds for the length of normalization sequences to
normal form,  
but we also achieve exact bounds
by refining the typing derivations of $\typesystem$ to those that are \emph{tight}.

Frameworks such as \CBPV or the \emph{Bang Calculus}
are able to \emph{encode} both \CBN and \CBV cohesively, 
by distinguishing \emph{values} from \emph{computations}, instead of
exhibiting explicitly the hybrid behavior they have.
Therefore, it would be worth studying whether \PCFh can be embedded
into such calculi.

An interesting question is whether the quantitative information of other hybrid 
settings can be expressed using type systems such as $\typesystem$.
Another question for future work is the study of the inhabitation problem in a 
hybrid-type setting, which we conjecture to be decidable, given that it was 
already proven to be decidable in \CBN~\cite{BucciarelliKR14}, \CBV, and \CBPV~\cite{ArrialGK23}.

\bibliographystyle{plain}
\bibliography{biblio}

\newpage
\appendix

\section*{Technical Appendix}
\section{Omitted proofs of \cref{sec:PCF}}
\label{app:PCF}
\formsnfnatures*
% Label lemCBV:forms_nf_natures

\begin{proof}
By induction on the derivation of the judgment $\tm \in \NFv{\enature}$.
\begin{itemize}
\item $\vruleNFAbs$.
  Then $\tm = \lam{\var}{\tmtwo} \in \NFv{\abssym}$,
  where $\enature = \abssym$.
  This case checks item 1.
\item $\vruleNFApp$.
  Then $\tm = \tmtwo \, \tmthree \in \NFv{\stucksym}$,
  where $\enature = \stucksym$.
  This case trivially holds
  since $\tm$ is not an abstraction nor $\valnat$,
  and also $\enature \notin \set{\abssym, \natsym}$.
\item $\vruleNFZero$.
  Then $\tm = \zero \in \NFv{\natsym}$,
  where $\enature = \natsym$.
  This case checks item 2.
\item $\vruleNFSuccNat$.
  Then $\tm = \succ{\tmtwo} \in \NFv{\natsym}$,
  where $\enature = \natsym$, and it is derived from
  $\tmtwo \in \NFv{\natsym}$.
  Then $\tmtwo$ is of the form $\valnat$ by the \ih on $\tmtwo$.
  Hence by definition, this case checks item 2.
\item $\vruleNFSuccErr$.
  Then $\tm = \succ{\tmtwo} \in \NFv{\stucksym}$,
  where $\enature = \stucksym$, and it is derived from
  $\tmtwo \in \NFv{\enature'}$, $\enature' \neq \natsym$.
  By the \ih on $\tmtwo$, $\tmtwo$ is not of the form $\valnat$.
  Hence by definition, $\succ{\tmtwo}$ is not an abstraction nor a $\valnat$,
  so this case also holds.
\item $\vruleNFIf$.
  Then $\tm = \ifz{\tmtwo}{\tmthree}{\var}{\tmfour} \in \NFv{\stucksym}$,
  where $\enature = \stucksym$.
  This case trivially holds since $\enature \notin \set{\abssym, \natsym}$ and
  $\tm$ is not an abstraction nor of the form $\valnat$.
\end{itemize}
\end{proof}

\characterizationNF*
% Label propCBV:characterization_NF

\begin{proof}
\,\\
$1 \Rightarrow 2)$
  By induction on the derivation of the judgment $\tm \in \NFv{\enature}$.
  Note that cases $\vruleNFAbs$ and $\vruleNFZero$
  are immediate since there are no rules for reducing abstractions or zero.
  \begin{itemize}
  % \item $\vruleNFAbs$.
  %   This case is immediate since there are no rules to reduce abstractions.
  \item $\vruleNFApp$.
    Then $\tm = \tmtwo \, \tmthree \in \NFv{\stucksym}$
    is derived from $\tmtwo \in \NFv{\enature}$ and
    $\tmthree \in \NFv{\enaturetwo}$, with $\enature \neq \abssym$.
    Hence $\tmtwo$ is $\tovsil$-irreducible
    by \ih on $\tmtwo$, and $\tmthree$ is $\tovsil$-irreducible
    by \ih on $\tmthree$.
    Moreover, it follows from \cref{lemCBV:forms_nf_natures} that 
    $\tmtwo$ is not an abstraction 
    since $\enature \neq \abssym$ by hypothesis.
    Therefore 
    $\tmtwo \, \tmthree$ is $\tovsil$-irreducible
    since it is not a redex nor its subterms reduce.
  % \item $\vruleNFZero$.
  %   This case is immediate since 
  %   there are no rules to reduce $\tm = \zero$.
  \item $\vruleNFSuccNat$.
    Then $\tm = \succ{\tmthree} \in \NFv{\natsym}$
    is derived from $\tmthree \in \NFv{\natsym}$.
    The only rule to evaluate $\tm$ is $\vruleToCongSucc$.
    But given that $\tmthree$ is $\tovsil$-irreducible
    by the \ih on $\tmthree$, then $\succ{\tmthree}$ does not reduce.
  \item $\vruleNFSuccErr$.
    Analogous to the previous case.
    % Then $\tm = \succ{\tmtwo} \in \NFv{\stucksym}$, 
    % is derived from $\tmtwo \in \NFv{\enature}$, with $\enature \neq \natsym$.
    % The only rule to reduce $\tm$ is $\vruleToCongSucc$,
    % but given that $\tmtwo$ is $\tovsil$-irreducible
    % by \ih on $\tmtwo$, we have that $\succ{\tmtwo}$ does not reduce.
  \item $\vruleNFIf$.
    Then $\tm = \ifz{\tmtwo}{\tmthree}{\var}{\tmfour} \in \NFv{\stucksym}$
    is derived from $\tmtwo \in \NFv{\enature}$, with $\enature \neq \natsym$.
    Hence $\tmtwo$ is $\tovsil$-irreducible by \ih on $\tmtwo$.
    Moreover, it follows from \cref{lemCBV:forms_nf_natures} that 
    $\tmtwo$ is not of the form $\valnat$
    since $\enature \neq \natsym$ by hypothesis.
    Therefore
    $\ifz{\tmtwo}{\tmthree}{\var}{\tmfour}$ 
    is $\tovsil$-irreducible since it is not a redex nor $\tmtwo$ reduces.
  \end{itemize}
\noindent
$2 \Rightarrow 1)$
  By induction on $\tm$, with $\tm$ closed so that $\tm \neq \var$.
  Note that case $\tm = \fix{\var}{\tmtwo}$ is not possible
  since rule $\vruleToFix$ states that any term of the form
  $\fix{\var}{\tmtwo}$ reduces.
  \begin{itemize}
  % \item $\tm = \var$.
  %   This case is not possible since 
  %   $\tm$ is closed by hypothesis.
  \item $\tm = \lam{\var}{\tmtwo}$.
    Assuming $\enature = \abssym$ and
    applying rule $\vruleNFAbs$ then
    $\lam{\var}{\tmtwo} \in \NFv{\abssym}$.
  \item $\tm = \tmtwo \, \tmthree$.
    Since $\tmtwo \, \tmthree$ is $\tovsil$-irreducible,
    it follows that both $\tmtwo$ and $\tmthree$ are also $\tovsil$-irreducible,
    as otherwise $\tm$ would reduce by rules $\vruleToCongAppL$ or $\vruleToCongAppR$.
    Consequently, there exists a fallible nature $\enature_1$ such that
    $\tmtwo \in \NFv{\enature_1}$
    and another fallible nature $\enature_2$ such that
    $\tmthree \in \NFv{\enature_2}$, by \ih on $\tmtwo$ and $\tmthree$ respectively.
    It is worth noting that $\tmtwo$ cannot be of the form $\lam{\var}{\tmtwo'}$,
    since otherwise $\tm$ would reduce by rule $\vruleToBeta$.
    Hence, it must be the case that $\enature_1 \neq \abssym$
    by \cref{lemCBV:forms_nf_natures}.
    Then $\tmtwo \, \tmthree \in \NFv{\stucksym}$ by applying rule $\vruleNFApp$,
    with $\enature = \stucksym$.
  \item $\tm = \zero$.
    Assuming $\enature = \natsym$ and
    applying rule $\vruleNFZero$ then
    $\zero \in \NFv{\natsym}$.
  \item $\tm = \succ{\tmtwo}$.
    Since $\succ{\tmtwo}$ is $\tovsil$-irreducible,
    it follows that $\tmtwo$ is also $\tovsil$-irreducible,
    as otherwise $\tm$ would reduce by rule $\vruleToCongSucc$.
    Then there exists a fallible nature $\enature'$ such that
    $\tmtwo \in \NFv{\enature'}$ by \ih on $\tmtwo$.
    To conclude, we apply rule $\vruleNFSuccNat$ if $\enature' = \natsym$
    or rule $\vruleNFSuccErr$ otherwise, and derive
    $\succ{\tmtwo} \in \NFv{\natsym}$ if $\enature' = \natsym$,
    otherwise $\succ{\tmtwo} \in \NFv{\stucksym}$.
  \item $\tm = \ifz{\tmtwo}{\tmthree}{\var}{\tmfour}$.
    Since $\ifz{\tmtwo}{\tmthree}{\var}{\tmfour}$ is $\tovsil$-irreducible,
    it follows that $\tmtwo$ is also $\tovsil$-irreducible,
    as otherwise $\tm$ would reduce by rule $\vruleToCongIf$.
    Then there exists a fallible nature $\enature'$ such that
    $\tmtwo \in \NFv{\enature'}$ by \ih on $\tmtwo$.
    Moreover, $\tmtwo$ is not of the form $\zero$ nor $\succ{\valnat}$,
    otherwise $\tm$ would reduce by rule $\vruleToIfZero$ or $\vruleToIfSucc$, respectively.
    Therefore it must be the case that $\enature' \neq \natsym$
    by \cref{lemCBV:forms_nf_natures}.
    Then
    $\ifz{\tmtwo}{\tmthree}{\var}{\tmfour} \in \NFv{\stucksym}$
    by applying rule $\vruleNFIf$, with $\enature = \stucksym$.
  \end{itemize}
\end{proof}

\diamond*
\begin{proof}
By induction on $\tm$.
Note cases $\tm = \var$, $\tm = \lam{\var}{\tmtwo}$ and $\tm = \zero$
are impossible by hypothesis.
\begin{itemize}
\item $\tm = \tmtwo \, \tmthree$.
  We can reduce an application using 
  rules $\vruleToBeta$, $\vruleToCongAppL$ and $\vruleToCongAppR$,
  so we have the following cases:
  \begin{itemize}
  \item $\vruleToBeta$/$\vruleToBeta$.
    This case does not apply since we would have that
    $\tm_1 = \tm_2$, which contradicts the hypothesis.
  \item $\vruleToBeta$/$\vruleToCongAppL$.
    We have
    $\tm = (\lam{\var}{\tmtwo'}) \, \val
     \tov{\bBeta}
     \tmtwo'\sub{\var}{\val} = \tm_1$,
    where
    $\tmtwo = \lam{\var}{\tmtwo'}$, $\tmthree = \val$,
    and $\rulename_1 = \bBeta$.
    On the other hand,
    $\tm = (\lam{\var}{\tmtwo'}) \, \val 
     \tov{\rulename_2}
     \tmtwo_2 \, \val = \tm_2$
    is derived from 
    $\lam{\var}{\tmtwo'} \tov{\rulename_2} \tmtwo_2$,
    but this is not possible since 
    there are no rules to reduce abstractions.
  \item $\vruleToBeta$/$\vruleToCongAppR$.
    Analogous to the previous case.
    % We have
    % $\tm = (\lam{\var}{\tmtwo'}) \, \val
    %  \tov{\bBeta}
    %  \tmtwo'\sub{\var}{\val} = \tm_1$,
    % where
    % $\tmtwo = \lam{\var}{\tmtwo'}$, $\tmthree = \val$,
    % and $\rulename_1 = \bBeta$.
    % On the other hand,
    % $\tm = (\lam{\var}{\tmtwo'}) \, \val 
    %  \tov{\rulename_2}
    %  (\lam{\var}{\tmtwo'}) \, \tmthree_2 = \tm_2$
    % is derived from 
    % $\val \tov{\rulename_2} \tmthree_2$,
    % but this is not possible since $\val$ is a value, which by definition does not reduce.
  \item $\vruleToCongAppL$/$\vruleToCongAppL$.
    We have
    $\tm = \tmtwo \, \tmthree \tov{\rulename_1} \tmtwo_1 \, \tmthree = \tm_1$
    derived from 
    $\tmtwo \tov{\rulename_1} \tmtwo_1$,
    and we have
    $\tm = \tmtwo \, \tmthree \tov{\rulename_2} \tmtwo_2 \, \tmthree = \tm_2$
    derived from 
    $\tmtwo \tov{\rulename_2} \tmtwo_2$,
    where $\tmtwo_1 \neq \tmtwo_2$
    since $\tmtwo_1 \, \tmthree \neq \tmtwo_2 \, \tmthree$ by hypothesis.
    We apply the \ih on $\tmtwo$, yielding $\tmtwo'$ such that
    $\tmtwo_1 \tov{\rulename_2} \tmtwo'$ and
    $\tmtwo_2 \tov{\rulename_1} \tmtwo'$.
    Applying rule $\vruleToCongAppL$ to both
    $\tmtwo_1$ and $\tmtwo_2$,
    we obtain
    $\tm_1 = \tmtwo_1 \, \tmthree \tov{\rulename_2} \tmtwo' \, \tmthree = \tm'$
    and
    $\tm_2 = \tmtwo_2 \, \tmthree \tov{\rulename_1} \tmtwo' \, \tmthree = \tm'$
    respectively.
    The following diagram summarizes the proof:
    \[
      \xymatrix{
        \tm = \tmtwo \, \tmthree 
          \arVr{\rulename_1}
          \arVd{\rulename_2}
      & \tmtwo_1 \, \tmthree = \tm_1
          \arsdVd{\rulename_2}
      \\
        \tm_2 = \tmtwo_2 \, \tmthree
          \arsdVr{\rulename_1}
      & \tmtwo' \, \tmthree = \tm'
    }
    \]
  \item $\vruleToCongAppR$/$\vruleToCongAppR$.
    Analogous to the previous case.
    % We have
    % $\tm = \tmtwo \, \tmthree \tov{\rulename_1} \tmtwo \, \tmthree_1 = \tm_1$
    % derived from 
    % $\tmthree \tov{\rulename_1} \tmthree_1$,
    % and we have
    % $\tm = \tmtwo \, \tmthree \tov{\rulename_2} \tmtwo \, \tmthree_2 = \tm_2$
    % derived from 
    % $\tmthree \tov{\rulename_2} \tmthree_2$,
    % where $\tmthree_1 \neq \tmthree_2$
    % since $\tmtwo \, \tmthree_1 \neq \tmtwo \, \tmthree_2$ by hypothesis.
    % We apply \ih on $\tmthree$, yielding $\tmthree'$ such that
    % $\tmthree_1 \tov{\rulename_2} \tmthree'$ and
    % $\tmthree_2 \tov{\rulename_1} \tmthree'$.
    % Applying rule $\vruleToCongAppR$ to both
    % $\tmthree_1$ and $\tmthree_2$,
    % we obtain
    % $\tm_1 = \tmtwo \, \tmthree_1 \tov{\rulename_2} \tmtwo \, \tmthree' = \tm'$
    % and
    % $\tm_2 = \tmtwo \, \tmthree_2 \tov{\rulename_1} \tmtwo \, \tmthree' = \tm'$
    % respectively.
    % The following diagram summarizes the proof:
    % \[
    % \xymatrix{
    %   \tm = \tmtwo \, \tmthree 
    %     \arVr{\rulename_1}
    %     \arVd{\rulename_2}
    % & \tmtwo \, \tmthree_1 = \tm_1
    %     \arsdVd{\rulename_2}
    % & (\vruleToCongAppR)
    % \\
    %   \tm_2 = \tmtwo \, \tmthree_2
    %     \arsdVr{\rulename_1}
    % & \tmtwo \, \tmthree' = \tm'
    % & (\vruleToCongAppR)
    % \\
    %   (\vruleToCongAppR)
    % & (\vruleToCongAppR)
    % }
    % \]
  \item $\vruleToCongAppL$/$\vruleToCongAppR$.
    We have
    $\tm = \tmtwo \, \tmthree \tov{\rulename_1} \tmtwo_1 \, \tmthree = \tm_1$
    derived from 
    (1) $\tmtwo \tov{\rulename_1} \tmtwo_1$,
    and on the other hand we have
    $\tm = \tmtwo \, \tmthree \tov{\rulename_2} \tmtwo \, \tmthree_2 = \tm_2$
    derived from 
    (2) $\tmthree \tov{\rulename_2} \tmthree_2$.
    Given (2), we can apply rule $\vruleToCongAppR$ to $\tmtwo_1 \, \tmthree$,
    yielding
    $\tm_1 = \tmtwo_1 \, \tmthree \tov{\rulename_2} \tmtwo_1 \, \tmthree_2 = \tm'$.
    Lastly, given (1), we can apply rule $\vruleToCongAppL$ to $\tmtwo \, \tmthree_2$,
    yielding
    $\tm_2 = \tmtwo \, \tmthree_2 \tov{\rulename_1} \tmtwo_1 \, \tmthree_2 = \tm'$.
    The following diagram summarizes the proof:
    \[
      \xymatrix{
        \tm = \tmtwo \, \tmthree 
          \arVr{\rulename_1}
          \arVd{\rulename_2}
      & \tmtwo_1 \, \tmthree = \tm_1
          \arsdVd{\rulename_2}
      \\
        \tm_2 = \tmtwo \, \tmthree_2
          \arsdVr{\rulename_1}
      & \tmtwo_1 \, \tmthree_2 = \tm'
    }
    \]
  \end{itemize}
\item $\tm = \succ{\tmtwo}$.
  The only rule to reduce $\tm$ is $\vruleToCongSucc$.
  Hence, we have
  $\tm = \succ{\tmtwo} \tov{\rulename_1} \succ{\tmtwo_1} = \tm_1$,
  derived from
  $\tmtwo \tov{\rulename_1} \tmtwo_1$,
  and on the other hand
  $\tm = \succ{\tmtwo} \tov{\rulename_2} \succ{\tmtwo_2} = \tm_2$,
  derived from
  $\tmtwo \tov{\rulename_2} \tmtwo_2$,
  where $\tmtwo_1 \neq \tmtwo_2$ since $\succ{\tmtwo_1} \neq \succ{\tmtwo_2}$
  by hypothesis.
  We can apply the \ih on $\tmtwo$, yielding $\tmtwo'$ such that 
  $\tmtwo_1 \tov{\rulename_2} \tmtwo'$ and $\tmtwo_2 \tov{\rulename_1} \tmtwo'$.
  Applying rule $\vruleToCongSucc$ to both $\tmtwo_1$ and $\tmtwo_2$, we obtain
  $\succ{\tmtwo_1} \tov{\rulename_2} \succ{\tmtwo'}$ and
  $\succ{\tmtwo_2} \tov{\rulename_1} \succ{\tmtwo'}$, respectively.
  The following diagram summarizes the proof:
    \[
      \xymatrix{
          \tm = \succ{\tmtwo} 
            \arVr{\rulename_1}
            \arVd{\rulename_2}
        & \succ{\tmtwo_1} = \tm_1
            \arsdVd{\rulename_2}
      \\
          \tm_2 = \succ{\tmtwo_2}
            \arsdVr{\rulename_1}
        & \succ{\tmtwo'} = \tm'
      }
    \]
\item $\tm = \ifz{\tmtwo}{\tmthree}{\var}{\tmfour}$.
  We can reduce $\tm$ using
  rules $\vruleToIfZero$, $\vruleToIfSucc$ and $\vruleToCongIf$,
  so we have the following cases:
  \begin{itemize}
  \item $\vruleToIfZero$/$\vruleToIfZero$.
    This case does not apply since it should be that
    $\tm_1 = \tm_2$, which contradicts the hypothesis.
  \item $\vruleToIfSucc$/$\vruleToIfSucc$.
    Analogous to the previous case.
    % This case does not apply since it should be that
    % $\tm_1 = \tm_2$, which contradicts the hypothesis.
  \item $\vruleToIfZero$/$\vruleToIfSucc$.
    This case is not possible since on the one hand
    we have $\tmtwo = \zero$ given that $\tm$ reduces by rule $\vruleToIfZero$,
    but on the other hand, $\tmtwo = \succ{\valnat}$ since $\tm$ also reduces
    by rule $\vruleToIfSucc$.
  \item $\vruleToIfZero$/$\vruleToCongIf$.
    We have
    $\tm = \ifz{\zero}{\tmthree}{\var}{\tmfour}
     \tov{\bIfZero}
     \tmthree = \tm_1$,
    where
    $\tmtwo = \zero$
    and $\rulename_1 = \bIfZero$.
    On the other hand
    $\tm = \ifz{\zero}{\tmthree}{\var}{\tmfour}
     \tov{\rulename_2}
     \ifz{\tmtwo_2}{\tmthree}{\var}{\tmfour} = \tm_2$
    is derived from 
    $\zero \tov{\rulename_2} \tmtwo_2$,
    but this is not possible since 
    there are no rules to reduce $\zero$.
  \item $\vruleToIfSucc$/$\vruleToCongIf$.
    We have
    $\tm = \ifz{\succ{\valnat}}{\tmthree}{\var}{\tmfour} 
     \tov{\bIfSucc}
     \tmfour\sub{\var}{\valnat} = \tm_1$
    where $\tmtwo = \succ{\valnat}$ and $\rulename_1 = \bIfSucc$,
    and we have
    $\tm = \ifz{\succ{\valnat}}{\tmthree}{\var}{\tmfour} 
     \tov{\rulename_2}
     \ifz{\tmtwo_2}{\tmthree}{\var}{\tmfour} = \tm_2$
    derived from 
    $\succ{\valnat} \tov{\rulename_2} \tmtwo_2$.
    This reduction can only be derived by rule $\vruleToCongSucc$,
    hence $\tmtwo_2 = \succ{\tmtwo'_2}$, and the reduction is derived from
    $\valnat \tov{\rulename_2} \tmtwo'_2$,
    but this is not possible since $\valnat$ is a value.
  \item $\vruleToCongIf$/
    $\vruleToCongIf$.
    We have
    $\tm = \ifz{\tmtwo}{\tmthree}{\var}{\tmfour} 
     \tov{\rulename_1} 
     \ifz{\tmtwo_1}{\tmthree}{\var}{\tmfour} = \tm_1$
    derived from 
    $\tmtwo \tov{\rulename_1} \tmtwo_1$,
    and we have
    $\tm = \ifz{\tmtwo}{\tmthree}{\var}{\tmfour} 
     \tov{\rulename_2} 
     \ifz{\tmtwo_2}{\tmthree}{\var}{\tmfour} = \tm_2$
    derived from 
    $\tmtwo \tov{\rulename_2} \tmtwo_2$,
    where $\tmtwo_1 \neq \tmtwo_2$ since 
    $\ifz{\tmtwo_1}{\tmthree}{\var}{\tmfour} 
    \neq \ifz{\tmtwo_2}{\tmthree}{\var}{\tmfour}$ 
    by hypothesis.
    We apply the \ih on $\tmtwo$, yielding $\tmtwo'$ such that
    $\tmtwo_1 \tov{\rulename_2} \tmtwo'$
    and
    $\tmtwo_2 \tov{\rulename_1} \tmtwo'$.
    Applying rule $\vruleToCongIf$ to both
    $\tmtwo_1$ and $\tmtwo_2$,
    we obtain
    $\tm_1 = \ifz{\tmtwo_1}{\tmthree}{\var}{\tmfour} 
    \tov{\rulename_2} 
    \ifz{\tmtwo'}{\tmthree}{\var}{\tmfour} = \tm'$
    and
    $\tm_2 = \ifz{\tmtwo_2}{\tmthree}{\var}{\tmfour} 
    \tov{\rulename_1} 
    \ifz{\tmtwo'}{\tmthree}{\var}{\tmfour} = \tm'$
    respectively.
    The following diagram summarizes the proof:
    \[
      \xymatrix{
        \tm = \ifz{\tmtwo}{\tmthree}{\var}{\tmfour}
          \arVr{\rulename_1} 
          \arVd{\rulename_2}
      & \ifz{\tmtwo_1}{\tmthree}{\var}{\tmfour} = \tm_1
          \arsdVd{\rulename_2}
      \\
        \tm_2 = \ifz{\tmtwo_2}{\tmthree}{\var}{\tmfour}
          \arsdVr{\rulename_1}
      & \ifz{\tmtwo'}{\tmthree}{\var}{\tmfour} = \tm'
      }
    \]
  \end{itemize}
\item $\tm = \fix{\var}{\tmtwo}$.
  This case does not apply since we would have that
  $\tm_1 = \tm_2$, which contradicts the hypothesis.
\end{itemize}
\end{proof}

\section{Omitted proofs of \cref{sec:quantitative_type_system_PCF}}
\label{app:quantitative_type_system_PCF}

Let I and J be two sets.
We write $I \uplus J$ to denote the \defn{disjoint union} of $I$ and $J$.

\multitypesplitting*
% Label lemCBV:multitype_splitting

\begin{proof}
The first item is immediate by definition of $\mleq$.
Both directions of the third item follow from the first two items
by induction on the cardinality of $I$, which is a finite set.
For the direction $(\Rightarrow)$ of the second item, we consider
four subcases, depending on whether each of $\optmtyp_1$
and $\optmtyp_2$ is $\none$ or not.
We only discuss three subcases, as two of them are symmetric:
\begin{itemize}
\item
  If $\optmtyp_1 = \optmtyp_2 = \none$,
  then $\optmtyp_1 + \optmtyp_2 = \none \mleq \mtyp$
  so $\mtyp$ is of the form $\emsetnu{\nature}$ for some proper nature $\nature$.
  It suffices to take $\mtyp_1 = \mtyp_2 \defeq \emsetnu{\nature}$,
  as then
  $\mtyp = \emsetnu{\nature} = \emsetnu{\nature} + \emsetnu{\nature} = \mtyp_1 + \mtyp_2$,
  and $\optmtyp_i = \none \mleq \emsetnu{\nature} = \mtyp_i$
  for all $i \in \set{1,2}$.
\item
  If $\optmtyp_1 \neq \none$ and $\optmtyp_2 = \none$,
  then $\optmtyp_1 + \optmtyp_2 = \optmtyp_1 + \none = \optmtyp_1 \mleq \mtyp$.
  There are two subcases, depending on whether $\mtyp$ is
  an $\abssym$-multitype or a $\natsym$-multitype.
  We only treat the first case, the second one being
  symmetric.
  Let $\mtyp$ be of the form $\msetnu{\abssym}{\typabs_i}_\iI$.
  Then 
  $\optmtyp_1 = \mtyp = \msetnu{\abssym}{\typabs_i}_\iI$
  by definition of $\mleq$,
  since $\optmtyp_1 \mleq \mtyp$ and $\optmtyp_1 \neq \none$.
  It suffices to take $\mtyp_1 \defeq \msetnu{\abssym}{\typabs_i}_\iI$
  and $\mtyp_2 \defeq \emsetnu{\abssym}$,
  as then
  $\mtyp = \msetnu{\abssym}{\typabs_i}_\iI
        = \msetnu{\abssym}{\typabs_i}_\iI + \emsetnu{\abssym}
        = \mtyp_1 + \mtyp_2$,
  where $\optmtyp_1 \mleq \mtyp_1$ and 
  $\optmtyp_2 = \none \mleq \emsetnu{\abssym}$.
\item
  If $\optmtyp_1 \neq \none$ and $\optmtyp_2 \neq \none$,
  then since $\optmtyp_1 + \optmtyp_2$ is well-defined
  we know that $\optmtyp_1$ and $\optmtyp_2$ must be compatible.
  There are two cases, depending on whether they are
  $\abssym$-multitypes or optional $\natsym$-multitypes.
  We only treat the first case, the second one being symmetric.
  Then there exist sets $I,J$ such that
  $\optmtyp_1 = \msetnu{\abssym}{\typabs_i}_\iI$
  and
  $\optmtyp_2 = \msetnu{\abssym}{\typabs_j}_\jJ$.
  Furthermore, by hypothesis
  $\optmtyp_1 + \optmtyp_2
  = \msetnu{\abssym}{\typabs_k}_{k \in I \cup J}
  \mleq \mtyp$,
  so $\mtyp = \msetnu{\abssym}{\typabs_k}_{k \in I \cup J}$,
  and it suffices to take
  $\mtyp_1 \defeq \msetnu{\abssym}{\typabs_i}_\iI$
  and
  $\mtyp_2 \defeq \msetnu{\abssym}{\typabs_j}_\jJ$.
\end{itemize}

For the direction $(\Leftarrow)$, we want
to show that if there exist multitypes $\mtyp_1,\mtyp_2$
such that $\mtyp = \mtyp_1 + \mtyp_2$ and
$\optmtyp_i \mleq \mtyp_i$ for all $i \in \set{1,2}$,
then we have $\optmtyp_1 + \optmtyp_2 \mleq \mtyp$,
we consider four subcases again, depending on whether
each of $\mtyp_1$ and $\mtyp_2$ are empty or not.
The steps are analogous as the other direction.
\end{proof}

\begin{lemma}[Value Splitting / Merging]
\label{lemCBV:value_splitting_merging}
% Para proximos lemas de splitting/merging, mejor hacer el caso para 2 y luego generalizar y decir que se demuestra por induccion.
Let $\optmtyp_1, \optmtyp_2$ be optional multitypes, and 
let $\ftyp_1, \ftyp_2$ be multitype families.
Then the following are equivalent:
\begin{enumerate}
\item
  $\derivof{\deriv}{\judgv[\mcount]{\fctx}{\tctx}{\val}{\mtyp}}$
  with $\mtyp$ a multitype such that $\optmtyp_1 + \optmtyp_2 \mleq \mtyp$
\item
  There exist family contexts $\fctx_1, \fctx_2$, typing contexts $\tctx_1, \tctx_2$,
  multi-counters $\mcount_1, \mcount_2$, and multitypes $\mtyp_1, \mtyp_2$
  such that $\fctx = \fctx_1 + \fctx_2$ and $\tctx = \tctx_1 + \tctx_2$
  and $\mcount = \mcount_1 + \mcount_2$,
  and for all $i = 1, 2$,
  $\derivof{\deriv_i}{\judgv[\mcount_i]{\fctx_i}{\tctx_i}{\val}{\mtyp_i}}$ and
  $\optmtyp_i \mleq \mtyp_i$.
\end{enumerate}
\end{lemma}
% Label lemCBV:value_splitting_merging

\begin{proof}
\,\\
$1 \Rightarrow 2)$
We proceed by induction on $\deriv$.
Since $\val$ is value, the only rules that apply are $\vruleTypAbs$,
$\vruleTypZero$ and $\vruleTypSucc$.
\begin{itemize}
\item $\vruleTypAbs$.
  Then $\val = \lam{\var}{\tm}$, and
  $\derivof{\deriv}{
    \judgv[+_\iI \mcount_i]
      {+_\iI \fctx_i}{+_\iI \tctx_i}
      {\lam{\var}{\tm}}{\msetabs{\optmtyptwo_i \to \mtypthree_i}_\iI}
  }$,
  where $\fctx = +_\iI \fctx_i$ and $\tctx = +_\iI \tctx_i$ and
  $\mcount = +_\iI \mcount_i$ and
  $\mtyp = \msetabs{\optmtyptwo_i \to \mtypthree_i}_\iI$.
  The judgment is derived from
  $\derivof{\deriv'_i}{\judgv[\mcount_i]{\fctx_i}{\tctx_i, \var : \optmtyptwo_i}{\tm}{\mtypthree_i}}$ 
  for each $\iI$.
  Given that $\optmtyp_1 + \optmtyp_2 \mleq \msetabs{\optmtyptwo_i \to \mtypthree_i}_\iI$,
  then there are two subcases to analyze, depending on whether $\mtyp$ is empty or not:
  \begin{itemize}
  \item 
    If $\mtyp = \emset$, then 
    $\mtyp = \msetabs{\optmtyptwo_i \to \mtypthree_i}_\iI = \emsetnu{\abssym}$,
    so $I = \emptyset$.
    Therefore $\deriv$ has no premises and moreover
    $\fctx = \emptyctx$ and $\tctx = \emptyctx$ and $\mcount = \emset$.
    Taking $\fctx_1 = \fctx_2 \defeq \emptyctx$ and $\tctx_1 = \tctx_2 \defeq \emptyctx$
    and $\mcount_1 = \mcount_2 \defeq \emset$ 
    and $\mtyp_1 = \mtyp_2 \defeq \emsetnu{\abssym}$, we obtain the derivations
    $\derivof{\deriv_1}{\judgv[\mcount_1]{\fctx_1}{\tctx_1}{\lam{\var}{\tm}}{\mtyp_1}}$ and 
    $\derivof{\deriv_2}{\judgv[\mcount_2]{\fctx_2}{\tctx_2}{\lam{\var}{\tm}}{\mtyp_2}}$
    by applying rule $\vruleTypAbs$.
    Then $\optmtyp_1 \mleq \mtyp_1$ and $\optmtyp_2 \mleq \mtyp_2$
    by \cref{lemCBV:multitype_splitting} (item 2).
    It is trivial to check that the conditions hold, since $\deriv_1 = \deriv_2 = \deriv$.
  \item 
    If $\mtyp \neq \emset$, then $I \neq \emptyset$, so we can split it as
    $I = I_1 \uplus I_2$ such that 
    $\mtyp_1 = \msetabs{\optmtyptwo_i \to \mtypthree_i}_{i \in I_1}$ and
    $\mtyp_2 = \msetabs{\optmtyptwo_i \to \mtypthree_i}_{i \in I_2}$.
    Taking $\fctx_1 \defeq +_{i \in I_1} \fctx_i$ and
    $\fctx_2 \defeq +_{i \in I_2} \fctx_i$, and 
    $\tctx_1 \defeq +_{i \in I_1} \tctx_i$ and
    $\tctx_2 \defeq +_{i \in I_2} \tctx_i$, and
    $\mcount_1 \defeq +_{i \in I_1} \mcount_i$ and
    $\mcount_2 \defeq +_{i \in I_2} \mcount_i$, and
    $\mtyp_1 \defeq \msetabs{\optmtyptwo_i \to \mtypthree_i}_{i \in I_1}$ and
    $\mtyp_2 \defeq \msetabs{\optmtyptwo_i \to \mtypthree_i}_{i \in I_2}$,
    we can then split $(\deriv'_i)_\iI$ into 
    $(\deriv'_i)_{i \in I_1}$ and $(\deriv'_i)_{i \in I_2}$.
    We can easily check that
    $\fctx = \fctx_1 + \fctx_2$ and $\tctx = \tctx_1 + \tctx_2$ and 
    $\mcount = \mcount_1 + \mcount_2$ and $\mtyp = \mtyp_1 + \mtyp_2$.
    Applying rule $\vruleTypAbs$ to both groups of derivations, we obtain
    $\derivof{\deriv_1}{\judgv[\mcount_1]{\fctx_1}{\tctx_1}{\lam{\var}{\tm}}{\mtyp_1}}$ and
    $\derivof{\deriv_2}{\judgv[\mcount_2]{\fctx_2}{\tctx_2}{\lam{\var}{\tm}}{\mtyp_2}}$.
    And moreover $\optmtyp_1 \mleq \mtyp_1$ and $\optmtyp_2 \mleq \mtyp_2$
    by \cref{lemCBV:multitype_splitting} (item 2).
  \end{itemize}
\item $\vruleTypZero$.
  Then $\val = \zero$ and 
  $\derivof{\deriv}{\judgv[\emset]{\emptyctx}{\emptyctx}{\zero}{\msetnat{\zeroTyp}_\iI}}$,
  where $\fctx = \emptyctx$ and $\tctx = \emptyctx$ and $\mcount = \emset$ and 
  $\mtyp = \msetnat{\zeroTyp}_\iI$.
  Since $\optmtyp_1 + \optmtyp_2 \mleq \msetnat{\zeroTyp}_\iI$, 
  there are two subcases to analyze, depending on whether $\mtyp$ is empty or not:
  \begin{itemize}
  \item
    If $\mtyp = \emset$, then $\optmtyp_1 = \optmtyp_2 = \none$ and
    $\msetnat{\zeroTyp}_\iI = \emsetnu{\natsym}$,
    so $I = \emptyset$.
    This subcase is analogous to the subcase $\mtyp = \emset$ for $\vruleTypAbs$.
  \item
    If $\mtyp \neq \emset$, then 
    $I \neq \emptyset$, so we can split it as
    $I = I_1 \uplus I_2$ so that $\msetnat{\zeroTyp}_\iI 
    = \msetnat{\zeroTyp}_{i \in I_1} + \msetnat{\zeroTyp}_{i \in I_2}$.
    Moreover,
    there exist multitypes $\mtyp_1, \mtyp_2$ such that
    $\optmtyp_1 \mleq \mtyp_1$ and $\optmtyp_2 \mleq \mtyp_2$
    by \cref{lemCBV:multitype_splitting} (item 2).
    Taking $\fctx_1 = \fctx_2 \defeq \emptyctx$, and $\tctx_1 = \tctx_2 \defeq \emptyctx$,
    and $\mcount_1 = \mcount_2 \defeq \emset$, and $\mtyp_1$ and $\mtyp_2$, we check
    $\fctx = \emptyctx = \fctx_1 + \fctx_2$ and $\tctx = \emptyctx = \tctx_1 + \tctx_2$
    and $\mcount = \emset = \mcount_1 + \mcount_2$.
    By rule $\vruleTypZero$ we obtain the derivations
    $\derivof{\deriv_1}{\judgv[\emset]{\emptyctx}{\emptyctx}{\zero}{\msetnat{\zeroTyp}_{i \in I_1}}}$
    and
    $\derivof{\deriv_2}{\judgv[\emset]{\emptyctx}{\emptyctx}{\zero}{\msetnat{\zeroTyp}_{i \in I_2}}}$.
  \end{itemize}
\item $\vruleTypSucc$.
  Then $\val = \succ{\valnat}$ and
  $\derivof{\deriv}{
    \judgv[\mcount]{\fctx}{\tctx}{\succ{\valnat}}{\msetnat{\succTyp{\mtypnat_i}}_\iI}
  }$,
  where $\mtyp = \msetnat{\succTyp{\mtypnat_i}}_\iI$.
  The judgment is derived from
  $\derivof{\deriv'}{\judgv[\mcount]{\fctx}{\tctx}{\valnat}{+_\iI \mtypnat_i}}$.
  Since $\optmtyp_1 + \optmtyp_2 \mleq \msetnat{\succTyp{\mtypnat_i}}_\iI$,
  there are two subcases to analyze, depending on whether $\mtyp$ is empty or not:
  \begin{itemize}
  \item
    Analogous to the subcase $\mtyp = \emset$ in cases $\vruleTypAbs$ and $\vruleTypZero$.
  \item 
    If $\mtyp \neq \emset$, then 
    $I \neq \emptyset$, so we can split it as
    $I = I_1 \uplus I_2$ so that $\msetnat{\succTyp{\mtypnat_i}}_\iI 
    = \msetnat{\succTyp{\mtypnat_i}}_{i \in I_1} + \msetnat{\succTyp{\mtypnat_i}}_{i \in I_2}
    = \mtyp_1 + \mtyp_2$.
    Then $\optmtyp_1 \mleq \mtyp_1$ and $\optmtyp_2 \mleq \mtyp_2$ 
    by \cref{lemCBV:multitype_splitting} (item 2).
    Moreover, the premise can be written as
    $\derivof{\deriv'}{\judgv[\mcount]{\fctx}{\tctx}{\valnat}{+_{i \in I_1} \mtypnat_i +_{i \in I_2} \mtypnat_i}}$,
    and $+_{i \in I_1} \mtypnat_i +_{i \in I_2} \mtypnat_i \mleq +_{i \in I_1} \mtypnat_i +_{i \in I_2} \mtypnat_i$
    by definition.
    We can apply \ih on $\deriv'$, yielding that there exist 
    family contexts $\fctx_1$, $\fctx_2$,
    typing contexts $\tctx_1$, $\tctx_2$,
    multi-counters $\mcount_1$, $\mcount_2$ and 
    multitypes $\mtypnat_1, \mtypnat_2$ such that
    $\fctx = \fctx_1 + \fctx_2$ and
    $\tctx = \tctx_1 + \tctx_2$ and
    $\mcount = \mcount_1 + \mcount_2$, and the derivations
    $\derivof{\deriv'_1}{\judgv[\mcount_1]{\fctx_1}{\tctx_1}{\valnat}{\mtypnat_1}}$ and
    $\derivof{\deriv'_2}{\judgv[\mcount_2]{\fctx_2}{\tctx_2}{\valnat}{\mtypnat_2}}$ hold,
    with $+_{i \in I_1} \mtypnat_i \mleq \mtypnat_1$ and $+_{i \in I_2} \mtypnat_i \mleq \mtypnat_2$.
    Taking $\fctx_1,\fctx_2$ and $\tctx_1,\tctx_2$ and $\mcount_1,\mcount_2$ and 
    the previously defined
    $\mtyp_1 \defeq \msetnat{\succTyp{\mtypnat_i}}_{i \in I_1}$ and
    $\mtyp_2 \defeq \msetnat{\succTyp{\mtypnat_i}}_{i \in I_2}$,
    then $\mtyp = \mtyp_1 + \mtyp_2$; 
    the remainder conditions are already checked in the \ih.
    Applying rule $\vruleTypSucc$ to both $\deriv'_1$ and $\deriv'_2$, we obtain
    $\derivof{\deriv_1}{\judgv[\mcount_1]{\fctx_1}{\tctx_1}{\succ{\valnat}}{\mtyp_1}}$ and
    $\derivof{\deriv_2}{\judgv[\mcount_2]{\fctx_2}{\tctx_2}{\succ{\valnat}}{\mtyp_2}}$, with
    $\optmtyp_1 \mleq \mtyp_1$ and $\optmtyp_2 \mleq \mtyp_2$.
  \end{itemize}
\end{itemize}

\noindent
$2 \Rightarrow 1)$
We proceed by induction on $\deriv_1$ and $\deriv_2$.
Since $\val$ is value, the only rules that apply are $\vruleTypAbs$,
$\vruleTypZero$ and $\vruleTypSucc$.
\begin{itemize}
\item $\vruleTypAbs$.
  Then $\val = \lam{\var}{\tm}$, and we have sets $I_1$ and $I_2$ 
  such that $\deriv_1$ is of the form
  \[
    \deriv_1 \defeq
      \indrule{\vruleTypAbs}{
        (\derivof{\deriv_i}{
          \judgv[\mcount_i]
            {\fctx_i}{\tctx_i, \var : \optmtyptwo_i}
            {\tm}
            {\mtypthree_i})_{i \in I_1}
        }
      }{
        \judgv[+_{i \in I_1} \mcount_i]
          {+_{i \in I_1} \fctx_i}{+_{i \in I_1} \tctx_i}
          {\lam{\var}{\tm}}
          {\msetabs{\optmtyptwo_i \to \mtypthree_i}_{i \in I_1}}
      }
  \]
  where $\fctx_1 = +_{i \in I_1} \fctx_i$ and $\tctx_1 = +_{i \in I_1} \tctx_i$ and
  $\mcount_1 = +_{i \in I_1} \mcount_i$ and
  $\mtyp_1 = \msetabs{\optmtyptwo_i \to \mtypthree_i}_{i \in I_1}$.
  And on the other hand $\deriv_2$ is of the form
  \[
    \deriv_2 \defeq
      \indrule{\vruleTypAbs}{
        (\derivof{\deriv_i}{
          \judgv[\mcount_i]
            {\fctx_i}{\tctx_i, \var : \optmtyptwo_i}
            {\tm}
            {\mtypthree_i})_{i \in I_2}
        }
      }{
        \judgv[+_{i \in I_2} \mcount_i]
          {+_{i \in I_2} \fctx_i}{+_{i \in I_2} \tctx_i}
          {\lam{\var}{\tm}}
          {\msetabs{\optmtyptwo_i \to \mtypthree_i}_{i \in I_2}}
      }
  \]
  where $\fctx_2 = +_{i \in I_2} \fctx_i$ and $\tctx_2 = +_{i \in I_2} \tctx_i$ 
  $\mcount_2 = +_{i \in I_2} \mcount_i$ and
  and $\mtyp_2 = \msetabs{\mtyptwo_i \to \mtypthree_i}_{i \in I_2}$.
  Moreover
  $\fctx = \fctx_1 + \fctx_2 = +_{i \in I_1} \fctx_i +_{i \in I_2} \fctx_i$ and
  $\tctx = \tctx_1 + \tctx_2 = +_{i \in I_1} \tctx_i +_{i \in I_2} \tctx_i$ and
  $\mcount = \mcount_1 + \mcount_2 = +_{i \in I_1} \mcount_i +_{i \in I_2} \mcount_i$ and
  $\optmtyp_1 \mleq \msetabs{\mtyptwo_i \to \mtypthree_i}_{i \in I_1}$ and 
  $\optmtyp_2 \mleq \msetabs{\mtyptwo_i \to \mtypthree_i}_{i \in I_2}$.
  By joining the premises of $\deriv_1$ and $\deriv_2$, and applying rule $\vruleTypAbs$
  we yield
  \[
    \deriv \defeq \left(
    \indrule{\ruleTypAbsV}{
      (\derivof{\deriv_i}{
        \judgv[\mcount_i]]
          {\fctx_i}{\tctx_i, \var : \optmtyptwo_i}
          {\tm}{\mtypthree_i}
      })_{i \in I_1 \uplus I_2}
    }{
      \judgv[+_{i \in I_1 \uplus I_2} \mcount_i]
        {+_{i \in I_1 \uplus I_2} \fctx_i}{+_{i \in I_1 \uplus I_2} \tctx_i}
        {\lam{\var}{\tm}}
        {\msetabs{\optmtyptwo_i \to \mtypthree_i}_{i \in I_1 \uplus I_2}}
    }
    \right)
  \]
  and $\optmtyp_1 + \optmtyp_2 \mleq \msetabs{\optmtyptwo_i \to \mtypthree_i}_{i \in I_1 \uplus I_2}$
  holds by \cref{lemCBV:multitype_splitting} (item 3).
\item $\vruleTypZero$.
  Analogous to the previous case.
  % Lo de abajo no esta adaptado a la version del 10/02, no tiene contadores.
  % Then $\val = \zero$, and we have sets $I_1$ and $I_2$ such that
  % $\deriv_1$ is of the form
  % \[
  %   \deriv_1 \defeq
  %     \indrule{\vruleTypZero}{
  %       \emptyPremise
  %     }{
  %       \judgv
  %         {\emptyctx}{\emptyctx}
  %         {\zero}
  %         {\msetnat{\zeroTyp}_{i \in I_1}}
  %     }
  % \]
  % where $\fctx_1 = \emptyctx$ and $\tctx_1 = \emptyctx$ and
  % $\mtyp_1 = \msetnat{\zeroTyp}_{i \in I_1}$.
  % And on the other hand $\deriv_2$ is of the form
  % \[
  %   \deriv_2 \defeq
  %     \indrule{\vruleTypZero}{
  %       \emptyPremise
  %     }{
  %       \judgv
  %         {\emptyctx}{\emptyctx}
  %         {\zero}
  %         {\msetnat{\zeroTyp}_{i \in I_2}}
  %     }
  % \]
  % where $\fctx_2 = \emptyctx$ and $\tctx_2 = \emptyctx$ and
  % $\mtyp_2 = \msetnat{\zeroTyp}_{i \in I_2}$.
  % Moreover
  % $\fctx = \fctx_1 + \fctx_2 = \emptyctx + \emptyctx = \emptyctx$ and
  % $\tctx = \tctx_1 + \tctx_2 = \emptyctx + \emptyctx = \emptyctx$ and
  % $\optmtyp_1 \mleq \msetnat{\zeroTyp}_{i \in I_1}$ and $\optmtyp_2 \mleq \msetnat{\zeroTyp}_{i \in I_2}$.
  % By rule $\vruleTypZero$ we conclude
  % $\judgv{\emptyctx}{\emptyctx}{\zero}{\msetnat{\zeroTyp}_{i \in I_1 \uplus I_2}}$,
  % and $\optmtyp_1 + \optmtyp_2 \mleq \msetnat{\zeroTyp}_{i \in I_1 \uplus I_2}$
  % holds by \cref{lemCBV:multitype_splitting} (item 3).
\item $\vruleTypSucc$.
  Then $\val = \succ{\valnat}$, and we have sets $I_1$ and $I_2$
  such that $\deriv_1$ is of the form
  \[
    \deriv_1 \defeq 
      \indrule{\vruleTypSucc}{
        \derivof{\deriv'_1}{
          \judgv[\mcount_1]
            {\fctx_1}{\tctx_1}{\valnat}{+_{i \in I_1} \mtypnat_i}
        }
      }{
        \judgv[\mcount_1]
          {\fctx_1}{\tctx_1}{\succ{\valnat}}{\msetnat{\succTyp{\mtypnat_i}}_{i \in I_1}}
      }
  \]
  where $\mtyp_1 = \msetnat{\succTyp{\mtypnat_i}}_{i \in I_1}$.
  And on the other hand $\deriv_2$ is of the form
  \[
    \deriv_2 \defeq 
      \indrule{\vruleTypSucc}{
        \derivof{\deriv'_2}{
          \judgv[\mcount_2]
            {\fctx_2}{\tctx_2}{\valnat}{+_{i \in I_2} \mtypnat_i}
        }
      }{
        \judgv[\mcount_2]
          {\fctx_2}{\tctx_2}{\succ{\valnat}}{\msetnat{\succTyp{\mtypnat_i}}_{i \in I_2}}
      }
  \]
  where $\mtyp_2 = \msetnat{\succTyp{\mtypnat_i}}_{i \in I_2}$.
  Moreover
  $\fctx = \fctx_1 + \fctx_2$ and 
  $\tctx = \tctx_1 + \tctx_2$ and 
  $\mcount = \mcount_1 + \mcount_2$ and 
  $\optmtyp_1 \mleq \msetnat{\succTyp{\mtypnat_i}}_{i \in I_1}$ and
  $\optmtyp_2 \mleq \msetnat{\succTyp{\mtypnat_i}}_{i \in I_2}$.
  Let $j = 1, 2$, then let us define the optional multitype $\optmtyptwo_j$ as
  \[
    \optmtyptwo_j \defeq 
      \left\{ 
      \begin{array}{ll} 
        \none                    & \text{if } \optmtyp_j = \none \\ 
        +_{i \in I_j} \mtypnat_i & \text{otherwise}
      \end{array} \right.
  \]
  so that
  $\optmtyptwo_1 \mleq +_{i \in I_1} \mtypnat_i$ and
  $\optmtyptwo_2 \mleq +_{i \in I_2} \mtypnat_i$.
  We can apply the \ih, yielding
  $\derivof{\deriv'}{\judgv[\mcount]{\fctx}{\tctx}{\valnat}{\mtypnat}}$,
  where $\optmtyptwo_1 + \optmtyptwo_2 \mleq \mtypnat$.
  Moreover, $\mtypnat = +_{i \in I_1} \mtypnat_i +_{i \in I_2} \mtypnat_i$
  by \cref{lemCBV:multitype_splitting} (item 3).
  Applying rule $\vruleTypSucc$, we conclude
  $\judgv[\mcount]{\fctx}{\tctx}{\succ{\valnat}}{\msetnat{\succTyp{\mtypnat_i}}_{i \in I_1 \uplus I_2}}$.
\end{itemize}
\end{proof}

\valuesubstitution*
% Label lemCBV:value_substitution

\begin{proof}
We proceed by induction on $\deriv$.
\begin{enumerate}
\item $\vruleTypVarOne$.
  Then $\derivof{\deriv}{\judgv[\emset]{\emptyctx}{\vartwo : \mtyp}{\vartwo}{\mtyp}}$,
  where $\fctx = \emptyctx$ and $(\tctx,\var:\optmtyptwo) = (\vartwo:\mtyp)$
  and $\mcount = \emset$ and $\tm = \vartwo$.
  We consider two subcases, depending on whether $\var = \vartwo$ or not:
  \begin{enumerate}
  \item
    If $\var = \vartwo$, then $\tctx = \emptyctx$ and $\optmtyptwo = \mtyp$.
    Since $\optmtyptwo = \mtyp \mleq \mtyptwo$, then $\mtyptwo = \mtyp$.
    We conclude by letting $\deriv' \defeq \derivtwo$.
  \item
    If $\var \neq \vartwo$,
    then $\tctx = (\vartwo:\mtyp)$ and thus $\optmtyptwo = \none$.
    By hypothesis,
    $\derivof{\derivtwo}{\judgv[\mcounttwo]{\fctxtwo}{\tctxtwo}{\val}{\mtyptwo}}$
    and moreover $\optmtyptwo = \none \mleq \mtyptwo$, 
    which means that $\mtyptwo$ must be of the form $\emsetnu{\nature}$ 
    for some proper nature $\nature$.
    Hence $\fctxtwo = \emptyctx$ and $\tctxtwo = \emptyctx$ and $\mcounttwo = \emset$
    by \cref{lemCBV:typable_values_emptyctx}.
    We conclude by letting $\deriv' \defeq \deriv$.
  \end{enumerate}
\item $\vruleTypVarTwo$.
  Then $\derivof{\deriv}{\judgv[\emset]{\vartwo : \fset{\mtyp}}{\emptyctx}{\vartwo}{\mtyp}}$,
  where $\fctx = \vartwo : \fset{\mtyp}$ and 
  $(\tctx,\var:\optmtyptwo) = \emptyctx$ and $\mcount = \emset$ and $\tm = \vartwo$.
  Moreover, $\tctx = \emptyctx$ and $\optmtyptwo = \none$,
  since $\var$ must be different from $\vartwo$ by the invariant.
  The rest of the proofs is analogous to subcase $\var \neq \vartwo$ from $\vruleTypVarOne$.
  % By hypothesis,
  % $\derivof{\derivtwo}{\judgv{\fctxtwo}{\tctxtwo}{\val}{\mtyptwo}}$ and moreover 
  % $\optmtyptwo = \none \mleq \mtyptwo$, which means that $\mtyptwo$ must be of 
  % the form $\emsetnu{\nature}$ for some proper nature $\nature$.
  % Hence $\fctxtwo = \emptyctx$ and $\tctxtwo = \emptyctx$ and $\mcounttwo = \emset$
  % by \cref{lemCBV:typable_values_emptyctx}.
  % We conclude by letting $\deriv' \defeq \deriv$.
\item $\vruleTypAbs$.
  Then
  $\derivof{\deriv}{
    \judgv[+_\iI \mcount_i]
      {+_\iI \fctx_i}{+_\iI\tctx_i,\var:+_\iI\optmtyptwo_i}
      {\lam{\vartwo}{\tmtwo}}
      {\msetnu{\abssym}{\optmtypthree_i \to \mtypfour_i}_\iI}
  }$,
  where 
  $\fctx = +_\iI \fctx_i$ and $\tctx = +_\iI\tctx_i$
  and $\optmtyptwo = +_\iI\optmtyptwo_i$ and $\mcount = +_\iI \mcount_i$ 
  and $\tm = \lam{\vartwo}{\tmtwo}$ and 
  $\mtyp = \msetnu{\abssym}{\optmtypthree_i \to \mtypfour_i}_\iI$.
  The judgment is derived from
  $(\derivof{\deriv_i}{
    \judgv[\mcount_i]
      {\fctx_i}{\tctx_i,\var:\optmtyptwo_i,\vartwo:\optmtypthree_i}
      {\tmtwo}{\mtypfour_i})_\iI
  }$.
  Since $+_\iI \optmtyptwo_i = \optmtyptwo \mleq \mtyptwo$, then by 
  \cref{lemCBV:generalized_value_splitting_merging} there exist 
  family contexts $(\fctxtwo_i)_\iI$, typing contexts $(\tctxtwo_i)_\iI$,
  multi-counters $(\mcounttwo_i)_\iI$ and 
  multitypes $(\mtyptwo_i)_\iI$ such that $\fctxtwo = +_\iI \fctxtwo_i$ and 
  $\tctxtwo = +_\iI \tctxtwo_i$ and $\mcounttwo = +_\iI \mcounttwo_i$
  and for all $\iI$
  $\derivof{\derivtwo_i}{\judgv[\mcounttwo_i]{\fctxtwo_i}{\tctxtwo_i}{\val}{\mtyptwo_i}}$
  and $\optmtyptwo_i \mleq \mtyptwo_i$.
  We can apply the \ih on $\deriv_i$ with $\derivtwo_i$, yielding
  $(\derivof{\deriv'_i}{
    \judgv[\mcount_i + \mcounttwo_i]
      {\fctx_i + \fctxtwo_i}{(\tctx_i,\vartwo:\optmtypthree_i)+\tctxtwo_i}
      {\tmtwo\sub{\var}{\val}}{\mtypfour_i})_\iI
  }$.
  By \cref{lemCBV:relevance} and $\alpha$-conversion
  we can write $(\tctx_i,\vartwo:\optmtypthree_i)+\tctxtwo_i$
  as $(\tctx_i+\tctxtwo_i),\vartwo:\optmtypthree_i$.
  Then we can apply rule $\vruleTypAbs$ to $(\deriv'_i)_\iI$, yielding
  $\derivof{\deriv'}{
    \judgv[+_\iI (\mcount_i + \mcounttwo_i)]
      {+_\iI (\fctx_i + \fctxtwo_i)}{+_\iI (\tctx_i + \tctxtwo_i)}
      {\lam{\vartwo}{\tmtwo\sub{\var}{\val}}}
      {\msetabs{\optmtypthree_i\to\mtypfour_i}_\iI}
  }$, and we are done since $\lam{\vartwo}{\tmtwo\sub{\var}{\val}}
  = (\lam{\vartwo}{\tmtwo})\sub{\var}{\val}$.
\item $\vruleTypApp$.
  Then
  $\derivof{\deriv}{
    \judgv[\mset{\bBeta} + \mcount_1 + \mcount_2]
      {\fctx_1 + \fctx_2}{\tctx_1+\tctx_2,\var:\optmtyptwo_1+\optmtyptwo_2}
      {\tmtwo\,\tmthree}
      {\mtyp}}$, where
  $\fctx = \fctx_1 + \fctx_2$ and $\tctx = \tctx_1 + \tctx_2$ and 
  $\optmtyptwo = \optmtyptwo_1 + \optmtyptwo_2$ and
  $\mcount = \mset{\bBeta} + \mcount_1 + \mcount_2$ and
  $\tm = \tmtwo \, \tmthree$, 
  and the judgment is derived from 
  $\derivof{\deriv_1}{\judgv[\mcount_1]{\fctx_1}{\tctx_1,\var:\optmtyptwo_1}{\tmtwo}{\msetabs{\optmtypthree \to \mtyp}}}$
  and $\derivof{\deriv_2}{\judgv[\mcount_2]{\fctx_2}{\tctx_2,\var:\optmtyptwo_2}{\tmthree}{\mtypthree}}$
  and (1) $\optmtypthree \mleq \mtypthree$.
  By \cref{lemCBV:value_splitting_merging} there exist family contexts 
  $\fctxtwo_1$, $\fctxtwo_2$, typing contexts $\tctxtwo_1$, $\tctxtwo_2$,
  multi-counters $\mcounttwo_1$, $\mcounttwo_2$ 
  and multitypes $\mtyptwo_1$, $\mtyptwo_2$ 
  such that $\fctxtwo = \fctxtwo_1 + \fctxtwo_2$ and
  $\tctxtwo = \tctxtwo_1 + \tctxtwo_2$ and
  $\mcounttwo = \mcounttwo_1 + \mcounttwo_2$
  and for each 
  $i = 1,2$ we have
  $\derivof{\derivtwo_i}{\judgv[\mcounttwo_i]{\fctxtwo_i}{\tctxtwo_i}{\val}{\mtyptwo_i}}$ and
  $\optmtyptwo_i \mleq \mtyptwo_i$.
  By \ih on $\deriv_1$ with $\derivtwo_1$ and $\deriv_2$ with $\derivtwo_2$ we have
  $\derivof{\deriv'_1}{
    \judgv[\mcount_1 + \mcounttwo_1]
      {\fctx_1 + \fctxtwo_1}{\tctx_1 + \tctxtwo_1}
      {\tmtwo\sub{\var}{\val}}{\msetabs{\optmtypthree \to \mtyp}}}$ and
  $\derivof{\deriv'_2}{
    \judgv[\mcount_2 + \mcounttwo_2]
      {\fctx_2 + \fctxtwo_2}{\tctx_2+\tctxtwo_2}
      {\tmthree\sub{\var}{\val}}{\mtypthree}}$.
  We can apply rule $\vruleTypApp$
  with $\deriv'_1$, $\deriv'_2$ and (3) as premises, yielding
  $\derivof{\deriv'}
    {\judgv[\mcount + \mcounttwo]
      {\fctx + \fctxtwo}{\tctx + \tctxtwo}{\tmtwo\sub{\var}{\val}\,\tmthree\sub{\var}{\val}}{\mtyp}}$,
  and we are done since $\tmtwo\sub{\var}{\val}\,\tmthree\sub{\var}{\val}
  = (\tmtwo\,\tmthree)\sub{\var}{\val}$.
\item $\vruleTypZero$.
  Then $\derivof{\deriv}{\judgv[\emset]{\emptyctx}{\emptyctx}{\zero}{\msetnat{\zeroTyp}_\iI}}$,
  where $\fctx = \emptyctx$ and $\tctx = \emptyctx$ and $\optmtyptwo = \none$ 
  and $\mcount = \emset$ and $\tm = \zero$ and $\mtyp = \msetnat{\zeroTyp}_\iI$.
  Since $\none \mleq \mtyptwo$, then $\mtyptwo$ is of
  the form $\emsetnu{\nature}$, for some proper nature $\nature$.
  Moreover, $\fctxtwo = \emptyctx$ and $\tctxtwo = \emptyctx$ and $\mcounttwo = \emset$
  by \cref{lemCBV:typable_values_emptyctx}.
  We conclude by letting $\deriv' \defeq \deriv$, since $\zero\sub{\var}{\val} = \zero$.
\item $\vruleTypSucc$.
  Then 
  $\derivof{\deriv}{
    \judgv[\mcount]
      {\fctx}{\tctx, \var : \optmtyptwo}{\succ{\tmthree}}{\msetnat{\succTyp{\mtypnat_i}}_\iI}}$,
  where $\tm = \succ{\tmthree}$ and $\mtyp = \msetnat{\succTyp{\mtypnat_i}}_\iI$.
  The judgment is derived from
  $\derivof{\deriv_0}{\judgv[\mcount]{\fctx}{\tctx, \var : \optmtyptwo}{\tmthree}{+_\iI \mtypnat_i}}$.
  Applying the \ih on $\deriv_0$ with $\derivtwo$ we obtain
  $\derivof{\deriv'_0}{
  \judgv[\mcount + \mcounttwo]
    {\fctx + \fctxtwo}{\tctx + \tctxtwo}
    {\tmthree\sub{\var}{\val}}{+_\iI \mtypnat_i}}$.
  Applying rule $\vruleTypSucc$ we conclude with
  $\derivof{\deriv'}{
    \judgv[\mcount + \mcounttwo]
      {\fctx + \fctxtwo}{\tctx + \tctxtwo}
      {\succ{\tmthree\sub{\var}{\val}}}{\msetnat{\succTyp{\mtypnat_i}}_\iI}}$, 
  and we are done since $\succ{\tmthree\sub{\var}{\val}} = (\succ{\tmthree})\sub{\var}{\val}$.
\item $\vruleTypIfZero$.
  Then
  $\derivof{\deriv}{
    \judgv[\mset{\bIfZero} + \mcount_1 + \mcount_2]
      {\fctx_1 + \fctx_2}{\tctx_1 + \tctx_2, \var : \optmtyptwo_1 + \optmtyptwo_2}
      {\ifz{\tmtwo}{\tmthree}{\vartwo}{\tmfour}}
      {\mtyp}}$
  where $\fctx = \fctx_1+\fctx_2$, $\tctx = \tctx_1+\tctx_2$ and
  $\optmtyptwo = \optmtyptwo_1 + \optmtyptwo_2$ and
  $\mcount = \mset{\bIfZero} + \mcount_1 + \mcount_2$ and
  $\tm = \ifz{\tmtwo}{\tmthree}{\vartwo}{\tmfour}$.
  The judgment is derived from
  $\derivof{\deriv_1}{\judgv[\mcount_1]{\fctx_1}{\tctx_1,\var:\optmtyptwo_1}{\tmtwo}{\msetnat{\zeroTyp}}}$
  and $\derivof{\deriv_2}{\judgv[\mcount_2]{\fctx_2}{\tctx_2,\var:\optmtyptwo_2}{\tmthree}{\mtyp}}$.
  We may assume $\vartwo \notin \set{\var}\cup\fv{\val}$ by $\alpha$-conversion.
  By \cref{lemCBV:value_splitting_merging},
  there exist family contexts $\fctxtwo_1$, $\fctxtwo_2$, 
  typing contexts $\tctxtwo_1$, $\tctxtwo_2$,
  multi-counters $\mcounttwo_1$, $\mcounttwo_2$
  and multitypes $\mtyptwo_1$, $\mtyptwo_2$
  such that $\fctxtwo = \fctxtwo_1 + \fctxtwo_2$ and 
  $\tctxtwo = \tctxtwo_1 + \tctxtwo_2$ and
  $\mcounttwo = \mcounttwo_1 + \mcounttwo_2$
  and for each $i = 1, 2$ we have
  $\derivof{\derivtwo_i}{\judgv[\mcounttwo_i]{\fctxtwo_i}{\tctxtwo_i}{\val}{\mtyptwo_i}}$
  and $\optmtyptwo_i \mleq \mtyptwo_i$.
  By \ih on $\deriv_1$ with $\derivtwo_1$ and $\deriv_2$ with $\derivtwo_2$ we have
  $\derivof{\deriv'_1}{
    \judgv[\mcount_1 + \mcounttwo_1]
      {\fctx_1+\fctxtwo_1}{\tctx_1+\tctxtwo_1}
      {\tmtwo\sub{\var}{\val}}{\msetnat{\zeroTyp}}}$ and 
  $\derivof{\deriv'_2}{
    \judgv[\mcount_1 + \mcounttwo_1]
      {\fctx_2+\fctxtwo_2}{\tctx_2+\tctxtwo_2}
      {\tmthree\sub{\var}{\val}}{\mtyp}}$.
  Applying rule $\vruleTypIfZero$ we conclude with
  $\derivof{\deriv'}{
    \judgv[\mcount + \mcounttwo]
      {\fctx + \fctxtwo}{\tctx + \tctxtwo}
      {\ifz{\tmtwo\sub{\var}{\val}}{\tmthree\sub{\var}{\val}}{\vartwo}{\tmfour\sub{\var}{\val}}}
      {\mtyp}}$,
  and we are done since 
  $\ifz{\tmtwo\sub{\var}{\val}}{\tmthree\sub{\var}{\val}}{\vartwo}{\tmfour\sub{\var}{\val}} 
  = (\ifz{\tmtwo}{\tmthree}{\vartwo}{\tmfour})\sub{\var}{\val}$
\item $\vruleTypIfSucc$.
  Then
  $\derivof{\deriv}{
    \judgv[\mset{\bIfSucc} + \mcount_1 + \mcount_2]
      {\fctx_1 + \fctx_2}{\tctx_1 + \tctx_2, \var : \optmtyptwo_1 + \optmtyptwo_2}
      {\ifz{\tmtwo}{\tmthree}{\vartwo}{\tmfour}}{\mtyp}}$
  where $\fctx = \fctx_1 + \fctx_2$ and $\tctx = \tctx_1 + \tctx_2$ and
  $\optmtyptwo = \optmtyptwo_1 + \optmtyptwo_2$ and
  $\mcount = \mset{\bIfSucc} + \mcount_1 + \mcount_2$ and
  $\tm = \ifz{\tmtwo}{\tmthree}{\vartwo}{\tmfour}$.
  The judgment is derived from
  $\derivof{\deriv_1}{\judgv[\mcount_1]{\fctx_1}{\tctx_1,\var:\optmtyptwo_1}{\tmtwo}{\msetnat{\succTyp{\mtypnat}}}}$
  and 
  $\derivof{\deriv_2}{\judgv[\mcount_2]{\fctx_2}{\tctx_2,\var:\optmtyptwo_2,\vartwo:\optmtypnat}{\tmfour}{\mtyp}}$
  and (1) $\optmtypnat \mleq \mtypnat$.
  We may assume $\vartwo \notin \set{\var}\cup\fv{\val}$ by $\alpha$-conversion.
  By \cref{lemCBV:value_splitting_merging}
  there exist family contexts $\fctxtwo_1$, $\fctxtwo_2$, 
  typing contexts $\tctxtwo_1$, $\tctxtwo_2$,
  multi-counters $\mcounttwo_1$, $\mcounttwo_2$
  and multitypes $\mtyptwo_1$, $\mtyptwo_2$
  such that $\fctxtwo = \fctxtwo_1 + \fctxtwo_2$ and
  $\tctxtwo = \tctxtwo_1 + \tctxtwo_2$ and
  $\mcounttwo = \mcounttwo_1 + \mcounttwo_2$
  and such that for each $i = 1, 2$ we have
  $\derivof{\derivtwo_i}{\judgv[\mcounttwo_i]{\fctxtwo_i}{\tctxtwo_i}{\val}{\mtyptwo_i}}$
  and $\optmtyptwo_i \mleq \mtyptwo_i$.
  By \ih on $\deriv_1$ with $\derivtwo_1$ and $\deriv_2$ with $\derivtwo_2$ we have
  $\derivof{\deriv'_1}{
    \judgv[\mcount_1 + \mcounttwo_1]
      {\fctx_1+\fctxtwo_1}{\tctx_1+\tctxtwo_1}
      {\tmtwo\sub{\var}{\val}}
      {\msetnat{\succTyp{\mtypnat}}}}$ and
  $\derivof{\deriv'_2}{
    \judgv[\mcount_2 + \mcounttwo_2]
      {\fctx_2+\fctxtwo_2}{(\tctx_2,\vartwo:\optmtypnat)+\tctxtwo_2}
      {\tmfour\sub{\var}{\val}}{\mtyp}}$.
  By \cref{lemCBV:relevance} and $\alpha$-conversion
  we can write $(\tctx_2,\vartwo:\optmtypnat)+\tctxtwo_2$
  as $\tctx_2+\tctxtwo_2,\vartwo:\optmtypnat$.
  Applying rule $\vruleTypIfSucc$ with $\deriv'_1$, $\deriv'_2$ and (1)
  as premises we conclude with
  $\derivof{\deriv'}{
    \judgv[\mcount + \mcounttwo]
      {\fctx+\fctxtwo}{\tctx+\tctxtwo}
      {\ifz{\tmtwo\sub{\var}{\val}}{\tmthree\sub{\var}{\val}}{\vartwo}{\tmfour\sub{\var}{\val}}}
      {\mtyp}}$.
\item $\vruleTypFix$.
  Then
  $\derivof{\deriv}{
    \judgv[\mset{\bFix} + \mcount_0 +_\iI \mcount_i]
      {\fctx_0+_\iI\fctx_i}{\tctx_0+_\iI\tctx_i,\var:\optmtyptwo_0+_\iI\optmtyptwo_i}
      {\fix{\vartwo}{\tmtwo}}{\mtyp}}$ 
  where $\fctx = \fctx_0+_\iI\fctx_i$ and $\tctx = \tctx_0+_\iI\tctx_i$ and
  $\optmtyptwo = \optmtyptwo_0+_\iI\optmtyptwo_i$ and
  $\mcount = \mset{\bFix} + \mcount_0 +_\iI \mcount_i$ and
  $\tm = \fix{\vartwo}{\tmtwo}$.
  The judgment is derived from
  $\derivof{\deriv_0}{\judgv[\mcount_0]{\fctx_0,\vartwo:\fset{\mtypthree_i}_\iI}{\tctx_0,\var:\optmtyptwo_0}{\tmtwo}{\mtyp}}$
  and $\derivof{\deriv_i}{\judgv[\mcount_i]{\fctx_i}{\tctx_i,\var:\optmtyptwo_i}{\fix{\vartwo}{\tmtwo}}{\mtypthree_i}}$
  for each $\iI$.
  We may assume $\vartwo \notin \set{\var}\cup\fv{\val}$ by $\alpha$-conversion.
  By \cref{lemCBV:generalized_value_splitting_merging},
  there exist family contexts $\fctxtwo_0$, $(\fctxtwo_i)_\iI$,
  typing contexts $\tctxtwo_0$, $(\tctxtwo_i)_\iI$,
  multi-counters $\mcounttwo_0$, $(\mcounttwo_i)_\iI$
  and multitypes $\mtyptwo_0$, $(\mtyptwo_i)_\iI$
  such that $\fctxtwo = \fctxtwo_0 +_\iI \fctxtwo_i$ and
  $\tctxtwo = \tctxtwo_0 +_\iI \tctxtwo_i$ and
  $\mcounttwo = \mcounttwo_0 +_\iI \mcounttwo_i$
  and such that
  $\derivof{\derivtwo_i}{\judgv[\mcounttwo_i]{\fctxtwo_i}{\tctxtwo_i}{\val}{\mtyptwo_i}}$
  and $\optmtyptwo_i \mleq \mtyptwo_i$
  for each $i \in \set{0} \cup I$.
  By \ih on $\deriv_0$ with $\derivtwo_0$ and $\deriv_i$ with $\derivtwo_i$ for
  each $\iI$, we have
  $\derivof{\deriv'_0}{
    \judgv[\mcount_0 + \mcounttwo_0]
      {(\fctx_0,\vartwo:\fset{\mtypthree_i}_\iI)+\fctxtwo_0}{\tctx_0+\tctxtwo_0}
      {\tmtwo\sub{\var}{\val}}{\mtyp}}$ and
  $\derivof{\deriv'_i}{
    \judgv[\mcount_i + \mcounttwo_i]
      {\fctx_i+\fctxtwo_i}{\tctx_i+\tctxtwo_i}
      {\fix{\vartwo}{\tmtwo\sub{\var}{\val}}}
      {\mtypthree_i}}$ for each $\iI$.
  By \cref{lemCBV:relevance} and $\alpha$-conversion
  we can write $(\fctx_0,\vartwo:\fset{\mtypthree_i}_\iI)+\fctxtwo_0$
  as $\fctx_0+\fctxtwo_0,\vartwo:\fset{\mtypthree_i}_\iI$.
  Applying rule $\vruleTypFix$ with $\deriv'_0$ and $(\deriv'_i)_\iI$ as premises,
  we conclude with
  $\derivof{\deriv'}{
    \judgv[\mcount + \mcounttwo]
      {\fctx+\fctxtwo}{\tctx+\tctxtwo}
      {\fix{\vartwo}{\tmtwo\sub{\var}{\val}}}{\mtyp}}$,
  and we are done since 
  $\fix{\vartwo}{\tmtwo\sub{\var}{\val}} 
  = (\fix{\vartwo}{\tmtwo})\sub{\var}{\val}$
\end{enumerate}
\end{proof}

\substitution*
% Label lemCBV:substitution

\begin{proof}
We proceed by induction on $\deriv$.
\begin{enumerate}
\item $\vruleTypVarOne$.
  Then $\derivof{\deriv}{\judgv[\emset]{\emptyctx}{\vartwo : \mtyp}{\vartwo}{\mtyp}}$,
  where $\fctx, \var: \fset{\mtyptwo_i}_\iI = \emptyctx$ and $\tctx = \vartwo:\mtyp$ and
  $\mcount = \emset$ and $\tm = \vartwo$.
  By the invariant we have $\var \neq \vartwo$, hence $I = \emptyset$.
  We conclude by letting $\deriv' \defeq \deriv$.
\item $\vruleTypVarTwo$.
  Then $\derivof{\deriv}{\judgv[\emset]{\vartwo : \fset{\mtyp}}{\emptyctx}{\vartwo}{\mtyp}}$,
  where $(\fctx, \var : \fset{\mtyptwo_i}_\iI) = \vartwo : \fset{\mtyp}$ and
  $\tctx = \emptyctx$, $\mcount = \emset$ and $\tm = \vartwo$.
  We consider two subcases, depending on whether $\var = \vartwo$ or not:
  \begin{enumerate}
  \item
    If $\var = \vartwo$, then $\fctx = \emptyctx$ and 
    $\fset{\mtyptwo_i}_\iI = \fset{\mtyp}$, \ie $I$ is a singleton,
    so we have only one judgment of the form
    $\judgv[\mcounttwo]{\fctxtwo}{\tctxtwo}{\tmsix}{\mtyptwo}$
    and $\mtyptwo = \mtyp$.
    We conclude by letting $\deriv' \defeq \derivtwo$.
  \item
    If $\var \neq \vartwo$,
    then $\fctx = \vartwo:\fset{\mtyp}$ and $\fset{\mtyptwo_i}_\iI = \efset$,
    that is, $I = \emptyset$.
    We conclude by letting $\deriv' \defeq \deriv$.
  \end{enumerate}
\item $\vruleTypAbs$.
  Then
  $\derivof{\deriv}{
    \judgv[+_\jJ \mcount_j]
      {+_\jJ \fctx_j, \var : +_\jJ \fset{\mtyptwo_i}_{i \in I_j}}{+_\jJ \tctx_j}
      {\lam{\vartwo}{\tmtwo}}
      {\msetabs{\optmtypthree_j \to \mtypfour_j}_\jJ}}$,
  where $\fctx = +_\jJ \fctx_j$ and $\tctx = +_\jJ \tctx_j$ and
  $\mcount = +_\jJ \mcount_j$ and $\tm = \lam{\vartwo}{\tmtwo}$ 
  and $\mtyp = \msetabs{\optmtypthree_j \to \mtypfour_j}_\jJ$ and
  we can write $I$ as $\uplus_\jJ (I_j)_\jJ$ so that
  $\fset{\mtyptwo_i}_\iI = +_\jJ \fset{\mtyptwo_i}_{i \in I_j}$.
  The judgment is derived from premises of the form
  $\derivof{\deriv_j}{
    \judgv[\mcount_j]
      {\fctx_j, \var : \fset{\mtyptwo_i}_{i \in I_j}}{\tctx_j, \vartwo : \optmtypthree_j}
      {\tmtwo}
      {\mtypfour_j}}$
  for each $\jJ$.
  Then for all $\jJ$ we can apply the \ih on $\deriv_j$ with $\derivtwo_i$ for all $i \in I_j$,
  yielding
  $(\derivof{\deriv'_j}{
    \judgv[\mcount_j +_{i \in I_j} \mcounttwo_i]
      {\fctx_j +_{i \in I_j} \fctxtwo_i}{(\tctx_j,\vartwo : \optmtypthree_j )+_{i \in I_j} \tctxtwo_i}
      {\tmtwo\sub{\var}{\tmsix}}
      {\mtypfour_j}})_\jJ$.
  By \cref{lemCBV:relevance} and $\alpha$-conversion
  we can write $(\tctx_j,\vartwo : \optmtypthree_j )+_{i \in I_j} \tctxtwo_i$
  as $\tctx_j+_{i \in I_j} \tctxtwo_i,\vartwo : \optmtypthree_j$.
  We conclude 
  $\derivof{\deriv'}{
    \judgv[\mcount +_\iI \mcounttwo_i]
      {\fctx +_\iI \fctxtwo_i}{\tctx +_\iI \tctxtwo_i}
      {\lam{\vartwo}{\tmtwo\sub{\var}{\tmsix}}}
      {\msetabs{\optmtypthree_j \to \mtypfour_j}}_\jJ}$.
  by applying rule $\vruleTypAbs$ to $(\deriv'_j)_\jJ$,
  and we are done since $\lam{\vartwo}{\tmtwo\sub{\var}{\val}}
  = (\lam{\vartwo}{\tmtwo})\sub{\var}{\val}$.
\item $\vruleTypApp$.
  Then
  $\derivof{\deriv}{
    \judgv[\mset{\bBeta} + \mcount_1 + \mcount_2]
      {\fctx_1 + \fctx_2, \var : \fset{\mtyptwo_j}_\jJ + \fset{\mtyptwo_k}_\kK}{\tctx_1 + \tctx_2}
      {\tmtwo \, \tmthree}
      {\mtyp}}$,
  where $\fctx = \fctx_1 + \fctx_2$ and $\tctx = \tctx_1 + \tctx_2$ and
  $\mcount = \mset{\bBeta} + \mcount_1 + \mcount_2$ and $\tm = \tmtwo \, \tmthree$,
  and we can write $I$ as $J \uplus K$ so that
  $\fset{\mtyptwo_i}_\iI = \fset{\mtyptwo_j}_\jJ + \fset{\mtyptwo_k}_\kK$.
  The judgment is derived from
  $\derivof{\deriv_1}{
    \judgv[\mcount_1]
      {\fctx_1, \var : \fset{\mtyptwo_j}_\jJ}{\tctx_1}
      {\tmtwo}
      {\msetabs{\optmtypthree \to \mtyp}}}$ and
  $\derivof{\deriv_2}{
    \judgv[\mcount_2]
      {\fctx_2, \var : \fset{\mtyptwo_k}_\kK}{\tctx_2}
      {\tmthree}
      {\mtypthree}}$,
  where (1) $\optmtypthree \mleq \mtypthree$.
  We can apply \ih on $\deriv_1$ and $(\derivtwo_j)_\jJ$, yielding
  $\derivof{\deriv'_1}{
    \judgv[\mcount_1 +_\jJ \mcounttwo_j]
      {\fctx_1 +_\jJ \fctxtwo_j}{\tctx_1 +_\jJ \tctxtwo_j}
      {\tmtwo\sub{\var}{\tmsix}}
      {\msetabs{\optmtypthree \to \mtyp}}}$.
  On the other hand, we can apply \ih on $\deriv_2$ and $(\derivtwo_k)_\kK$, yielding
  $\derivof{\deriv'_2}{
    \judgv[\mcount_2 +_\kK \mcounttwo_k]
      {\fctx_2 +_\kK \fctxtwo_k}{\tctx_2 +_\kK \tctxtwo_k}
      {\tmthree\sub{\var}{\tmsix}}
      {\mtypthree}}$.
  Applying rule $\vruleTypApp$ with $\deriv'_1$, $\deriv'_2$, and (1) as premises
  we conclude with
  $\derivof{\deriv'}{
    \judgv[\mcount +_\iI \mcounttwo_i]
      {\fctx +_\iI \fctxtwo_i}{\tctx +_\iI \tctxtwo_i}
      {\tmtwo\sub{\var}{\tmsix}\,\tmthree\sub{\var}{\tmsix}}
      {\mtyp}}$,
  and we are done since $\tmtwo\sub{\var}{\tmsix} \, \tmthree\sub{\var}{\tmsix}
  = (\tmtwo \, \tmthree)\sub{\var}{\tmsix}$.
\item $\vruleTypZero$.
  Then $\derivof{\deriv}{\judgv[\emset]{\emptyctx}{\emptyctx}{\zero}{\msetnat{\zeroTyp}_\jJ}}$,
  where $\fctx, \var : \fset{\mtyptwo_i}_\iI = \emptyctx$ and $\tctx = \emptyctx$
  and $\mcount = \emset$ and $\tm = \zero$ and $\mtyp = \msetnat{\zeroTyp}_\jJ$.
  Hence $\fset{\mtyptwo_i}_\iI = \efset$, then $I = \emptyset$.
  We conclude by letting $\deriv' \defeq \deriv$, since $\zero\sub{\var}{\tmsix} = \zero$.
\item $\vruleTypSucc$.
  Then
  $\derivof{\deriv}{
    \judgv[\mcount]
      {\fctx, \var : \fset{\mtyptwo_i}_\iI}{\tctx}
      {\succ{\tmthree}}
      {\msetnat{\succTyp{\mtypnat_j}}_\jJ}}$,
  where $\tm = \succ{\tmthree}$ and $\mtyp = \msetnat{\succTyp{\mtypnat_j}}_\jJ$.
  The judgment is derived from
  $\derivof{\deriv_0}{
    \judgv[\mcount]
      {\fctx, \var : \fset{\mtyptwo_i}_\iI}{\tctx}{\tmthree}{+_\jJ \mtypnat_j}}$.
  We can apply \ih on $\deriv_0$ with $(\derivtwo_i)_\iI$, yielding
  $\derivof{\deriv'_0}{
    \judgv[\mcount +_\iI \mcounttwo_i]
      {\fctx +_\iI \fctxtwo_i}{\tctx +_\iI \tctxtwo_i}
      {\tmthree\sub{\var}{\tmsix}}
      {+_\jJ \mtypnat_j}}$.
  Applying rule $\vruleTypSucc$, we conclude with
  $\derivof{\deriv'}{
    \judgv[\mcount +_\iI \mcounttwo_i]
      {\fctx +_\iI \fctxtwo_i}{\tctx +_\iI \tctxtwo_i}
      {\succ{\tmthree\sub{\var}{\tmsix}}}
      {\msetnat{\succTyp{\mtypnat_j}}_\jJ}}$,
  and we are done since $\succ{\tmthree\sub{\var}{\tmsix}} 
  = (\succ{\tmthree})\sub{\var}{\tmsix}$.
\item $\vruleTypIfZero$.
  Then
  $\derivof{\deriv}{
    \judgv[\mset{\bIfZero}+\mcount_1+\mcount_2]
      {\fctx_1 + \fctx_2, \var:\fset{\mtyptwo_j}_\jJ + \fset{\mtyptwo_j}_\kK}{\tctx_1 + \tctx_2}
      {\ifz{\tmtwo}{\tmthree}{\vartwo}{\tmfour}}
      {\mtyp}}$
  where $\fctx = \fctx_1+\fctx_2$ and $\tctx = \tctx_1+\tctx_2$ and
  $\mcount = \mset{\bIfZero}+\mcount_1+\mcount_2$ and
  $\tm = \ifz{\tmtwo}{\tmthree}{\vartwo}{\tmfour}$, and
  $I = J \uplus K$ so that $\fset{\mtyptwo_i}_\iI = 
  \fset{\mtyptwo_j}_\jJ + \fset{\mtyptwo_k}_\kK$.
  The judgment is derived from
  $\derivof{\deriv_1}{\judgv[\mcount_1]{\fctx_1,\var:\fset{\mtyptwo_j}_\jJ}{\tctx_1}{\tmtwo}{\msetnat{\zeroTyp}}}$
  and $\derivof{\deriv_2}{\judgv[\mcount_2]{\fctx_2,\var:\fset{\mtyptwo_k}_\kK}{\tctx_2}{\tmthree}{\mtyp}}$.
  By $\alpha$-conversion we may assume $\vartwo \notin \set{\var} \cup \fv{\tmsix}$.
  By \ih on both $\deriv_1$ with $(\derivtwo_j)_\jJ$ and $\deriv_2$ with $(\deriv_k)_\kK$ we have
  $\derivof{\deriv'_1}{
    \judgv[\mcount_1+_\jJ\mcounttwo_j]
      {\fctx_1+_\jJ\fctxtwo_j}{\tctx_1+_\jJ\tctxtwo_j}
      {\tmtwo\sub{\var}{\tmsix}}{\msetnat{\zeroTyp}}}$
  and
  $\derivof{\deriv'_2}{
    \judgv[\mcount_2+_\kK\mcounttwo_k]
      {\fctx_2+_\kK\fctxtwo_k}{\tctx_2+_\kK\tctxtwo_k}
      {\tmthree\sub{\var}{\tmsix}}{\mtyp}}$.
  Applying rule $\vruleTypIfZero$, we conclude with
  $\derivof{\deriv'}{
    \judgv[\mcount+_\iI\mcounttwo_i]
      {\fctx+_\iI\fctxtwo}{\tctx+_\iI\tctxtwo_i}
      {\ifz{\tmtwo\sub{\var}{\tmsix}}{\tmthree\sub{\var}{\tmsix}}{\vartwo}{\tmfour\sub{\var}{\tmsix}}}
      {\mtyp}}$,
  and we are done since $\ifz{\tmtwo\sub{\var}{\tmsix}}{\tmthree\sub{\var}{\tmsix}}{\vartwo}{\tmfour\sub{\var}{\tmsix}}
  = (\ifz{\tmtwo}{\tmthree}{\vartwo}{\tmfour})\sub{\var}{\tmsix}$.
\item $\vruleTypIfSucc$.
  Then
  $\derivof{\deriv}{
    \judgv[\mset{\bIfSucc}+\mcount_1+\mcount_2]
      {\fctx_1+\fctx_2+\var:\fset{\mtyptwo_j}_\jJ + \fset{\mtyptwo_k}_\kK}{\tctx_1+\tctx_2}
      {\ifz{\tmtwo}{\tmthree}{\vartwo}{\tmfour}}
      {\mtyp}}$
  where $\fctx = \fctx_1+\fctx_2$ and $\tctx = \tctx_1+\tctx_2$ and
  $\mcount = \mset{\bIfSucc}+\mcount_1+\mcount_2$ and
  $\tm = \ifz{\tmtwo}{\tmthree}{\vartwo}{\tmfour}$, and $I = J \uplus K$,
  so that $\fset{\mtyptwo_i}_\iI = \fset{\mtyptwo_j}_\jJ + \fset{\mtyptwo_k}_\kK$.
  The judgment is derived from
  $\derivof{\deriv_1}{\judgv[\mcount_1]{\fctx_1,\var:\fset{\mtyptwo_j}_\jJ}{\tctx_1}{\tmtwo}{\msetnat{\succTyp{\mtypnat}}}}$
  and $\derivof{\deriv_2}{\judgv[\mcount_2]{\fctx_2,\var:\fset{\mtyptwo_k}_\kK}{\tctx_2,\vartwo:\optmtypnat}{\tmfour}{\mtyp}}$
  and (1) $\optmtypnat \mleq \mtypnat$.
  By $\alpha$-conversion we may assume $\vartwo \notin \set{\var}\cup\fv{\tmsix}$.
  Then by the \ih on both $\deriv_1$ with $(\derivtwo_j)_\jJ$ and $\deriv_2$ with $(\derivtwo_k)_\kK$ we have
  $\derivof{\deriv'_1}{
    \judgv[\mcount_1+_\jJ\mcounttwo_j]
      {\fctx_1+_\jJ\fctxtwo_j}{\tctx_1+_\jJ\tctxtwo_j}
      {\tmtwo\sub{\var}{\tmsix}}
      {\msetnat{\succTyp{\mtypnat}}}}$
  and
  $\derivof{\deriv'_2}{
    \judgv[\mcount_2+_\kK\mcounttwo_k]
      {\fctx_2+_\kK\fctxtwo_k}{(\tctx_2,\vartwo:\optmtypnat)+_\kK\tctxtwo_k}
      {\tmfour\sub{\var}{\tmsix}}
      {\mtyp}}$. 
  By \cref{lemCBV:relevance} and $\alpha$-conversion
  we can write $(\tctx_2,\vartwo:\optmtypnat)+_\kK\tctxtwo_k$
  as $\tctx_2+_\kK\tctxtwo_k,\vartwo:\optmtypnat$.
  Applying rule $\vruleTypIfSucc$ with $\deriv'_1$, $\deriv'_2$ and (1) as premises
  we conclude with
  $\derivof{\deriv'}{
    \judgv[\mcount+_\iI\mcounttwo_i]
      {\fctx+_\iI\fctxtwo_i}{\tctx+_\iI\tctxtwo_i}
      {\ifz{\tmtwo\sub{\var}{\tmsix}}{\tmthree\sub{\var}{\tmsix}}{\vartwo}{\tmfour\sub{\var}{\tmsix}}}
      {\mtyp}}$.
\item $\vruleTypFix$.
  Then
  $\derivof{\deriv}{
    \judgv[\mset{\bFix}+\mcount_0+_\jJ\mcount_j]
      {\fctx_0 +_\jJ \fctx_j, \var : \fset{\mtyptwo_i}_{i\in I_0} +_\jJ \fset{\mtyptwo_i}_{i\in I_j}}
      {\tctx_0 +_\jJ \tctx_j}
      {\fix{\vartwo}{\tmtwo}}
      {\mtyp}}$
  where $\fctx = \fctx_0 +_\jJ \fctx_j$ and $\tctx = \tctx_0 +_\jJ \tctx_j$ and
  $\mcount = \mset{\bFix} + \mcount_0 +_\jJ \mcount_j$ and $\tm = \fix{\vartwo}{\tmtwo}$, 
  and we can write $I = I_0 \uplus_\jJ I_j$, so that 
  $\fset{\mtyptwo_i}_\iI = \fset{\mtyptwo_i}_{i \in I_0} +_\jJ \fset{\mtyptwo_i}_{i \in I_j}$.
  The judgment is derived from
  $\derivof{\deriv_0}{
    \judgv[\mcount_0]
      {\fctx_0, \var : \fset{\mtyptwo_i}_{i \in I_0}, \vartwo:\fset{\mtypthree_j}_\jJ}{\tctx_0}
      {\tmtwo}{\mtyp}}$ and
  $\derivof{\deriv_j}{
    \judgv[\mcount_j]
      {\fctx_j, \var : \fset{\mtyptwo_i}_{i \in I_j}}{\tctx_j}
      {\fix{\vartwo}{\tmtwo}}{\mtypthree_j}}$ for each $\jJ$.
  By $\alpha$-conversion we may assume $\vartwo \notin \set{\var}\cup\fv{\tmsix}$.
  Then by the \ih on $\deriv_0$ with $(\derivtwo_i)_{i \in I_0}$ and $(\deriv_j)_\jJ$
  with $(\derivtwo_i)_{i \in I_j}$ we have
  $\derivof{\deriv'_0}{
    \judgv[\mcount_0+_{i\in I_0}\mcounttwo_i]
      {(\fctx_0,\vartwo:\fset{\mtypthree_j}_\jJ)+_{i\in I_0} \fctxtwo_i}
      {\tctx_0+_{i\in I_0} \tctxtwo_i}
      {\tmtwo\sub{\var}{\tmsix}}{\mtyp}}$
  and
  $\derivof{\deriv'_j}{
    \judgv[\mcount_j+_{i \in I_j} \mcounttwo_i]
      {\fctx_j+_{i \in I_j}\fctxtwo_i}{\tctx_j+_{i\in I_j} \tctxtwo_i}
      {(\fix{\vartwo}{\tmtwo})\sub{\var}{\tmsix}}{\mtypthree_j}}$
  for each $\jJ$, with $(\fix{\vartwo}{\tmtwo})\sub{\var}{\tmsix} = 
  \fix{\vartwo}{\tmtwo\sub{\var}{\tmsix}}$.
  By \cref{lemCBV:relevance} and $\alpha$-conversion
  we can write $(\fctx_0,\vartwo:\fset{\mtypthree_j}_\jJ)+_{i\in I_0} \fctxtwo_i$
  as $\fctx_0+_{i\in I_0} \fctxtwo_i,\vartwo:\fset{\mtypthree_j}_\jJ$.
  Applying rule $\vruleTypFix$ we conclude with
  $\derivof{\deriv'}{
    \judgv[\mcount +_\iI \mcounttwo_i]
      {\fctx +_\iI \fctxtwo_i}{\tctx +_\iI \tctxtwo_i}
      {\fix{\vartwo}{\tmtwo\sub{\var}{\tmsix}}}{\mtyp}}$.
\end{enumerate}
\end{proof}

\subjectreduction*
% Label lemCBV:subject_reduction

\begin{proof}
By induction on the derivation of $\tm \tov{\rulename} \tm'$.
\begin{itemize}
\item $\vruleToBeta$.
  Then $\tm = (\lam{\var}{\tmtwo}) \, \val \tov{\bBeta} \tmtwo\sub{\var}{\val} = \tm'$.
  The conclusion of $\deriv$ can then only be derived using rule $\vruleTypApp$,
  so $\deriv$ has the form:
  \[
    \indrule{\vruleTypApp}{
      \indrule{\vruleTypAbs}{
        \derivof{\deriv_1}{\judgv[\mcount_1]{\fctx_1}{\tctx_1, \var : \optmtyptwo}{\tmtwo}{\mtyp}}
      }{
        \judgv[\mcount_1]
          {\fctx_1}{\tctx_1}
          {\lam{\var}{\tmtwo}}
          {\msetabs{\optmtyptwo \to \mtyp}}
      }
      (1)\ \optmtyptwo \mleq \mtyptwo
      \HS
      \derivof{\deriv_2}{\judgv[\mcount_2]{\fctx_2}{\tctx_2}{\val}{\mtyptwo}}
    }{
      \judgv[\mset{\bBeta} + \mcount_1 + \mcount_2]
        {\fctx_1 + \fctx_2}{\tctx_1 + \tctx_2}
        {(\lam{\var}{\tmtwo}) \, \val}
        {\mtyp}
    }
  \]
  where $\fctx = \fctx_1 + \fctx_2$ and $\tctx = \tctx_1 + \tctx_2$
  and $\mcount = \mset{\bBeta} + \mcount_1 + \mcount_2$.
  Given (1), we can apply \cref{lemCBV:value_substitution} to $\deriv_1$ with $\deriv_2$,
  yielding
  $\derivof{\deriv'}{
    \judgv[\mcount_1 + \mcount_2]
      {\fctx_1 + \fctx_2}{\tctx_1 + \tctx_2}{\tmtwo\sub{\var}{\val}}{\mtyp}}$,
   and we conclude with $\mcount' = \mcount_1 + \mcount_2$, since 
   $\mset{\bBeta} + \mcount' =\mset{\bBeta} + \mcount_1 + \mcount_2 = \mcount$. 
\item $\vruleToIfZero$.
  Then $\tm = \ifz{\zero}{\tmtwo}{\var}{\tmthree} \tov{\bIfZero} \tmtwo = \tm'$.
  The conclusion of $\deriv$ can then only be derived using rule $\vruleTypIfZero$,
  so $\deriv$ has the form:
  \[
    \indrule{\vruleTypIfZero}{
      \indrule{\vruleTypZero}{
        \emptyPremise
      }{
        \judgv[\emset]
          {\emptyctx}{\emptyctx}
          {\zero}
          {\msetnat{\zeroTyp}}
      }
        \derivof{\deriv'}{\judgv[\mcount']{\fctx}{\tctx}{\tmtwo}{\mtyp}}
 
    }{
      \judgv[\mset{\bIfZero} + \mcount']{\fctx}{\tctx}{\ifz{\zero}{\tmtwo}{\var}{\tmthree}}{\mtyp}
    }
  \]
  where $\mcount = \mset{\bIfZero} + \mcount'$.
  It suffices to note that $\deriv'$ is the needed derivation, so we are done.
\item $\vruleToIfSucc$.
  Then $\tm = \ifz{\succ{\valnat}}{\tmtwo}{\var}{\tmthree} \tov{\bIfSucc}
  \tmthree\sub{\var}{\valnat}= \tm'$.
  The conclusion of $\deriv$ can then only be derived using rule $\vruleTypIfSucc$,
  so $\deriv$ has the form:
  \[
    \indrule{\vruleTypIfSucc}{
      \indrule{\vruleTypSucc}{
        \derivof{\deriv_1}{\judgv[\mcount_1]{\fctx_1}{\tctx_1}{\valnat}{\mtypnat}}
      }{
        \judgv[\mcount_1]{\fctx_1}{\tctx_1}{\succ{\valnat}}{\msetnat{\succTyp{\mtypnat}}}
      }
      \optmtypnat \mleq \mtypnat
      \HS
      \derivof{\deriv_2}{\judgv[\mcount_2]{\fctx_2}{\tctx_2,\var:\optmtypnat}{\tmthree}{\mtyp}}
    }{
      \judgv[\mset{\bIfSucc} + \mcount_1 + \mcount_2]
        {\fctx_1 + \fctx_2}{\tctx_1 + \tctx_2}
        {\ifz{\succ{\valnat}}{\tmtwo}{\var}{\tmthree}}
        {\mtyp}
    }
  \]
  where $\fctx = \fctx_1 + \fctx_2$ and $\tctx = \tctx_1 + \tctx_2$
  and $\mcount = \mset{\bIfSucc} + \mcount_1 + \mcount_2$.
  Applying \cref{lemCBV:value_substitution} on $\deriv_2$ with $\deriv_1$ leads us to
  $\derivof{\deriv'}{
    \judgv[\mcount_1 + \mcount_2 ]
      {\fctx_1 + \fctx_2 }{\tctx_1 + \tctx_2 }
      {\tmthree\sub{\var}{\valnat}}
      {\mtyp}}$,
  where $\mcount' = \mcount_1 + \mcount_2$
  and we conclude since $\mset{\bIfSucc} + \mcount' =\mcount$.
\item $\vruleTypFix$.
  Then $\tm = \fix{\var}{\tmtwo}\tov{\rulename} \tmtwo\sub{\var}{\fix{\var}{\tmtwo}}= \tm'$.
  The conclusion of $\deriv$ can then only be derived by rule $\vruleTypFix$, so
  $\deriv$ has the form:
  \[
    \indrule{\vruleTypFix}{
      \derivof{\deriv_0}{\judgv[\mcount_0]{\fctx_0,\var:\fset{\mtyptwo_i}_\iI}{\tctx_0}{\tmtwo}{\mtyp}}
      \HS
      (\derivof{\deriv_i}{\judgv[\mcount_i]{\fctx_i}{\tctx_i}{\fix{\var}{\tmtwo}}{\mtyptwo_i}})_\iI
    }{
      \judgv[\mset{\bFix} +\mcount_0 +_\iI \mcount_i]
        {\fctx_0 +_\iI \fctx_i}{\tctx_0 +_\iI\tctx_i}
        {\fix{\var}{\tmtwo}}{\mtyp}
    }
  \]
  where $\fctx = \fctx_0 +_\iI\fctx_i$ and $\tctx = \tctx_0 +_\iI\tctx_i$ and
  $\mcount = \mcount_0 +_\iI\mcount_i$.
  By \Cref{lemCBV:substitution} 
  on $\deriv_0$ with $(\deriv_i)_\iI$ we obtain
  $\derivof{\deriv'}{
    \judgv[\mcount_0 +_\iI\mcount_i]
      {\fctx_0 +_\iI\fctx_i}{\tctx_0 +_\iI\tctx_i}
      {\tmtwo\sub{\var}{\fix{\var}{\tmtwo}}}
      {\mtyp}}$. 
  We let $\mcount' = \mcount_0 +_\iI\mcount_i$
  and we conclude since $\mset{\bFix} + \mcount'=
  \mset{\bFix} + \mcount_0 +_\iI\mcount_i = \mcount$.
\item $\vruleToCongAppL$.
  Then $\tm = \tmtwo \, \tmthree \tov{\rulename} \tmtwo' \, \tmthree = \tm'$,
  derived from $\tmtwo \tov{\rulename} \tmtwo'$.
  The conclusion of $\deriv$ can then only be derived by rule $\vruleTypApp$, 
  so $\deriv$ has the form:
  \[
    \indrule{\vruleTypApp}{
      \derivof{\deriv_1}{\judgv[\mcount_1]{\fctx_1}{\tctx_1}{\tmtwo}{\msetabs{\optmtyptwo \to \mtyp}}}
      \HS
      (1)\ \optmtyptwo \mleq \mtyptwo
      \HS
      \derivof{\deriv_2}{\judgv[\mcount_2]{\fctx_2}{\tctx_2}{\tmthree}{\mtyptwo}}
    }{
      \judgv[\mset{\bBeta} + \mcount_1 + \mcount_2]
        {\fctx_1 + \fctx_2}{\tctx_1 + \tctx_2}
        {\tmtwo \, \tmthree}{\mtyp}
    }
  \]
  where $\fctx = \fctx_1 + \fctx_2$ and $\tctx = \tctx_1 + \tctx_2$ and
  $\mcount = \mset{\bBeta} + \mcount_1 + \mcount_2$.
  By the \ih on $\tmtwo \tov{\rulename} \tmtwo'$ and using $\deriv_1$, 
  there exists $\mcount'_1$ such that $\mcount_1 = \mset{\rulename} + \mcount'_1$ and
  $\derivof{\deriv'_1}{\judgv[\mcount'_1]{\fctx_1}{\tctx_1}{\tmtwo'}{\msetabs{\optmtyptwo \to \mtyp}}}$.
  Applying rule $\vruleTypApp$ with $\deriv'_1$, (1) and $\deriv_2$ as premises,
  we conclude with
  $\derivof{\deriv'}{
    \judgv[\mset{\bBeta} + \mcount'_1 + \mcount_2]
      {\fctx_1 + \fctx_2}{\tctx_1 + \tctx_2}{\tmtwo' \, \tmthree}{\mtyp}}$.
  We let $\mcount' = \mset{\bBeta} + \mcount'_1 + \mcount_2$,
  thus $\mset{\rulename} + \mcount' 
  = \mset{\bBeta} + \mcount'_1 + \mcount_2
  = \mset{\bBeta} + \mcount_1 + \mcount_2 
  = \mcount$.
\item $\vruleToCongAppR$.
  Analogous to the previous case.
\item $\vruleToCongSucc$.
  Then $\tm = \succ{\tmtwo} \tov{\rulename} \succ{\tmtwo'} = \tm'$,
  derived from $\tmtwo \tov{\rulename} \tmtwo'$.
  The conclusion of $\deriv$ can then only be derived by rule $\vruleTypSucc$,
  so $\deriv$ has the form:
  \[
    \indrule{\vruleTypSucc}{
      \derivof{\deriv_0}{\judgv[\mcount]{\fctx}{\tctx}{\tmtwo}{+_\iI \mtyptwonat_i}}
    }{
      \judgv[\mcount]
        {\fctx}{\tctx}
        {\succ{\tmtwo}}{\msetnat{\succTyp{\mtyptwonat_i}}_\iI}
    }
  \]
  where $\mtyp = \msetnat{\succTyp{\mtyptwonat_i}}_\iI$.
  By the \ih on $\tmtwo \tov{\rulename} \tmtwo'$ and using $\deriv_0$, 
  there exists $\mcount'$ such that $\mcount = \mset{\rulename} + \mcount'$ and
  $\derivof{\deriv'_0}{\judgv[\mcount']{\fctx}{\tctx}{\tmtwo'}{+_\iI \mtyptwonat_i}}$.
  Applying rule $\vruleTypSucc$ with $\deriv'_0$ as premise,
  we conclude with
  $\derivof{\deriv'}{
    \judgv[\mcount']
      {\fctx}{\tctx}{\succ{\tmtwo'}}{\msetnat{\succTyp{\mtyptwonat_i}}_\iI}}$,
  where
  $\mcount = \mset{\rulename} + \mcount'$.
\item $\vruleToCongIf$.
  Then $\tm = \ifz{\tmtwo}{\tmthree}{\var}{\tmfour}
  \tov{\rulename} \ifz{\tmtwo'}{\tmthree}{\var}{\tmfour}= \tm'$,
  derived from $\tmtwo \tov{\rulename} \tmtwo'$.
  The conclusion of $\deriv$ can be derived either by rule $\vruleTypIfZero$
  or by rule $\vruleTypIfSucc$.
  Let us suppose the first case, as the second one is analogous.
  Then $\deriv$ is of the form:
  \[   
    \indrule{\vruleTypIfZero}{
      \derivof{\deriv_1}{\judgv[\mcount_1]{\fctx_1}{\tctx_1}{\tmtwo}{\msetnat{\zeroTyp}}}
      \HS
      \derivof{\deriv_2}{\judgv[\mcount_2]{\fctx_2}{\tctx_2}{\tmthree}{\mtyp}}
    }{
      \judgv[\mset{\bIfZero} + \mcount_1 + \mcount_2]{\fctx_1+\fctx_2}{\tctx_1 + \tctx_2}{\ifz{\tmtwo}{\tmthree}{\var}{\tmfour}}{\mtyp}
    }
  \]
  where $\fctx = \fctx_1 + \fctx_2$ and $\tctx = \tctx_1 + \tctx_2$ and
  $\mcount = \mset{\bIfZero} + \mcount_1 + \mcount_2$.
  By the \ih on $\tmtwo \tov{\rulename} \tmtwo'$ and using $\deriv_1$,
  there exists $\mcount'_1$ such that $\mcount_1 = \mset{\rulename} + \mcount'_1$ and
  $\derivof{\deriv'_1}{\judgv[\mcount'_1]{\fctx_1}{\tctx_1}{\tmtwo'}{\msetnat{\zeroTyp}}}$.
  Applying rule $\vruleTypIfZero$ with $\deriv'_1$ and $\deriv_2$ as premises,
  we conclude with
  $\derivof{\deriv'}{
    \judgv[\mset{\bIfZero} + \mcount'_1+\mcount_2]
      {\fctx_1+\fctx_2}{\tctx_1 + \tctx_2}
      {\ifz{\tmtwo'}{\tmthree}{\var}{\tmfour}}{\mtyp}}$.
  We let $\mcount' = \mset{\bIfZero} + \mcount'_1 + \mcount_2$,
  thus $\mset{\rulename} + \mcount'
  = \mset{\rulename} + \mset{\bIfZero} + \mcount'_1 + \mcount_2
  = \mset{\bIfZero} + \mcount_1 + \mcount_2
  = \mcount$.
\end{itemize}
\end{proof}

\valueantisubstitution*
% Label lemCBV:value_anti_substitution

\begin{proof}
We first consider the case $\tm = \var$, for which we have 
$\derivof{\deriv'}{\judgv[\mcount]{\fctx}{\tctx}{\val}{\mtyp}}$ as hypothesis.
Then taking $\fctx_1 \defeq \emptyctx$ and $\fctx_2 \defeq \fctx$ and
$\tctx_1 \defeq \emptyctx$ and $\tctx_2 \defeq \tctx$ and $\mcount_1 \defeq \emset$ and 
$\mcount_2 \defeq \mcount$ and $\optmtyptwo \defeq \mtyp$ and $\mtyptwo \defeq \mtyp$,
we conclude that the following hold:
\begin{enumerate}
\item
  $\derivof{\deriv}{\judgv[\emset]{\emptyctx}{\var : \mtyp}{\var}{\mtyp}}$, by rule $\vruleTypVarOne$
\item
  $\derivof{\derivtwo}{\judgv[\mcount]{\fctx}{\tctx}{\val}{\mtyp}}$, by letting $\derivtwo \defeq \deriv'$
\item
  $\fctx = \emptyctx + \fctx = \fctx_1 + \fctx_2$ and
  $\tctx = \emptyctx + \tctx = \tctx_1 + \tctx_2$ and
  $\mcount = \emset + \mcount = \mcount_1 + \mcount_2$ and
  $\mtyp \mleq \mtyp$
\end{enumerate}

If $\tm \neq \var$, then we proceed by induction on $\deriv'$.
\begin{itemize}
\item $\vruleTypVarOne$.
  Then $\tm = \vartwo$, with $\vartwo \neq \var$, and 
  $\derivof{\deriv'}{\judgv[\emset]{\emptyctx}{\vartwo : \mtyp}{\vartwo}{\mtyp}}$,
  where $\fctx = \emptyctx$ and $\tctx = \vartwo : \mtyp$ and $\mcount = \emset$.
  Taking $\fctx_1 = \fctx_2 \defeq \emptyctx$ and
  $\tctx_1 \defeq \vartwo : \mtyp$ and $\tctx_2 \defeq \emptyctx$ and 
  $\mcount_1 = \mcount_2 \defeq \emset$ and $\optmtyptwo \defeq \none$ and 
  $\mtyptwo \defeq \emsetnu{\nature}$ with $\nature$ a proper nature, 
  we conclude with:
  \begin{enumerate}
  \item
    $\derivof{\deriv}{\judgv[\emset]{\emptyctx}{\vartwo : \mtyp, \var : \none}{\vartwo}{\mtyp}}$,
    by letting $\deriv \defeq \deriv'$
  \item
    $\derivof{\derivtwo}{\judgv[\emset]{\emptyctx}{\emptyctx}{\val}{\emsetnu{\nature}}}$, 
    by \cref{lemCBV:typable_values_emptyctx}
  \item
    $\fctx = \emptyctx = \emptyctx + \emptyctx = \fctx_1 + \fctx_2$ and
    $\tctx = (\vartwo : \mtyp) + \emptyctx = \tctx_1 + \tctx_2$ and
    $\mcount = \emset = \emset + \emset = \mcount_1 + \mcount_2$ and
    $\none \mleq \emsetnu{\nature}$.
  \end{enumerate}
\item $\vruleTypVarTwo$.
  Analogous to the previous case.
  % Then $\tm = \vartwo$, with $\vartwo \neq \var$, and
  % $\judgv[\emset]{\vartwo : \fset{\mtyp}}{\emptyctx}{\vartwo}{\mtyp}$,
  % where $\fctx = \vartwo : \fset{\mtyp}$ and $\tctx = \emptyctx$ and $\mcount = \emset$.
  % Taking $\fctx_1 = \vartwo : \fset{\mtyp}$ and $\fctx_2 = \emptyctx$ and
  % $\tctx_1 = \tctx_2 = \emptyctx$ and $\mcount_1 = \mcount_2 = \emset$ and
  % $\optmtyptwo = \none$ and $\mtyptwo = \emsetnu{\nature}$
  % with $\nature$ a proper nature, we conclude that the following hold:
  % \begin{enumerate}
  % \item
  %   $\judgv[\emset]{\vartwo : \fset{\mtyp}}{\var : \none}{\vartwo}{\mtyp}$,
  %   by hypothesis
  % \item
  %   $\judgv[\emset]{\emptyctx}{\emptyctx}{\val}{\emsetnu{\nature}}$, 
  %   by \cref{lemCBV:typable_values_emptyctx}
  % \end{enumerate}
  % and $\fctx = \vartwo : \fset{\mtyp} = (\vartwo : \fset{\mtyp}) + \emptyctx = \fctx_1 + \fctx_2$ and
  % $\tctx = \emptyctx = \emptyctx + \emptyctx = \tctx_1 + \tctx_2$ and
  % $\mcount = \emset = \emset + \emset = \mcount_1 + \mcount_2$ and
  % $\none \mleq \emsetnu{\nature}$.
\item $\vruleTypAbs$.
  Then 
  $\derivof{\deriv'}{
    \judgv[+_\iI \mcount_i]
      {+_\iI \fctx_i}{+_\iI \tctx_i}
      {\lam{\vartwo}{\tmtwo}}
      {\msetabs{\optmtypthree_i \to \mtypfour_i}_\iI}}$,
  with $\fctx = +_\iI \fctx_i$ and $\tctx = +_\iI \tctx_i$ and $\mcount = +_\iI \mcount_i$ and
  $\mtyp = \msetabs{\optmtypthree_i \to \mtypfour_i}_\iI$.
  Then $\tm = \lam{\vartwo}{\tmtwo'}$,
  with $\tmtwo = \tmtwo'\sub{\var}{\val}$.
  The judgment is derived from premises of the form
  $\derivof{\deriv'_i}{\judgv[\mcount_i]{\fctx_i}{\tctx_i, \vartwo : \optmtypthree_i}{\tmtwo}{\mtypfour_i}}$
  for each $\iI$.
  We can apply the \ih on $(\deriv_i)_\iI$, yielding that for each $\iI$ there exist
  family contexts $\fctx_{i1}$, $\fctx_{i2}$,
  typing contexts $\tctx_{i1}$, $\tctx_{i2}$, 
  multi-counters $\mcount_{i1}$, $\mcount_{i2}$,
  an optional multitype $\optmtyptwo_i$ and
  a multitype $\mtyptwo_i$ such that:
  \begin{enumerate}
  \item[i.1.]
    $(\derivof{\deriv_i}{
      \judgv[\mcount_{i1}]
        {\fctx_{i1}}{\tctx_{i1}, \vartwo : \optmtypthree_i, \var : \optmtyptwo_i}
        {\tmtwo'}
        {\mtypfour_i}})_\iI$
  \item[i.2.]
    $(\derivof{\derivtwo_i}{\judgv[\mcount_{i2}]{\fctx_{i2}}{\tctx_{i2}}{\val}{\mtyptwo_i}})_\iI$
  \item[i.3.]
    $\fctx_i = \fctx_{i1} + \fctx_{i2}$ and 
    $\tctx_i, \vartwo : \optmtypthree_i = (\tctx_{i1}, \vartwo : \optmtypthree_i) + \tctx_{i2}$ and
    $\mcount_i = \mcount_{i1} + \mcount_{i2}$ and 
    $\optmtyptwo_i \mleq \mtyptwo_i$.
  \end{enumerate}
  For each $\iI$, it is not necessary to split
  $\vartwo : \optmtypthree_i$, as we may assume $\vartwo \notin \fv{\val}$
  by $\alpha$-conversion, so $\var \notin \dom{\tctx_{i2}}$ by \cref{lemCBV:relevance}.
  Taking $\fctx_1 \defeq +_\iI \fctx_{i1}$ and $\fctx_2 \defeq +_\iI \fctx_{i2}$ and
  $\tctx_1 \defeq +_\iI \tctx_{i1}$ and $\tctx_2 \defeq +_\iI \tctx_{i2}$ and
  $\mcount_1 \defeq +_\iI \mcount_{i1}$ and $\mcount_2 \defeq +_\iI \mcount_{i2}$ and
  $\optmtyptwo \defeq +_\iI \optmtyptwo_i$ and $\mtyptwo \defeq +_\iI \mtyptwo_i$
  we conclude with:
  \begin{enumerate}
  \item
    $\derivof{\deriv}{
      \judgv[+_\iI \mcount_{i1}]
        {+_\iI \fctx_{i1}}{+_\iI \tctx_{i1}, \var : +_\iI \mtyptwo_i}
        {\lam{\vartwo}{\tmtwo'}}
        {\msetabs{\optmtypthree_i \to \mtypfour_i}_\iI}}$, 
    by applying rule $\vruleTypAbs$ to $(\deriv_i)_\iI$
  \item
    $\derivof{\derivtwo}{
      \judgv[+_\iI \mcount_{i2}]
        {+_\iI \fctx_{i2}}{+_\iI \tctx_{i2}}
        {\val}
        {+_\iI \mtyptwo_i}}$, 
    as the result of applying \cref{lemCBV:generalized_value_splitting_merging} to $(\derivtwo_i)_\iI$
  \item
    $\fctx = +_\iI \fctx_i = +_\iI (\fctx_{i1} + \fctx_{i2}) = \fctx_1 + \fctx_2$ and
    $\tctx = +_\iI \tctx_i = +_\iI (\tctx_{i1} + \tctx_{i2}) = \tctx_1 + \tctx_2$ and
    $\mcount = +_\iI \mcount_i = +_\iI (\mcount_{i1} + \mcount_{i2}) = \mcount_1 + \mcount_2$ and
    $+_\iI \optmtyptwo_i \mleq +_\iI \mtyptwo_i$ by \cref{lemCBV:generalized_value_splitting_merging}
  \end{enumerate}
\item $\vruleTypApp$.
  Then 
  $\derivof{\deriv'}{
    \judgv[\mset{\bBeta} + \mcount_a + \mcount_b]
      {\fctx_a + \fctx_b}
      {\tctx_a + \tctx_b}
      {\tmtwo \, \tmthree}
      {\mtyp}
  }$, where $\fctx = \fctx_a + \fctx_b$ and $\tctx = \tctx_a + \tctx_b$ and 
  $\mcount = \mset{\bBeta} + \mcount_a + \mcount_b$.
  Since $\tm \neq \var$, then 
  $\tm = \tmtwo' \, \tmthree'$ with $\tmtwo = \tmtwo'\sub{\var}{\val}$
  and $\tmthree = \tmthree'\sub{\var}{\val}$.
  The judgment is derived from
  $\derivof{\deriv'_a}{\judgv[\mcount_a]{\fctx_a}{\tctx_a}{\tmtwo}{\msetabs{\optmtypthree \to \mtyp}}}$ and
  $\derivof{\deriv'_b}{\judgv[\mcount_b]{\fctx_b}{\tctx_b}{\tmthree}{\mtypthree}}$,
  where $\optmtypthree \mleq \mtypthree$.
  We can apply \ih on $\deriv_a$, yielding that there exist family contexts
  $\fctx_{a1}$, $\fctx_{a2}$, typing contexts $\tctx_{a1}$, $\tctx_{a2}$, 
  multi-counters $\mcount_{a1}$, $\mcount_{a2}$, an optional multitype $\optmtyptwo_a$ 
  and a multitype $\mtyptwo_a$ satisfying:
  \begin{enumerate}
  \item[a.1.]
    $\derivof{\deriv_a}{\judgv[\mcount_{a1}]{\fctx_{a1}}{\tctx_{a1}, \var : \optmtyptwo_a}{\tmtwo'}{\msetabs{\optmtypthree \to \mtyp}}}$
  \item[a.2.]
    $\derivof{\derivtwo_a}{\judgv[\mcount_{a2}]{\fctx_{a2}}{\tctx_{a2}}{\val}{\mtyptwo_a}}$
  \item[a.3.]
    $\fctx_a = \fctx_{a1} + \fctx_{a2}$ and 
    $\tctx_a = \tctx_{a1} + \tctx_{a2}$ and 
    $\mcount_1 = \mcount_{a1} + \mcount_{a2}$ 
    and $\optmtyptwo_a \mleq \mtyptwo_a$
  \end{enumerate}
  On the other hand we can apply \ih on $\deriv_b$, yielding that there exist
  family contexts $\fctx_{b1}$, $\fctx_{b2}$, 
  typing contexts $\tctx_{b1}$, $\tctx_{b2}$, 
  multi-counters $\mcount_{b1}$, $\mcount_{b2}$,
  an optional multitype $\optmtyptwo_b$ and a multitype $\mtyptwo_b$ 
  satisfying: 
  \begin{enumerate}
  \item[b.1.]
    $\derivof{\deriv_b}{\judgv[\mcount_{b1}]{\fctx_{b1}}{\tctx_{b1}, \var : \optmtyptwo_b}{\tmthree'}{\mtypthree}}$
  \item[b.2.]
    $\derivof{\derivtwo_b}{\judgv[\mcount_{b2}]{\fctx_{b2}}{\tctx_{b2}}{\val}{\mtyptwo_b}}$
  \item[b.3.]
    $\fctx_b = \fctx_{b1} + \fctx_{b2}$ and 
    $\tctx_b = \tctx_{b1} + \tctx_{b2}$ and
    $\mcount_b = \mcount_{b1} + \mcount_{b2}$ and 
    $\optmtyptwo_b \mleq \mtyptwo_b$
  \end{enumerate}
  Then, taking $\fctx_1 \defeq \fctx_{a1} + \fctx_{b1}$ and 
  $\fctx_2 \defeq \fctx_{a2} + \fctx_{b2}$ and 
  $\tctx_1 \defeq \tctx_{a1} + \tctx_{b1}$ and $\tctx_2 \defeq \tctx_{a2} + \tctx_{b2}$ and
  $\mcount_1 \defeq \mset{\bBeta} + \mcount_{a1} + \mcount_{b1}$ and
  $\mcount_2 \defeq \mcount_{a2} + \mcount_{b2}$ and
  $\optmtyptwo \defeq \optmtyptwo_a + \optmtyptwo_b$ and 
  $\mtyptwo \defeq \mtyptwo_a + \mtyptwo_b$,
  we conclude that the following hold:
  \begin{enumerate}
  \item
    $\derivof{\deriv}{\judgv[\mcount_1]{\fctx_1}{\tctx_1, \var : \optmtyptwo}{\tmtwo' \, \tmthree'}{\mtyp}}$, 
    by applying rule $\vruleTypApp$ with $\deriv_a$ and $\deriv_b$ as premises,
    given that $\optmtypthree \mleq \mtypthree$ holds
  \item
    $\derivof{\derivtwo}{\judgv[\mcount_2]{\fctx_2}{\tctx_2}{\val}{\mtyptwo}}$, 
    as the result of applying \cref{lemCBV:value_splitting_merging}
    to $\derivtwo_a$ and $\derivtwo_b$
  \item
    It is easy to check that the remaining conditions hold.
  \end{enumerate}
\item $\vruleTypZero$.
  Analogous to case $\vruleTypVarOne$.
  % Then the judgment is of the form 
  % $\judgv[\emset]{\emptyctx}{\emptyctx}{\zero}{\msetnat{\zeroTyp}_\iI}$,
  % where $\fctx = \emptyctx$, $\tctx = \emptyctx$, $\mcount = \emset$, $\mtyp = \msetnat{\zeroTyp}_\iI$, and
  % since $\tm \neq \var$ we have $\tm = \zero$.
  % Taking $\fctx_1 = \fctx_2 = \emptyctx$, $\tctx_1 = \tctx_2 = \emptyctx$, 
  % $\mcount_1 = \mcount_2 = \emset$, $\optmtyptwo = \none$ 
  % and $\mtyptwo = \emsetnu{\nature}$, we conclude that the following hold:
  % \begin{enumerate}
  % \item
  %   $\judgv[\emset]{\emptyctx}{\emptyctx, \var : \none}{\zero}{\msetnat{\zeroTyp}_\iI}$,
  %   by rule $\vruleTypZero$
  % \item
  %   $\judgv[\mcount]{\emptyctx}{\emptyctx}{\val}{\emsetnu{\nature}}$, 
  %   by \cref{lemCBV:typable_values_emptyctx}.
  % \end{enumerate}
\item $\vruleTypSucc$.
  Then
  $\derivof{\deriv'}{\judgv[\mcount]{\fctx}{\tctx}{\succ{\tmthree}}{\msetnat{\succTyp{\mtypnat_i}}_\iI}}$,
  where $\mtyp = \msetnat{\succTyp{\mtypnat_i}}_\iI$. 
  Since $\tm \neq \var$, then $\tm = \succ{\tmthree'}$, 
  with $\tmthree = \tmthree'\sub{\var}{\val}$.
  The judgment is derived from
  $\derivof{\deriv'_0}{\judgv[\mcount]{\fctx}{\tctx}{\tmthree}{\mtypnat}}$,
  where $\mtypnat = +_\iI \mtypnat_i$.
  We can apply \ih on $\deriv'_0$, yielding that there exist
  family contexts $\fctx_1$, $\fctx_2$,
  typing contexts $\tctx_1$, $\tctx_2$,
  multi-counters $\mcount_1$, $\mcount_2$, 
  an optional multitype $\optmtyptwo$ and a multitype $\mtyptwo$ 
  such that:
  \begin{enumerate}
  \item[1'.]
    $\derivof{\deriv_0}{\judgv[\mcount_1]{\fctx_1}{\tctx_1, \var : \optmtyptwo}{\tmthree'}{\mtypnat}}$,
    where $\mtypnat = +_\iI \mtypnat_i$
  \item[2'.]
    $\derivof{\derivtwo_0}{\judgv[\mcount_2]{\fctx_2}{\tctx_2}{\val}{\mtyptwo}}$
  \item[3'.]
    $\fctx = \fctx_1 + \fctx_2$ and 
    $\tctx = \tctx_1 + \tctx_2$ and
    $\mcount = \mcount_1 + \mcount_2$ and 
    $\optmtyptwo \mleq \mtyptwo$
  \end{enumerate}
  By taking $\fctx_1$, $\fctx_2$, $\tctx_1$, $\tctx_2$, $\mcount_1$, $\mcount_2$,
  $\optmtyptwo$ and $\mtyptwo$ from the \ih, we conclude that 
  $\derivof{\deriv}{
    \judgv[\mcount_1]
      {\fctx_1}{\tctx_1, \var : \optmtyptwo}
      {\succ{\tmthree'}}
      {\msetnat{\succTyp{\mtypnat_i}}_\iI}}$, 
  by applying rule $\vruleTypSucc$ to $\deriv_0$.
  The remaining conditions to met already hold by the \ih,
  letting $\derivtwo \defeq \derivtwo_0$.
\item $\vruleTypIfZero$.
  Analogous to case $\vruleTypApp$.
\item $\vruleTypIfSucc$.
  Then
  $\derivof{\deriv'}{
    \judgv[\mset{\bIfSucc} + \mcount_a + \mcount_b]
      {\fctx_a + \fctx_b}{\tctx_a + \tctx_b}
      {\ifz{\tmtwo}{\tmthree}{\vartwo}{\tmfour}}
      {\mtyp}}$,
  where 
  $\fctx = \fctx_a + \fctx_b$ and $\tctx = \tctx_a + \tctx_b$ and
  $\mcount = \mset{\bIfSucc} + \mcount_a + \mcount_b$.
  Then $\tm = \ifz{\tmtwo'}{\tmthree'}{\vartwo}{\tmfour'}$, 
  where $\tmtwo = \tmtwo'\sub{\var}{\val}$, $\tmthree = \tmthree'\sub{\var}{\val}$
  and $\tmfour = \tmfour'\sub{\var}{\val}$.
  The judgment is derived from 
  $\derivof{\deriv'_a}{\judgv[\mcount_a]{\fctx_a}{\tctx_a}{\tmtwo}{\msetnat{\succTyp{\mtypnat}}}}$ and
  $\derivof{\deriv'_b}{\judgv[\mcount_b]{\fctx_b}{\tctx_b, \vartwo : \optmtypnat}{\tmfour}{\mtyp}}$,
  with $\optmtypnat \mleq \mtypnat$.
  We can apply \ih on $\deriv'_a$, yielding that there exist 
  family contexts $\fctx_{a1}$, $\fctx_{a2}$,
  typing contexts $\tctx_{a1}$, $\tctx_{a2}$,
  multi-counters $\mcount_{a1}$, $\mcount_{a2}$,
  an optional multitype $\optmtyptwo_a$ and a multitype $\mtyptwo_a$ such that:
  \begin{enumerate}
  \item[a.1.]
    $\derivof{\deriv_a}{\judgv[\mcount_{a1}]{\fctx_{a1}}{\tctx_{a1}, \var : \optmtyptwo_a}{\tmtwo'}{\msetnat{\succTyp{\mtypnat}}}}$
  \item[a.2.]
    $\derivof{\derivtwo_a}{\judgv[\mcount_{a2}]{\fctx_{a2}}{\tctx_{a2}}{\val}{\mtyptwo_a}}$
  \item[a.3.]
    $\fctx_a = \fctx_{a1} + \fctx_{a2}$ and 
    $\tctx_a = \tctx_{a1} + \tctx_{a2}$ and
    $\mcount_a = \mcount_{a1} + \mcount_{a2}$ and $\optmtyptwo_a \mleq \mtyptwo_a$
  \end{enumerate}
  And we can apply the \ih on $\deriv'_b$, yielding that there exist 
  family contexts $\fctx_{b1}$, $\fctx_{b2}$,
  typing contexts $\tctx_{b1}$, $\tctx_{b2}$,
  multi-counters $\mcount_{b1}$, $\mcount_{b2}$,
  an optional multitype $\optmtyptwo_b$ and a multitype $\mtyptwo_b$ such that:
  \begin{enumerate}
  \item[b.1.]
    $\derivof{\deriv_b}{\judgv[\mcount_{b1}]{\fctx_{b1}}{\tctx_{b1}, \vartwo : \optmtypnat, \var : \optmtyptwo_b}{\tmfour'}{\mtyp}}$
  \item[b.2.]
    $\derivof{\derivtwo_b}{\judgv[\mcount_{b2}]{\fctx_{b2}}{\tctx_{b2}}{\val}{\mtyptwo_b}}$
  \item[b.3.]
    $\fctx_b = \fctx_{b1} + \fctx_{b2}$ and $\tctx_b = \tctx_{b1} + \tctx_{b2}$ and 
    $\mcount_b = \mcount_{b1} + \mcount_{b2}$ and $\optmtyptwo_b \mleq \mtyptwo_b$.
  \end{enumerate}
  We are allowed to split
  $\tctx_b, \vartwo : \optmtypnat$ into $(\tctx_{b1}, \vartwo : \optmtypnat)$
  and $\tctx_{b2}$ by \cref{lemCBV:relevance},
  since by $\alpha$-conversion we can assume $\vartwo \notin \fv{\val}$.
  Taking $\fctx_1 \defeq \fctx_{a1} + \fctx_{b1}$ and 
  $\fctx_2 \defeq \fctx_{a2} + \fctx_{b2}$ and 
  $\tctx_1 \defeq \tctx_{a1} + \tctx_{b1}$ and 
  $\tctx_2 \defeq \tctx_{a2} + \tctx_{b2}$ and 
  $\mcount_1 \defeq \mset{\bIfSucc} + \mcount_{a1} + \mcount_{b1}$ and 
  $\mcount_2 \defeq \mcount_{a2} + \mcount_{b2}$ and
  $\optmtyptwo \defeq \optmtyptwo_a + \optmtyptwo_b$ and 
  $\mtyptwo \defeq \mtyptwo_a + \mtyptwo_b$,
  we conclude that the following hold:
  \begin{enumerate}
  \item
    $\derivof{\deriv}{
      \judgv[\mcount_1]
        {\fctx_1}{\tctx_1, \var : \optmtyptwo}
        {\ifz{\tmtwo'}{\vartwo}{\tmthree'}{\tmfour'}}
        {\mtyp}}$,
    by applying rule $\vruleTypIfSucc$ to $\deriv_a$ and $\deriv_b$,
    given that $\optmtypnat \mleq \mtypnat$ holds
  \item
    $\derivof{\derivtwo}{\judgv[\mcount_2]{\fctx_2}{\tctx_2}{\val}{\mtyptwo}}$,
    as the result of applying \cref{lemCBV:value_splitting_merging} 
    to $\derivtwo_a$ and $\derivtwo_b$
  \item
    $\mcount = \mset{\bIfSucc} + \mcount_a + \mcount_b = \mcount_1 + \mcount_2$, 
    and it is easy to check that the remaining conditions hold
  \end{enumerate}
\item $\vruleTypFix$.
  Then
  $\derivof{\deriv'}{
    \judgv[\mset{\bFix} + \mcount_0 +_\iI \mcount_i]
      {\fctx_0 +_\iI \fctx_i}
      {\tctx_0 +_\iI \tctx_i}
      {\fix{\vartwo}{\tmtwo}}
      {\mtyp}
  }$, where $\fctx = \fctx_0 +_\iI \fctx_i$ and $\tctx = \tctx_0 +_\iI \tctx_i$ 
  and $\mcount = \mset{\bFix} + \mcount_0 +_\iI \mcount_i$.
  Then $\tm = \fix{\vartwo}{\tmtwo'}$, 
  where $\tmtwo = \tmtwo'\sub{\var}{\val}$.
  The judgment is derived from 
  $\derivof{\deriv'_0}{
    \judgv[\mcount_0]
      {\fctx_0, \vartwo : \fset{\mtypthree_i}_\iI}{\tctx_0}
      {\tmtwo}{\mtyp}}$ and
  $(\derivof{\deriv'_i}{
    \judgv[\mcount_i]
      {\fctx_i}{\tctx_i}
      {\fix{\vartwo}{\tmtwo}}
      {\mtypthree_i}})_\iI$.
  We can then apply the \ih on $\deriv'_0$, yielding that then there exist
  family contexts $\fctx_{01}$, $\fctx_{02}$,
  typing contexts $\tctx_{01}$, $\tctx_{02}$,
  multi-counters $\mcount_{01}$, $\mcount_{02}$,
  an optional multitype $\optmtyptwo_0$ and a multitype $\mtyptwo_0$ such that:
  \begin{enumerate}
  \item[1'.]
    $\derivof{\deriv_0}{
      \judgv[\mcount_{01}]
        {\fctx_{01}, \vartwo : \fset{\mtypthree_i}_\iI}
        {\tctx_{01}, \var : \optmtyptwo_0}
        {\tmtwo'}
        {\mtyp}}$
  \item[2'.]
    $\derivof{\derivtwo_0}{\judgv[\mcount_{02}]{\fctx_{02}}{\tctx_{02}}{\val}{\mtyptwo_0}}$
  \item[3'.]
    $\fctx_0 = \fctx_{01} + \fctx_{02}$ and
    $\tctx_0 = \tctx_{01} + \tctx_{02}$ and 
    $\mcount_0 = \mcount_{01} + \mcount_{02}$ and $\optmtyptwo_0 \mleq \mtyptwo_0$.
  \end{enumerate}
  Note that it is not necessary to split $\fset{\mtypthree_i}_\iI$ since
  we may assume $\vartwo \notin \fv{\val}$ by $\alpha$-conversion,
  so $\vartwo \notin \dom{\fctx_{02}}$ by \cref{lemCBV:relevance}.
  The conclusions of $\deriv'_i$ for each $\iI$ are smaller than $\deriv'$,
  hence we can apply the \ih on $(\deriv'_i)_\iI$,
  yielding that for each $\iI$ there exist
  family contexts $(\fctx_{i1})_\iI$, $(\fctx_{i2})_\iI$,
  typing contexts $(\tctx_{i1})_\iI$, $(\tctx_{i2})_\iI$,
  multi-counters $(\mcount_{i1})_\iI$, $(\mcount_{i2})_\iI$,
  an optional multitype $(\optmtyptwo_i)_\iI$ and a multitype $(\mtyptwo_i)_\iI$ 
  such that:
  \begin{enumerate}
  \item[i.1.]
    $\derivof{\deriv_i}{
      \judgv[\mcount_{i1}]
        {\fctx_{i1}}{\tctx_{i1}, \var : \optmtyptwo_i}
        {\fix{\vartwo}{\tmtwo'}}
        {\mtypthree_i}}$
  \item[i.2.]
    $\derivof{\derivtwo_i}{\judgv[\mcount_{i2}]{\fctx_{i2}}{\tctx_{i2}}{\val}{\mtyptwo_i}}$
  \item[i.3.]
    $\fctx_i = \fctx_{i1} + \fctx_{i2}$ and
    $\tctx_i = \tctx_{i1} + \tctx_{i2}$ and 
    $\mcount_i = \mcount_{i1} + \mcount_{i2}$ and 
    $\optmtyptwo_i \mleq \mtyptwo_i$
  \end{enumerate}
  Taking 
  $\fctx_1 \defeq \fctx_{01} +_\iI \fctx_{i1}$ and $\fctx_2 \defeq \fctx_{02} +_\iI \fctx_{i2}$ and
  $\tctx_1 \defeq \tctx_{01} +_\iI \tctx_{i1}$ and $\tctx_2 \defeq \tctx_{02} +_\iI \tctx_{i2}$ and
  $\mcount_1 \defeq \mset{\bFix} + \mcount_{01} +_\iI \mcount_{i1}$ and 
  $\mcount_2 \defeq \mcount_{02} +_\iI \mcount_{i2}$ and
  $\optmtyptwo \defeq \optmtyptwo_0 +_\iI \optmtyptwo_i$, and 
  $\mtyptwo \defeq \mtyptwo_0 +_\iI \mtyptwo_i$,
  we conclude that the following hold:
  \begin{enumerate}
  \item
    $\derivof{\deriv}{\
    \judgv[\mcount_1]
      {\fctx_1}{\tctx_1, \var : \optmtyptwo}
      {\fix{\vartwo}{\tmtwo'}}
      {\mtyp}}$,
    by applying rule $\vruleTypFix$ to $\deriv_0$ and $(\deriv_i)_\iI$
  \item
    $\derivof{\derivtwo}{\judgv[\mcount_2]{\fctx_2}{\tctx_2}{\val}{\mtyptwo}}$,
    as the result of applying \cref{lemCBV:generalized_value_splitting_merging} 
    to $\derivtwo_0$ and $(\derivtwo_i)_\iI$
  \item
    It is easy to check that the remainder conditions hold.
  \end{enumerate}
\end{itemize}
\end{proof}

\antisubstitution*
% Label lemCBV:anti_substitution

\begin{proof}
We first consider the case $\tm = \var$, for which we have 
$\derivof{\deriv'}{\judgv[\mcount]{\fctx}{\tctx}{\tmsix}{\mtyp}}$ as hypothesis.
Taking a set $I$ such that its cardinality is 1, and
$\fctx_0 \defeq \emptyctx$, $\fctxtwo \defeq \fctx$ and
$\tctx_0 \defeq \emptyctx$, $\tctxtwo \defeq \tctx$ and
$\mcount_0 \defeq \emset$, $\mcounttwo \defeq \mcount$ and $\mtyptwo \defeq \mtyp$,
then the following hold:
\begin{enumerate}
\item
  $\derivof{\deriv}{\judgv[\emset]{\var : \fset{\mtyp}}{\emptyctx}{\var}{\mtyp}}$,
  by rule $\vruleTypVarTwo$
\item
  $\derivof{\derivtwo_i}{\judgv[\mcount]{\fctx}{\tctx}{\tmsix}{\mtyp}}$, by hypothesis,
  since $I$ is a singleton, and we let $\derivtwo_i \defeq \deriv'$
\item
  $\fctx = \emptyctx + \fctx = \fctx_0 + \fctxtwo$ and
  $\tctx = \emptyctx + \tctx = \tctx_0 + \tctxtwo$ and
  $\mcount = \emset + \mcount = \mcount_0 + \mcounttwo$
\end{enumerate}
If $\tm \neq \var$ then we proceed by induction on $\deriv'$.
\begin{itemize}
\item $\vruleTypVarOne$.
  Then $\tm = \vartwo$, with $\vartwo \neq \var$, and
  $\derivof{\deriv'}{\judgv[\emset]{\emptyctx}{\vartwo : \mtyp}{\vartwo}{\mtyp}}$,
  with $\fctx = \emptyctx$ and $\tctx = \vartwo : \mtyp$ and $\mcount = \emset$.
  Taking $I \defeq \emptyctx$ and $\fctx_0 \defeq \fctx$ and $\tctx_0 \defeq \tctx$ 
  and $\mcount_0 \defeq \mcount$, we conclude with
  $\derivof{\deriv}{\judgv[\emset]{\emptyctx}{\vartwo : \mtyp}{\vartwo}{\mtyp}}$
  by letting $\deriv \defeq \deriv'$.
  \cref{it:anti_substitution_2} and \cref{it:anti_substitution_3} 
  hold trivially since $I = \emptyset$.
\item $\vruleTypVarTwo$.
  Analogous to the previous case.
  % Then $\tm = \vartwo$, with $\vartwo \neq \var$, and
  % $\judgv[\emset]{\vartwo : \fset{\mtyp}}{\emptyctx}{\vartwo}{\mtyp}$,
  % with $\fctx = \vartwo : \fset{\mtyp}$ and $\tctx = \emptyctx$ and $\mcount = \emset$.
  % Taking $I = \emptyctx$ 
  % and $\fctx_0 = \fctx$, $\tctx_0 = \tctx$ and $\mcount_0 = \mcount$,
  % we met the required conditions, and the judgment
  % $\judgv[\emset]{\vartwo : \fset{\mtyp}}{\emptyctx}{\vartwo}{\mtyp}$
  % from \cref{it:anti_substitution_1} holds by hypothesis.
  % Note that \cref{it:anti_substitution_2} holds trivially
  % since $I = \emptyset$.
\item $\vruleTypAbs$.
  Then
  $\derivof{\deriv'}{
    \judgv[+_\jJ \mcount_j]
      {+_\jJ \fctx_j}{+_\jJ \tctx_j}
      {\lam{\vartwo}{\tmtwo}}
      {\msetabs{\optmtypthree_j \to \mtypfour_j}_\jJ}}$,
  with $\fctx = +_\jJ \fctx_j$ and $\tctx = +_\jJ \tctx_j$ 
  and $\mcount = +_\jJ \mcount_j$ 
  and $\mtyp = \msetabs{\optmtypthree_j \to \mtypfour_j}_\jJ$.
  Then $\tm = \lam{\vartwo}{\tmtwo'}$,
  with $\tmtwo = \tmtwo'\sub{\var}{\tmsix}$.
  The judgment is derived from premises of the form
  $\derivof{\deriv'_j}{
    \judgv[\mcount_j]
      {\fctx_j}{\tctx_j, \vartwo : \optmtypthree_j}
      {\tmtwo}
      {\mtypfour_j}}$ for each $\jJ$.
  We can apply the \ih on $(\deriv'_j)_\jJ$, yielding that
  for each $\jJ$ there exist
  a finite set $I_j$,
  family contexts $\fctx_{j0}$ and $(\fctxtwo_k)_{k \in I_j}$,
  typing contexts $\tctx_{j0}$ and $(\tctxtwo_k)_{k \in I_j}$,
  multi-counters and $\mcount_{j0}$ and $(\mcounttwo_k)_{k \in I_j}$ 
  and multitypes $(\mtyptwo_k)_{k \in I_j}$
  such that:
  \begin{enumerate}
  \item[j.1.]
    $\derivof{\deriv_j}{
      \judgv[\mcount_{j0}]
        {\fctx_{j0}, \var : \fset{\mtyptwo_k}_{k \in I_j}}{\tctx_{j0}, \vartwo : \optmtypthree_j}
        {\tmtwo'}
        {\mtypfour_j}}$
  \item[j.2.]
    $(\derivof{\derivtwo_k}{\judgv[\mcounttwo_k]{\fctxtwo_k}{\tctxtwo_k}{\tmsix}{\mtyptwo_k}})_{k \in I_j}$
  \item[j.3.]
    $\fctx_j = \fctx_{j0} +_{k \in I_j} \fctxtwo_k$ and
    $\tctx_j, \vartwo : \optmtypthree_j 
    = \tctx_{j0}, \vartwo : \optmtypthree_j +_{k \in I_j} \tctxtwo_k$ and
    $\mcount_j = \mcount_{j0} +_{k \in I_j} \mcounttwo_k$.
  \end{enumerate}
  Note that for each $\jJ$ 
  it is not necessary to split $\optmtypthree_j$ since we
  may assume $\vartwo \notin \fv{\tmsix}$
  by $\alpha$-conversion, so $\vartwo \notin \dom{+_{k \in I_j} \tctxtwo_k}$ 
  by \cref{lemCBV:relevance}.
  Taking $I \defeq \uplus_\jJ I_j$ and
  $\fctx_0 \defeq +_\jJ \fctx_{j0}$ and
  $\fctxtwo_i \defeq +_{k \in I_j} \fctxtwo_k$ and
  $\tctx_0 \defeq +_\jJ \tctx_{j0}$ and
  $\tctxtwo_i \defeq +_{k \in I_j} \tctxtwo_k$ and
  $\mcount_0 \defeq +_\jJ \mcount_{j0}$ and
  $\mcounttwo_i \defeq +_{k \in I_j} \mcounttwo_k$ and
  $\mtyptwo_i \defeq +_{k \in I_j} \mtyptwo_k$, we have:
  \begin{enumerate}
  \item
    $\derivof{\deriv}{
      \judgv[\mcount_0]
        {\fctx_0, \var : \fset{\mtyptwo_i}_\iI}{\tctx_0}
        {\lam{\vartwo}{\tmtwo'}}
        {\msetabs{\optmtypthree_j \to \mtypfour_j}_\jJ}}$,
    by applying rule $\vruleTypAbs$ to $(\deriv_j)_\jJ$
  \item
    $(\derivof{\derivtwo_i}{\judgv[\mcounttwo_i]{\fctxtwo_i}{\tctxtwo_i}{\tmsix}{\mtyptwo_i}})_\iI$,
    as the result of applying \cref{lemCBV:generalized_value_splitting_merging} 
    to $(\derivtwo_k)_{k \in I_j}$ for each $\jJ$
  \item
    $\fctx = +_\jJ \fctx_j = +_\jJ (\fctx_{j0} +_{k \in I_j} \fctxtwo_k) = \fctx_0 +_\iI \fctxtwo_i$ and
    $\tctx = +_\jJ \tctx_j = +_\jJ (\tctx_{j0} +_{k \in I_j} \tctxtwo_k) = \tctx_0 +_\iI \tctxtwo_i$ and
    $\mcount = +_\jJ \mcount_j = +_\jJ (\mcount_{j0} +_{k \in I_j} \mcounttwo_k) = \mcount_0 +_\iI \mcounttwo_i$
  \end{enumerate}
\item $\vruleTypApp$.
  Then
  $\derivof{\deriv'}{
    \judgv[\mset{\bBeta} + \mcount_a + \mcount_b]
      {\fctx_a + \fctx_b}{\tctx_a + \tctx_b}
      {\tmtwo \, \tmthree}
      {\mtyp}}$,
  with $\fctx = \fctx_a + \fctx_b$ and $\tctx = \tctx_a + \tctx_b$ and
  $\mcount = \mset{\bBeta} + \mcount_a + \mcount_b$.
  Then $\tm = \tmtwo' \, \tmthree'$,
  with $\tmtwo \, \tmthree = (\tmtwo' \, \tmthree')\sub{\var}{\tmsix}$.
  The judgment is derived from
  $\derivof{\deriv'_a}{\judgv[\mcount_a]{\fctx_a}{\tctx_a}{\tmtwo}{\msetabs{\optmtypthree \to \mtyp}}}$
  and $\derivof{\deriv'_b}{\judgv[\mcount_b]{\fctx_b}{\tctx_b}{\tmthree}{\mtypthree}}$,
  where $\optmtypthree \mleq \mtypthree$.
  We can apply the \ih on $\deriv_a$, yielding that there exist
  a finite set $I_a$, family contexts $\fctx_{a0}$ and $(\fctxtwo_i)_{i \in I_a}$,
  typing contexts $\tctx_{a0}$ and $(\tctxtwo_i)_{i \in I_a}$, 
  multi-counters $\mcount_{a0}$ and $(\mcounttwo_i)_{i \in I_a}$ 
  and multitypes $(\mtyptwo_i)_{i \in I_a}$ such that:
  \begin{enumerate}
  \item[a.1.]
    $\derivof{\deriv_a}{
      \judgv[\mcount_{a0}]
        {\fctx_{a0}, \var : \fset{\mtyptwo_i}_{i \in I_a}}{\tctx_{a0}}
        {\tmtwo'}
        {\msetabs{\optmtypthree \to \mtyp}}}$
  \item[a.2.]
    $(\derivof{\derivtwo_i}{\judgv[\mcounttwo_i]{\fctxtwo_i}{\tctxtwo_i}{\tmsix}{\mtyptwo_i}})_{i \in I_a}$
  \item[a.3.]
    $\fctx_a = \fctx_{a0} +_{i \in I_a} \fctxtwo_i$ and
    $\tctx_a = \tctx_{a0} +_{i \in I_a} \tctxtwo_i$ and
    $\mcount_a = \mcount_{a0} +_{i \in I_a} \mcounttwo_i$
  \end{enumerate}
  And we can apply the \ih on $\deriv_b$, yielding that there exist
  a finite set $I_b$,
  family contexts $\fctx_{b0}$ and $(\fctxtwo_i)_{i \in I_b}$,
  typing contexts $\tctx_{b0}$ and $(\tctxtwo_i)_{i \in I_b}$,
  multi-counters $\mcount_{b0}$ and $(\mcounttwo_i)_{i \in I_b}$ 
  and multitypes $(\mtyptwo_i)_{i \in I_b}$ such that:
  \begin{enumerate}
  \item[b.1.]
    $\derivof{\deriv_b}{
      \judgv[\mcount_{b0}]
        {\fctx_{b0}, \var : \fset{\mtyptwo_i}_{i \in I_b}}{\tctx_{b0}}
        {\tmthree'}
        {\mtypthree}}$
  \item[b.2.]
    $(\derivof{\derivtwo_i}{\judgv[\mcounttwo_i]{\fctxtwo_i}{\tctxtwo_i}{\tmsix}{\mtyptwo_i}})_{i \in I_b}$
  \item[b.3.]
    $\fctx_b = \fctx_{b0} +_{i \in I_b} \fctxtwo_i$ and
    $\tctx_b = \tctx_{b0} +_{i \in I_b} \tctxtwo_i$ and
    $\mcount_b = \mcount_{b0} +_{i \in I_b} \mcounttwo_i$
  \end{enumerate}
  Taking $I \defeq I_a \uplus I_b$ and $\fctx_0 \defeq \fctx_{a0} + \fctx_{b0}$ and
  $(\fctxtwo_i)_\iI \defeq (\fctxtwo_i)_{i \in I_a \uplus I_b}$ and
  $\tctx_0 \defeq \tctx_{a0} + \tctx_{b0}$ and
  $(\tctxtwo_i)_\iI \defeq (\tctxtwo_i)_{i \in I_a \uplus I_b}$ and
  $\mcount_0 \defeq \mset{\bBeta} + \mcount_{a0} + \mcount_{b0}$ and
  $(\mcounttwo_i)_\iI \defeq (\mcounttwo_i)_{i \in I_a \uplus I_b}$
  and $(\mtyptwo_i)_\iI \defeq (\mtyptwo_i)_{i \in I_a \uplus I_b}$,
  we have:
  \begin{enumerate}
  \item
    $\derivof{\deriv}{
      \judgv[\mset{\bBeta} + \mcount_{a0} + \mcount_{b0}]
        {\fctx_{a0} + \fctx_{b0}, \var : \fset{\mtyptwo_i}_{I_a \uplus I_b}}{\tctx_{a0} + \tctx_{b0}}
        {\tmtwo' \, \tmthree'}
        {\mtyp}}$,
    by applying rule $\vruleTypApp$ to $\deriv_a$ and $\deriv_b$,
    given $\optmtypthree \mleq \mtypthree$
  \item
    $(\derivof{\derivtwo_i}{\judgv[\mcounttwo_i]{\fctxtwo_i}{\tctxtwo_i}{\tmsix}{\mtyptwo_i}})_{i \in I_a \uplus I_b}$
    as the result of applying \cref{lemCBV:generalized_value_splitting_merging} 
    to $(\derivtwo_i)_{i \in I_a}$ and $(\derivtwo_i)_{i \in I_b}$
  \item
    $\fctx = \fctx_a + \fctx_b 
    = \fctx_{a0} + \fctx_{b0} +_{I_a \uplus I_b} \fctxtwo_i 
    = \fctx_0 +_\iI \fctxtwo_i$ and
    $\tctx = \tctx_a + \tctx_b 
    = \tctx_{a0} + \tctx_{b0} +_{I_a \uplus I_b} \tctxtwo_i 
    = \tctx_0 +_\iI \tctxtwo_i$ and
    $\mcount = \mset{\bBeta} + \mcount_a + \mcount_b 
    = \mset{\bBeta} + \mcount_{a0} + \mcount_{b0} +_{I_a \uplus I_b} \mcounttwo_i 
    = \mcount_0 +_\iI \mcounttwo_i$
  \end{enumerate}
\item $\vruleTypZero$.
  Analogous to case $\vruleTypVarOne$.
  % Then $\judgv[\emset]{\emptyctx}{\emptyctx}{\zero}{\msetnat{\zeroTyp}_\jJ}$,
  % where $\fctx = \emptyctx$ and $\tctx = \emptyctx$ and $\mcount = \emset$ and
  % $\mtyp = \msetnat{\zeroTyp}_\jJ$.
  % Taking $I = \emptyctx$ 
  % and $\fctx_0 = \fctx$, $\tctx_0 = \tctx$ and $\mcount_0 = \mcount$,
  % we met the required conditions, and the judgment
  % $\judgv[\emset]{\emptyctx}{\emptyctx}{\zero}{\msetnat{\zeroTyp}_\jJ}$
  % from \cref{it:anti_substitution_1} holds by hypothesis.
  % Note that \cref{it:anti_substitution_2} holds trivially
  % since $I = \emptyset$.
  % Then $\tm = \vartwo$, with $\vartwo \neq \var$, and
  % $\judgv[\emset]{\vartwo : \fset{\mtyp}}{\emptyctx}{\vartwo}{\mtyp}$,
  % with $\fctx = \vartwo : \fset{\mtyp}$ and $\tctx = \emptyctx$ and $\mcount = \emset$.
\item $\vruleTypSucc$.
  Then
  $\derivof{\deriv'}{\judgv[\mcount]{\fctx}{\tctx}{\succ{\tmthree}}{\msetnat{\succTyp{\mtypnat_j}}_\jJ}}$,
  where $\mtyp = \msetnat{\succTyp{\mtypnat_j}}_\jJ$.
  Then $\tm = \succ{\tmthree'}$, 
  with $\tmthree = \tmthree'\sub{\var}{\tmsix}$.
  The judgment is derived from
  $\derivof{\deriv'_0}{\judgv[\mcount]{\fctx}{\tctx}{\tmthree}{\mtypnat}}$,
  where $\mtypnat = +_\jJ \mtypnat_j$.
  We can apply the \ih on $\deriv'_0$, yielding that there exist
  a finite set $I$, family contexts $\fctx_0$ and $(\fctxtwo_i)_\iI$,
  typing contexts $\tctx_0$ and $(\tctxtwo_i)_\iI$,
  multi-counters $\mcount_0$ and $(\mcounttwo_i)_\iI$ and multitypes $(\mtyptwo_i)_\iI$
  such that:
  \begin{enumerate}
  \item[1'.]
    $\derivof{\deriv'_0}{
      \judgv[\mcount_0]
        {\fctx_0, \var : \fset{\mtyptwo_i}_\iI}{\tctx_0}
        {\tmthree'}
        {\mtypnat}}$, where
    $\mtypnat = +_\jJ \mtypnat_j$
  \item[2'.]
    $(\derivof{\derivtwo_{0i}}{\judgv[\mcounttwo_i]{\fctxtwo_i}{\tctxtwo_i}{\tmsix}{\mtyptwo_i}})_\iI$
  \item[3'.]
    $\fctx = \fctx_0 +_\iI \fctxtwo_i$ and
    $\tctx = \tctx_0 +_\iI \tctxtwo_i$ and
    $\mcount = \mcount_0 +_\iI \mcounttwo_i$
  \end{enumerate}
  By taking $I$ and $\fctx_0$ and $(\fctxtwo_i)_\iI$ and $\tctx_0$ and 
  $(\tctxtwo_i)_\iI$ and $\mcount_0$ and $(\mcounttwo_i)_\iI$ and $(\mtyptwo_i)_\iI$
  from the \ih, we obtain
  $\derivof{\deriv}{
    \judgv[\mcount_0]
      {\fctx_0, \var : \fset{\mtyptwo_i}_\iI}{\tctx_0 +_\iI \tctxtwo_i}
      {\succ{\tmthree'}}
      {\msetnat{\succTyp{\mtypnat_j}}_\jJ}}$
  by applying rule $\vruleTypSucc$ with $\deriv_0$ as premise,
  and the condition in \cref{it:anti_substitution_2} holds by letting $\derivtwo_i \defeq \derivtwo_{i0}$
  for all $\iI$.
  The conditions in \cref{it:anti_substitution_3} hold by the \ih.
\item $\vruleTypIfZero$.
  Analogous to case $\vruleTypApp$.
\item $\vruleTypIfSucc$.
  Then
  $\derivof{\deriv'}{
    \judgv[\mset{\bIfSucc} + \mcount_a + \mcount_b]
      {\fctx_a + \fctx_b}{\tctx_a + \tctx_b}
      {\ifz{\tmtwo}{\tmthree}{\vartwo}{\tmfour}}
      {\mtyp}}$,
  where $\fctx = \fctx_a + \fctx_b$ and $\tctx = \tctx_a + \tctx_b$ and
  $\mcount = \mset{\bIfSucc} + \mcount_a + \mcount_b$.
  Then $\tm = \ifz{\tmtwo'}{\tmthree'}{\vartwo}{\tmfour'}$,
  with $\tmtwo = \tmtwo'\sub{\var}{\tmsix}$ and $\tmthree = \tmthree'\sub{\var}{\tmsix}$
  and $\tmfour = \tmfour'\sub{\var}{\tmsix}$.
  The judgment is derived from
  $\derivof{\deriv'_a}{\judgv[\mcount_a]{\fctx_a}{\tctx_a}{\tmtwo}{\msetnat{\succTyp{\mtypnat}}}}$ and
  $\derivof{\deriv'_b}{\judgv[\mcount_b]{\fctx_b}{\tctx_b, \vartwo : \optmtypnat}{\tmfour}{\mtyp}}$,
  where $\optmtypnat \mleq \mtypnat$.
  We can apply the \ih on $\deriv'_a$, yielding that there exist
  a finite set $I_a$,
  family contexts $\fctx_{a0}$ and $(\fctxtwo_i)_{i \in I_a}$,
  typing contexts $\tctx_{a0}$ and $(\tctxtwo_i)_{i \in I_a}$,
  multi-counters $\mcount_{a0}$ and $(\mcounttwo_i)_{i \in I_a}$ 
  and multitypes $(\mtyptwo_i)_{i \in I_a}$ such that:
  \begin{enumerate}
  \item[a.1.]
    $\derivof{\deriv_a}{
      \judgv[\mcount_{a0}]
        {\fctx_{a0}, \var : \fset{\mtyptwo_i}_{i \in I_a}}{\tctx_{a0}}
        {\tmtwo'}
        {\msetnat{\succTyp{\mtypnat}}}}$
  \item[a.2.]
    $(\derivof{\derivtwo_i}{\judgv[\mcounttwo_i]{\fctxtwo_i}{\tctxtwo_i}{\tmsix}{\mtyptwo_i}})_{i \in I_a}$
  \item[a.3.]
    $\fctx_a = \fctx_{a0} +_{i \in I_a} \fctxtwo_i$ and
    $\tctx_a = \tctx_{a0} +_{i \in I_a} \tctxtwo_i$ and
    $\mcount_a = \mcount_{a0} +_{i \in I_a} \mcounttwo_i$
  \end{enumerate}
  And we can apply the \ih on $\deriv'_b$, yielding that there exist
  a finite set $I_b$,
  family contexts $\fctx_{b0}$ and $(\fctxtwo_i)_{i \in I_b}$,
  typing contexts $\tctx_{b0}$ and $(\tctxtwo_i)_{i \in I_b}$,
  multi-counters $\mcount_{b0}$ and $(\mcounttwo_i)_{i \in I_b}$ 
  and multitypes $(\mtyptwo_i)_{i \in I_b}$ such that:
  \begin{enumerate}
  \item[b.1.]
    $\derivof{\deriv_b}{
      \judgv[\mcount_{b0}]
        {\fctx_{b0}, \var : \fset{\mtyptwo_i}_{i \in I_b}}{\tctx_{b0}, \vartwo : \optmtypnat}
        {\tmfour'}
        {\mtyp}}$
  \item[b.2.]
    $(\derivof{\derivtwo_i}{\judgv[\mcounttwo_i]{\fctxtwo_i}{\tctxtwo_i}{\tmsix}{\mtyptwo_i}})_{i \in I_b}$
  \item[b.3.]
    $\fctx_b = \fctx_{b0} +_{i \in I_b} \fctxtwo_i$ and
    $\tctx_b, \vartwo : \optmtypnat 
    = \tctx_{b0}, \vartwo : \optmtypnat +_{i \in I_b} \tctxtwo_i$ and
    $\mcount_b = \mcount_{b0} +_{i \in I_b} \mcounttwo_i$.
    Note that it is not necessary to split $\optmtypnat_j$ since we
    may assume $\vartwo \notin \fv{\tmsix}$
    by $\alpha$-conversion, so $\vartwo \notin \dom{+_{i \in I_b} \tctxtwo_i}$ 
    by \cref{lemCBV:relevance}.
  \end{enumerate}
  Taking $I \defeq I_a \uplus I_b$ and $\fctx_0 \defeq \fctx_{a0} + \fctx_{b0}$ and
  $(\fctxtwo_i)_\iI \defeq (\fctxtwo_i)_{i \in I_a \uplus I_b}$ and
  $\tctx_0 \defeq \tctx_{a0} + \tctx_{b0}$ and
  $(\tctxtwo_i)_\iI \defeq (\tctxtwo_i)_{i \in I_a \uplus I_b}$ and
  $\mcount_0 \defeq \mset{\bIfSucc} + \mcount_{a0} + \mcount_{b0}$ and
  $(\mcounttwo_i)_\iI \defeq (\mcounttwo_i)_{i \in I_a \uplus I_b}$
  and $(\mtyptwo_i)_\iI \defeq (\mtyptwo_i)_{i \in I_a \uplus I_b}$,
  we have:
  \begin{enumerate}
  \item
    $\derivof{\deriv}{
      \judgv[\mset{\bIfSucc} + \mcount_{a0} + \mcount_{b0}]
        {\fctx_{a0} + \fctx_{b0}, \var : \fset{\mtyptwo_i}_{I_a \uplus I_b}}{\tctx_{a0} + \tctx_{b0}}
        {\ifz{\tmtwo'}{\tmthree'}{\vartwo}{\tmfour'}}
        {\mtyp}}$,
    by applying rule $\vruleTypIfSucc$ to $\deriv_a$ and $\deriv_b$,
    given $\optmtypthree \mleq \mtypthree$
  \item
    $(\derivof{\derivtwo_i}{\judgv[\mcounttwo_i]{\fctxtwo_i}{\tctxtwo_i}{\tmsix}{\mtyptwo_i}})_{i \in I_a \uplus I_b}$
    as the result of applying \cref{lemCBV:generalized_value_splitting_merging} 
    to $(\derivtwo_i)_{i \in I_a}$ and $(\derivtwo_i)_{i \in I_b}$
  \item
    $\fctx = \fctx_a + \fctx_b 
    = \fctx_{a0} + \fctx_{b0} +_{I_a \uplus I_b} \fctxtwo_i 
    = \fctx_0 +_\iI \fctxtwo_i$ and
    $\tctx = \tctx_a + \tctx_b 
    = \tctx_{a0} + \tctx_{b0} +_{I_a \uplus I_b} \tctxtwo_i 
    = \tctx_0 +_\iI \tctxtwo_i$ and
    $\mcount = \mset{\bIfSucc} + \mcount_a + \mcount_b 
    = \mset{\bIfSucc} + \mcount_{a0} + \mcount_{b0} +_{I_a \uplus I_b} \mcounttwo_i 
    = \mcount_0 +_\iI \mcounttwo_i$
  \end{enumerate}
\item $\vruleTypFix$.
  Then
  $\derivof{\deriv'}{
    \judgv[\mset{\bFix} + \mcount' +_\jJ \mcount'_j]
      {\fctx' +_\jJ \fctx'_j}{\tctx' +_\jJ \tctx'_j}
      {\fix{\vartwo}{\tmtwo}}
      {\mtyp}}$,
  where 
  $\fctx = \fctx' +_\jJ \fctx'_j$ and
  $\tctx = \tctx' +_\jJ \tctx'_j$ and
  $\mcount = \mset{\bFix} + \mcount' +_\jJ \mcount'_j$.
  Then $\tm = \fix{\vartwo}{\tmtwo'}$,
  with $\tmtwo = \tmtwo'\sub{\var}{\tmsix}$.
  The judgment is derived from
  $\derivof{\deriv'_0}{
    \judgv[\mcount']
      {\fctx', \vartwo : \fset{\mtypthree_j}_\jJ}{\tctx'}{\tmtwo}{\mtyp}}$ and
  $(\derivof{\deriv'_j}{
    \judgv[\mcount'_j]
      {\fctx'_j}{\tctx'_j}
      {\fix{\vartwo}{\tmtwo}}
      {\mtypthree_j}})_\jJ$.
  We can apply \ih on $\tmtwo$, yielding that there exist
  a finite set $I'$, family contexts $\fctx'_0$ and $(\fctxtwo'_i)_{i \in I'}$,
  typing contexts $\tctx'_0$ and $(\tctxtwo'_i)_{i \in I'}$,
  multi-counters $\mcount'_0$ and $(\mcounttwo'_i)_{i \in I'}$
  and multitypes $(\mtyptwo_i)_{i \in I'}$ such that:
  \begin{enumerate}
  \item[1'.]
    $\derivof{\deriv_0}{
      \judgv[\mcount'_0]
        {\fctx'_0, \vartwo : \fset{\mtypthree_j}_\jJ, \var : \fset{\mtyptwo_i}_{i \in I'}}{\tctx'_0}
        {\tmtwo'}{\mtyp}}$
  \item[2'.]
    $(\derivof{\derivtwo_i}{\judgv[\mcount'_i]{\fctxtwo'_i}{\tctxtwo'_i}{\tmsix}{\mtyptwo_i}})_{i \in I'}$
  \item[3'.]
    $\fctx', \vartwo : \fset{\mtypthree_j}_\jJ 
    = \fctx'_0, \vartwo : \fset{\mtypthree_j}_\jJ +_{i \in I'} \fctxtwo'_i$ and
    $\tctx' = \tctx'_0 +_{i \in I'} \tctxtwo'_i$ and
    $\mcount' = \mcount'_0 +_{i \in I'} \mcounttwo'_i$.
    Note that it is not necessary to split $\fset{\mtypthree_j}_\jJ$ since
    may assume $\vartwo \notin \fv{\tmsix}$
    by $\alpha$-conversion, so $\vartwo \notin \dom{+_{i \in I'} \fctxtwo_i}$ 
    by \cref{lemCBV:relevance}.
    \end{enumerate}
  Since the conclusion of $\deriv'_j$ is smaller than (1) for all $\jJ$,
  we can then apply the \ih on $(\deriv'_j)_\jJ$, yielding that for each $\jJ$, 
  there exist a finite set $I'_j$,
  family contexts $\fctx'_{j0}$ and $(\fctxtwo'_i)_{i \in I'_j}$,
  typing contexts $\tctx'_{j0}$ and $(\tctxtwo'_i)_{i \in I'_j}$,
  multi-counters $\mcount'_{j0}$ and $(\mcounttwo'_i)_{i \in I'_j}$ and
  multitypes $(\mtyptwo_i)_{i \in I'_j}$ such that:
  \begin{enumerate}
  \item[j.1.]
    $\derivof{\deriv_j}{
      \judgv[\mcount'_{j0}]
        {\fctx'_{j0}, \var : \fset{\mtyptwo_i}_{i \in I'_j}}{\tctx'_{j0}}
        {\fix{\vartwo}{\tmtwo'}}
        {\mtypthree_j}}$
  \item[j.2.]
    $(\derivof{\derivtwo_i}{\judgv[\mcount'_i]{\fctxtwo'_i}{\tctxtwo'_i}{\tmsix}{\mtyptwo_i}})_{i \in I'_j}$
  \item[j.3.]
    $\fctx = \fctx_{j0} +_{i \in I_j} \fctxtwo_i$ and
    $\tctx = \tctx_{j0} +_{i \in I_j} \tctxtwo_i$ and
    $\mcount = \mcount_{j0} +_{i \in I_j} \mcounttwo_i$
  \end{enumerate}
  Taking
  $I \defeq I' \uplus_\jJ I'_j$ and $\fctx_0 \defeq \fctx'_0 +_\jJ \fctx'_{j0}$ and 
  $(\fctxtwo_i)_\iI \defeq (\fctxtwo'_i)_{I' \uplus_\jJ I'_j}$ and
  $\tctx_0 \defeq \tctx'_0 +_\jJ \tctx'_{j0}$ and
  $(\tctxtwo_i)_\iI \defeq (\tctxtwo'_i)_{I' \uplus_\jJ I'_j}$ and
  $\mcount_0 \defeq \mset{\bFix} + \mcount'_0 +_\jJ \mcount'_{j0}$ and
  $(\mcounttwo_i)_\iI \defeq (\mcounttwo'_i)_{I' \uplus_\jJ I'_j}$ and
  and $(\mtyptwo_i)_\iI \defeq (\mtyptwo_i)_{I' \uplus_\jJ I'_j}$
  such that:
  \begin{enumerate}
  \item
    $\derivof{\deriv}{
      \judgv[\mcount_0]
        {\fctx_0, \var : \fset{\mtyptwo_i}_\iI}{\tctx_0}
        {\fix{\vartwo}{\tmtwo'}}
        {\mtyp}}$,
    by applying rule $\vruleTypFix$ to $\deriv_0$ and $(\deriv_j)_\jJ$
  \item
    $(\derivof{\derivtwo_i}{\judgv[\mcounttwo_i]{\fctxtwo_i}{\tctxtwo_i}{\tmsix}{\mtyptwo_i}})_\iI$,
    as the result of applying \cref{lemCBV:generalized_value_splitting_merging} 
    to $(\derivtwo_i)_{i \in I'}$ and $(\derivtwo_i)_{i \in I_j}$ for all $\jJ$
  \item
    $\fctx = \fctx' +_\jJ \fctx'_j
    = \fctx'_0 +_\jJ \fctx'_{j0} +_{I' \uplus_\jJ I'_j} (\fctxtwo'_i)_{I' \uplus_\jJ I'_j} 
    = \fctx_0 +_\iI \fctxtwo_i$ and
    $\tctx = \tctx' +_\jJ \tctx'_j
    = \tctx'_0 +_\jJ \tctx'_{j0} +_{I' \uplus_\jJ I'_j} (\tctxtwo'_i)_{I' \uplus_\jJ I'_j} 
    = \tctx_0 +_\iI \tctxtwo_i$ and
    $\mcount_0 = \mset{\bFix} + \mcount' +_\jJ \mcount'_j
    = \mcount'_0 +_\jJ \mcount'_{j0} +_{I' \uplus_\jJ I'_j} (\mcounttwo'_i)_{I' \uplus_\jJ I'_j} 
    = \mcount_0 +_\iI \mcounttwo_i$
  \end{enumerate}
\end{itemize}
\end{proof}

\subjectexpansion*
% Label lemCBV:subject_expansion

\begin{proof}
We proceed by induction on the derivation of $\tm \tov{\rulename} \tm'$.
\begin{enumerate}
\item $\vruleToBeta$.
  Then
  $\tm = (\lam{\var}{\tmtwo})\,\val \tov{\bBeta} \tmtwo\sub{\var}{\val} = \tm'$,
  with $\derivof{\deriv'}{\judgv[\mcount']{\fctx}{\tctx}{\tmtwo\sub{\var}{\val}}{\mtyp}}$.
  By \nameref{lemCBV:value_anti_substitution}, there exist
  $\fctx_1$, $\fctx_2$, $\tctx_1$, $\tctx_2$, $\mcount'_1$, $\mcount'_2$,
  $\optmtyptwo$ and $\mtyptwo$ such that
  $\fctx = \fctx_1 + \fctx_2$ and $\tctx = \tctx_1 + \tctx_2$ and
  $\mcount' = \mcount'_1 + \mcount'_2$ and $\optmtyptwo \mleq \mtyptwo$
  and $\derivof{\deriv_0}{\judgv[\mcount'_1]{\fctx_1}{\tctx_1, \var : \optmtyptwo}{\tmtwo}{\mtyp}}$
  and $\derivof{\derivtwo}{\judgv[\mcount'_2]{\fctx_2}{\tctx_2}{\val}{\mtyptwo}}$ hold.
  Hence we build the derivation $\deriv$ as:
  \[
    \indrule{\vruleTypApp}{
      \indrule{\vruleTypAbs}{
        \derivof{\deriv_0}{\judgv[\mcount'_1]{\fctx_1}{\tctx_1, \var : \optmtyptwo}{\tmtwo}{\mtyp}}
      }{
        \judgv[\mcount'_1]{\fctx_1}{\tctx_1}{\lam{\var}{\tmtwo}}{\msetabs{\optmtyptwo \to \mtyp}}
      }
      \optmtyptwo \mleq \mtyptwo
      \HS
      \derivof{\derivtwo}{\judgv[\mcount'_2]{\fctx_2}{\tctx_2}{\val}{\mtyptwo}}
    }{
      \judgv[\mset{\bBeta} + \mcount']{\fctx}{\tctx}{(\lam{\var}{\tmtwo})\,\val}{\mtyp}
    }
  \]
\item $\vruleToIfZero$.
  Then $\tm = \ifz{\zero}{\tm'}{\var}{\tmtwo} \tov{\bIfZero} \tm'$
  with $\derivof{\deriv'}{\judgv[\mcount']{\fctx}{\tctx}{\tm'}{\mtyp}}$,
  and we build the derivation $\deriv$ as:
  \[
    \indrule{\vruleTypIfZero}{
      \indrule{\vruleTypZero}{
        \emptyPremise
      }{
        \judgv[\emset]{\emptyctx}{\emptyctx}{\zero}{\mset{\zeroTyp}}
      }
      \derivof{\deriv'}{\judgv[\mcount']{\fctx}{\tctx}{\tm'}{\mtyp}}
    }{
      \judgv[\mset{\bIfZero}+\mcount']{\fctx}{\tctx}{\ifz{\zero}{\tm'}{\var}{\tmtwo}}{\mtyp}
    }
  \]
\item $\vruleToIfSucc$.
  Then 
  $\tm = \ifz{\succ{\valnat}}{\tmtwo}{\var}{\tmthree}
  \tov{\bIfSucc} \tmthree\sub{\var}{\valnat} = \tm'$
  with $\derivof{\deriv'}{\judgv[\mcount']{\fctx}{\tctx}{\tmthree\sub{\var}{\valnat}}{\mtyp}}$.
  By \nameref{lemCBV:value_anti_substitution} there exist
  family contexts $\fctx_1$, $\fctx_2$,
  typing contexts $\tctx_1$, $\tctx_2$,
  multi-counters $\mcount'_1$, $\mcount'_2$,
  an optional multitype $\optmtyptwo$, and a multitype $\mtyptwo$
  such that
  $\fctx = \fctx_1 + \fctx_2$ and
  $\tctx = \tctx_1 + \tctx_2$ and
  $\mcount' = \mcount'_1 + \mcount'_2$ and
  $\optmtyptwo \mleq \mtyptwo$
  and
  $\derivof{\deriv_0}{\judgv[\mcount'_1]{\fctx_1}{\tctx_1,\var:\optmtyptwo}{\tmthree}{\mtyp}}$
  $\derivof{\derivtwo}{\judgv[\mcount'_2]{\fctx_2}{\tctx_2}{\valnat}{\mtyptwo}}$.
  Moreover $\mtyptwo$ is of the form $\mtypnat$ by \cref{remCBV:valnat_has_multitype_nat}.
  Hence we build $\deriv$ as:
  \[
    \indrule{\vruleTypIfSucc}{
      \indrule{\vruleTypSucc}{
        \derivof{\derivtwo}{\judgv[\mcount'_2]{\fctx_2}{\tctx_2}{\valnat}{\mtyptwo}}
      }{
        \judgv[\mcount'_2]{\fctx_2}{\tctx_2}{\succ{\valnat}}{\mset{\succTyp{\mtyptwo}}}
      }
      \optmtyptwo \mleq \mtyptwo
      \HS
      \derivof{\deriv_0}{\judgv[\mcount'_1]{\fctx_1}{\tctx_1,\var:\optmtyptwo}{\tmthree}{\mtyp}}
    }{
      \judgv[\mset{\bIfSucc}+\mcount']{\fctx}{\tctx}{\ifz{\succ{\valnat}}{\tmtwo}{\var}{\tmthree}}{\mtyp}
    }
  \]
\item $\vruleToFix$.
  Then $\tm = \fix{\var}{\tmtwo}
        \tov{\bFix} \tmtwo\sub{\var}{\fix{\var}{\tmtwo}} = \tm'$
  with $\judgv[\mcount']{\fctx}{\tctx}{\tmtwo\sub{\var}{\fix{\var}{\tmtwo}}}{\mtyp}$.
  By \nameref{lemCBV:anti_substitution},
  there exist a finite set $I$,
  family contexts $\fctx_0$, $(\fctx_i)_\iI$,
  typing contexts $\tctx_0$, $(\tctx_i)_\iI$,
  multi-counters $\mcount'_0$, $(\mcount'_i)_\iI$,
  and multitypes $(\mtyptwo_i)_\iI$
  such that
  $\fctx = \fctx_0 +_\iI \fctx_i$ and
  $\tctx = \tctx_0 +_\iI \tctx_i$ and
  $\mcount' = \mcount'_0 +_\iI \mcount'_i$,
  and such that
  $\derivof{\deriv_0}{\judgv[\mcount'_0]{\fctx_0,\var:\fset{\mtyptwo_i}_\iI}{\tctx_0}{\tmtwo}{\mtyp}}$
  and
  $\derivof{\derivtwo_i}{\judgv[\mcount'_i]{\fctx_i}{\tctx_i}{\fix{\var}{\tmtwo}}{\mtyptwo_i}}$ hold
  for each $\iI$.
  Hence we build $\deriv$ as:
  \[
    \indrule{\vruleTypFix}{
      \derivof{\deriv_0}{\judgv[\mcount'_0]{\fctx_0,\var:\fset{\mtyptwo_i}_\iI}{\tctx_0}{\tmtwo}{\mtyp}}
      \HS
      (\derivof{\deriv_i}{\judgv[\mcount'_i]{\fctx_i}{\tctx_i}{\fix{\var}{\tmtwo}}{\mtyptwo_i}})_\iI
    }{
      \judgv[\mset{\bFix}+\mcount']{\fctx}{\tctx}{\fix{\var}{\tmtwo}}{\mtyp}
    }
  \] 
\item $\vruleToCongAppL$.
  Then
  $\tm = \tmtwo \, \tmthree \tov{\rulename} \tmtwo' \, \tmthree = \tm'$,
  which is derived from $\tmtwo \tov{\rulename} \tmtwo'$.
  Moreover, we have 
  $\derivof{\deriv'}{\judgv[\mcount']{\fctx}{\tctx}{\tmtwo' \, \tmthree}{\mtyp}}$,
  which can only be derived by rule $\vruleTypApp$,
  thus $\deriv'$ has the form:
  \[
    \indrule{\vruleTypApp}{
      \judgv[\mcount'_1]
        {\fctx_1}{\tctx_1}
        {\tmtwo'}{\msetabs{\optmtyptwo \to \mtyp}}
      \HS
      (1)\ \optmtyptwo \mleq \mtyptwo
      \HS
      \derivof{\deriv_2}{
        \judgv[\mcount'_2]
          {\fctx_2}{\tctx_2}
          {\tmthree}{\mtyptwo}
      }
    }{
      \judgv[\mset{\bBeta} + \mcount'_1 + \mcount'_2]
        {\fctx_1 + \fctx_2}{\tctx_1 + \tctx_2}
        {\tmtwo' \, \tmthree}{\mtyp}
    }
  \]
  where $\fctx = \fctx_1 + \fctx_2$ and $\tctx = \tctx_1 + \tctx_2$ and
  $\mcount' = \mset{\bBeta} + \mcount'_1 + \mcount'_2$.
  By the \ih on $\tmtwo \tov{\rulename} \tmtwo'$ there exists $\mcount_1$ such that
  $\mcount_1 = \mset{\rulename} + \mcount'_1$, and 
  $\derivof{\deriv_1}{\judgv[\mcount_1]{\fctx_1}{\tctx_1}{\tmtwo}{\msetabs{\optmtyptwo \to \mtyp}}}$.
  Applying rule $\vruleTypApp$ with $\deriv_1$, (1) and $\deriv_2$ as premises,
  we conclude
  $\derivof{\deriv}{
    \judgv[\mset{\bBeta} + \mcount_1 + \mcount'_2]
      {\fctx_1 + \fctx_2}{\tctx_1 + \tctx_2}
      {\tmtwo \, \tmthree}{\mtyp}}$,
  where $\mcount = \mset{\bBeta} + \mcount_1 + \mcount'_2$ and verifies
  $\mcount = \mset{\bBeta} + \mset{\rulename} + \mcount'_1 + \mcount'_2 
  = \mset{\rulename} + \mcount'$.
\item $\vruleToCongAppR$.
  Analogous to the previous case.
\item $\vruleToCongSucc$.
  Then
  $\tm = \succ{\tmtwo} \tov{\rulename} \succ{\tmtwo'} = \tm'$,
  which is derived from $\tmtwo \tov{\rulename} \tmtwo'$.
  Moreover, we have
  $\derivof{\deriv'}{\judgv[\mcount']{\fctx}{\tctx}{\succ{\tmtwo'}}{\mtyp}}$,
  which can only be derived by rule $\vruleTypSucc$,
  thus $\deriv'$ has the form:
  \[
    \indrule{\vruleTypSucc}{
      \judgv[\mcount']{\fctx}{\tctx}{\tmtwo'}{+_\iI \mtyptwonat_i}
    }{
      \judgv[\mcount']
        {\fctx}{\tctx}{\succ{\tmtwo'}}{\msetnat{\succTyp{\mtyptwonat_i}}_\iI}
    }
  \]
  where
  $\mtyp = \msetnat{\succTyp{\mtyptwonat_i}}_\iI$.
  By the \ih on $\tmtwo \tov{\rulename} \tmtwo'$, there exists $\mcount$ such that
  $\mcount = \mset{\rulename} + \mcount'$, and 
  $\derivof{\deriv_0}{\judgv[\mcount]{\fctx}{\tctx}{\tmtwo}{+_\iI \mtyptwonat_i}}$.
  Applying rule $\vruleTypSucc$ to $\deriv_0$, we conclude with
  $\derivof{\deriv}{\judgv[\mcount]{\fctx}{\tctx}{\succ{\tmtwo}}{\msetnat{\succTyp{\mtyptwonat_i}}_\iI}}$,
  where $\mcount$ verifies
  $\mcount = \mset{\rulename} + \mcount'$ by the \ih.
\item $\vruleToCongIf$.
  Then
  $\tm = \ifz{\tmtwo}{\tmthree}{\var}{\tmfour}
   \tov{\rulename} \ifz{\tmtwo'}{\tmthree}{\var}{\tmfour} = \tm'$
  which is derived from $\tmtwo \tov{\rulename} \tmtwo'$.
  Moreover, we have
  $\derivof{\deriv'}{\judgv[\mcount']{\fctx}{\tctx}{\ifz{\tmtwo'}{\tmthree}{\var}{\tmfour}}{\mtyp}}$.
  We consider two subcases, depending on whether the conclusion of $\deriv'$ 
  is obtained from $\vruleTypIfZero$ or $\vruleTypIfSucc$:
  \begin{enumerate}
  \item $\vruleTypIfZero$.
    Then there exist
    family contexts $\fctx_1$, $\fctx_2$,
    typing contexts $\tctx_1$, $\tctx_2$, and
    multi-counters $\mcount'_1$, $\mcount'_2$
    such that
    $\fctx = \fctx_1 + \fctx_2$ and
    $\tctx = \tctx_1 + \tctx_2$ and
    $\mcount' = \mset{\bIfZero} + \mcount'_1 + \mcount'_2$,
    and $\derivof{\deriv'_1}{\judgv[\mcount'_1]{\fctx_1}{\tctx_1}{\tmtwo'}{\mset{\zeroTyp}}}$
    and $\derivof{\deriv'_2}{\judgv[\mcount'_2]{\fctx_2}{\tctx_2}{\tmthree}{\mtyp}}$.
    By the \ih on $\tmtwo \tov{\rulename} \tmtwo'$, there exits $\mcount_1$ such that
    $\mcount_1 = \mset{\rulename} + \mcount'_1$, and
    $\derivof{\deriv_1}{\judgv[\mset{\rulename}+\mcount'_1]{\fctx_1}{\tctx_1}{\tmtwo}{\mset{\zeroTyp}}}$.
    Hence we build $\deriv$ as:
    \[
      \indrule{\vruleTypIfZero}{
        \derivof{\deriv_1}{\judgv[\mset{\rulename}+\mcount'_1]{\fctx_1}{\tctx_1}{\tmtwo}{\msetnat{\zeroTyp}}}
        \HS
        \derivof{\deriv'_2}{\judgv[\mcount'_2]{\fctx_2}{\tctx_2}{\tmthree}{\mtyp}}
      }{
        \judgv[\mset{\rulename}+\mcount']{\fctx}{\tctx}{\ifz{\tmtwo}{\tmthree}{\var}{\tmfour}}{\mtyp}
      }
    \]
  \item $\vruleTypIfSucc$.
    Then there exist
    family contexts $\fctx_1$, $\fctx_2$,
    typing contexts $\tctx_1$, $\tctx_2$,
    multi-counters $\mcount'_1$, $\mcount'_2$,
    an optional $\natsym$-multitype $\optmtypnat$
    and a $\natsym$-multitype $\mtypnat$
    such that
    $\fctx = \fctx_1 + \fctx_2$ and
    $\tctx = \tctx_1 + \tctx_2$ and
    $\mcount' = \mset{\bIfSucc} + \mcount'_1 + \mcount'_2$ and
    (1) $\optmtypnat \mleq \mtypnat$,
    and
    $\derivof{\deriv'_1}{\judgv[\mcount'_1]{\fctx_1}{\tctx_1}{\tmtwo'}{\msetnat{\succTyp{\mtypnat}}}}$
    and
    $\derivof{\deriv'_2}{\judgv[\mcount'_2]{\fctx_2}{\tctx_2,\var:\optmtypnat}{\tmfour}{\mtyp}}$.
    By the \ih on $\tmtwo \tov{\rulename} \tmtwo'$, there exists $\mcount_1$ such that
    $\mcount_1 = \mset{\rulename} + \mcount'_1$, and
    $\derivof{\deriv_1}{\judgv[\mset{\rulename}+\mcount'_1]{\fctx_1}{\tctx_1}{\tmtwo}{\msetnat{\succTyp{\mtypnat}}}}$.
    Hence we build $\deriv$ as:
    \[
      \indrule{\vruleTypIfSucc}{
        \derivof{\deriv_1}{\judgv[\mset{\rulename}+\mcount'_1]{\fctx_1}{\tctx_1}{\tmtwo}{\msetnat{\succTyp{\mtypnat}}}}
        \HS
        (1)\ \optmtypnat \mleq \mtypnat
        \HS
        \derivof{\deriv'_2}{\judgv[\mcount'_2]{\fctx_2}{\tctx_2,\var:\optmtypnat}{\tmfour}{\mtyp}}
      }{
        \judgv[\mset{\rulename}+\mcount']{\fctx}{\tctx}{\ifz{\tmtwo}{\tmthree}{\var}{\tmfour}}{\mtyp}
      }
    \]
  \end{enumerate}
\end{enumerate}
\end{proof}

\NFtypable*
% Label lemCBV:NF_typable

\begin{proof}
By induction on the derivation of $\tm \in \NFv{\nature}$.
Recall cases $\vruleNFApp$, $\vruleNFSuccErr$ and $\vruleNFIf$ do not apply since
$\stucksym$ is not a proper nature.
\begin{enumerate}
\item $\vruleNFAbs$.
  Then $\tm = \lam{\var}{\tmtwo} \in \NFv{\abssym}$
  and we can type $\tm$ using $\vruleTypAbs$.
\item $\vruleNFZero$.
  Then $\tm = \zero \in \NFv{\natsym}$
  and we can type $\tm$ using $\vruleTypZero$.
\item $\vruleNFSuccNat$.
  Then $\tm = \succ{\tmthree} \in \NFv{\natsym}$ where $\tmthree \in \NFv{\natsym}$.
  By \ih on $\tmthree \in \NFv{\natsym}$,
  $\judgv[\mcount]{\fctx}{\tctx}{\tmthree}{\mtyp}$.
  Moreover, since $\tmthree$ is a normal form with $\natsym$ nature, it must be
  the case that $\tmthree = \valnat$, so by \cref{remCBV:valnat_has_multitype_nat},
  $\mtyp$ is of the form $\mtypnat$, which can be written as $+_\iI \mtypnat_i$.
  Hence $\judgv[\mcount]{\fctx}{\tctx}{\succ{\tmthree}}{\msetnu{\natsym}{\succTyp{\mtypnat_i}}_\iI}$
  by applying $\vruleTypSucc$.
  \qedhere
\end{enumerate}
\end{proof}

\tightsoundnesscompleteness*
\begin{proof}
% Follows the same structure as in \cref{thmCBV:soundness_completeness}, 
% \textit{mutatis mutandis}.
% One key point is that to prove soundness $(1 \Rightarrow 2)$,
% in the case in which $\tm \in \NFv{\nature}$, 
% we know that $\tm$ must be a value,
% hence the typing derivation
%   $\judgv[\mcount]{\emptyctx}{\emptyctx}{\tm}{\emsetnu{\nature}}$
% of the hypothesis
% must be such that $\mcount = \emset$
% by \cref{lemCBV:typable_values_emptyctx}.
% The other key point is that to prove completeness $(2 \Rightarrow 1)$
% we resort to the stronger \cref{lemCBV:NF_tight_typable},
% rather than just to \cref{lemCBV:NF_typable},
% to ensure that we can construct a typing derivation
% of the form
%   $\judgv[\emset]{\emptyctx}{\emptyctx}{\tm}{\emsetnu{\nature}}$.
\,\\
($1 \Rightarrow 2$)
  By induction on the size of $\mcount$,
  analyzing whether $\tm \in \NFv{\nature}$ or not.
  \begin{itemize}
  \item If $\tm \in \NFv{\nature}$, 
    then $\tm$ is of the form $\lam{\var}{\tmtwo}$, $\zero$, or $\succ{\valnat}$,
    given that $\nature$ is a proper nature.
    Hence $\tm$ can only be derived by rule $\vruleTypAbs$, $\vruleTypZero$, or $\vruleTypSucc$ respectively.
    In the three cases we obtain $\mcount = \emset$ by \cref{lemCBV:typable_values_emptyctx},
    so taking $n = 0$ we are done.
  \item If $\tm \notin \NFv{\nature}$.
    Then by \cref{propCBV:characterization_NF} it must exist a term $\tm'$
    and a rule name $\rulename$ such that $\tm \tov{\rulename} \tm'$.
    By \nameref{lemCBV:subject_reduction} there exists $\mcount'$ such that
    $\mcount = \mset{\rulename} + \mcount'$, and 
    $\judgv[\mcount']{\emptyctx}{\emptyctx}{\tm'}{\emsetnu{\nature}}$.
    By \ih on $\mcount'$, there exists a sequence of steps
    $\tm' \tov{\rulename_1} \tm_1 \hdots \tov{\rulename_n} \tm_n$ where
    $\tm_n \in \NFv{\nature}$ and $\mcount' = \mset{\rulename_1, \hdots, \rulename_n}$.
    By joining this reduction sequence with the reduction step
    $\tm \tov{\rulename} \tm'$ we have a reduction sequence,
    $\tm \tov{\rulename} \tm' \tov{\rulename_1} \tm_1 \hdots \tov{\rulename_n} \tm_n$,
    and $\mcount = \mset{\rulename} + \mset{\rulename_1, \hdots, \rulename_n}$,
    so we are done.
  \end{itemize}

\noindent
($2 \Rightarrow 1$)
  By induction on $n$.
  \begin{enumerate}
  \item $n = 0$.
    Then $\tm \in \NFv{\nature}$, and $\mcount = \emset$.
    Therefore 
    $\judgv[\emset]{\emptyctx}{\emptyctx}{\tm}{\emsetnu{\nature}}$ holds
    by \cref{lemCBV:NF_tight_typable}.
  \item $n > 0$, assuming the property holds for $n-1$.
    Taking the reduction sequence of $(n-1)$ steps from $\tm_1$ to $\tm_n$,
    and $\mcount' = \mset{\rulename_2,\hdots,\rulename_n}$
    we can apply \ih, yielding 
    $\judgv[\mcount']{\emptyctx}{\emptyctx}{\tm_1}{\emsetnu{\nature}}$.
    Since $\tm \tov{\rulename_1} \tm_1$, then by \nameref{lemCBV:subject_expansion}
    we have $\judgv[\mcount]{\emptyctx}{\emptyctx}{\tm}{\emsetnu{\nature}}$, 
    with $\mcount = \mset{\rulename_1} + \mcount'$, as stated in the hypothesis,
    so we are done.
  \end{enumerate}
\end{proof}

\end{document}